\documentclass[letterpaper]{article} 
\usepackage{aaai24}  
\usepackage{times}  
\usepackage{helvet}  
\usepackage{courier}  
\usepackage[hyphens]{url}  
\usepackage{graphicx} 
\usepackage{natbib}  
\usepackage{caption} 
\usepackage{algorithm}
\usepackage{algorithmic}

\usepackage{newfloat}
\usepackage{listings}
\DeclareCaptionStyle{ruled}{labelfont=normalfont,labelsep=colon,strut=off} 
\floatstyle{ruled}
\newfloat{listing}{tb}{lst}{}
\floatname{listing}{Listing}
\usepackage[utf8]{inputenc}
\usepackage{amsmath}
\usepackage{xcolor}
\usepackage{amssymb}
\usepackage{mathrsfs} 
\usepackage{enumitem}
\usepackage{cleveref}
\newenvironment{proof}{\trivlist\item\emph{Proof. }}{\endtrivlist}

\newtheorem{definition}{Definition}

\newtheorem{lemma}{Lemma}

\newcommand{\op}{{\mathrm{OP}}}

\newcommand{\overalpha}{{\overline{\alpha}}}

\usepackage{thm-restate}
\newenvironment{restthm}[1]{\restatable{theorem}{#1}}{\endrestatable}
\newenvironment{restlem}[1]{\restatable{lemma}{#1}}{\endrestatable}

\title{Competition among Pairwise Lottery Contests}
\author{
    Xiaotie Deng\textsuperscript{\rm 1},
    Hangxin Gan\textsuperscript{\rm 2},
    Ningyuan Li\textsuperscript{\rm 1},
    Weian Li\textsuperscript{\rm 1},
    Qi Qi \textsuperscript{\rm 3}
}
\affiliations{
    \textsuperscript{\rm 1} Center on Frontiers of Computing Studies, School of Computer Science, Peking University, Beijing, China.\\
    \textsuperscript{\rm 2} School of Mathematical Science, Nankai University, Tianjin, China.\\
    \textsuperscript{\rm 3} Gaoling School of Artificial Intelligence, Renmin University of China, Beijing, China.\\
    \{xiaotie, liningyuan, weian\_li\}@pku.edu.cn,
    hangxin.gan@mail.nankai.edu.cn,
    qi.qi@ruc.edu.cn

}

\usepackage{bibentry}

\begin{document}

\maketitle

\begin{abstract}
We investigate a two-stage competitive model involving multiple contests. In this model, each contest designer chooses two participants from a pool of candidate contestants and determines the biases. Contestants strategically distribute their efforts across various contests within their budget. We first show the existence of a pure strategy Nash equilibrium (PNE) for the contestants, and propose a polynomial-time algorithm to compute an $\epsilon$-approximate PNE. In the scenario where designers simultaneously decide the participants and biases, the subgame perfect equilibrium (SPE) may not exist. Nonetheless, when designers' decisions are made in two substages, the existence of SPE is established. In the scenario where designers can hold multiple contests, we show that the SPE exists under mild conditions and can be computed efficiently.
\end{abstract}


\section{Introduction}
\label{sec:intro}

Contest theory is a commonly used and classic tool in the field of economics to define competition. In fact, many competitive scenarios can be perceived as contests. These may include political elections, sports events, promotional contests between firms aiming to increase their market share, and so forth. When designing a contest, the objective is to motivate the contestants to put forth greater effort in order to achieve specific goals. This involves determining the prize amount, the number of participants, and the winning rule.

Pairwise contests are a type of competition where the number of participants is limited to two. Classic examples of pairwise contests include the Colonel Blotto games \cite{B21}, which depict two players engaged in a battle where the outcome determines the victor. Such contests have numerous real-world applications. For instance, the US presidential election is a well-known example where two candidates compete over all states. Similarly, in competitive sports such as the NBA, two teams compete multiple times to determine the champion. The Internet price war \cite{LYDQCS19} provides another example, where two e-commerce platforms compete for regional markets by offering discount coupons.


The lottery contest is a form of imperfectly discriminatory competition, where the contestant who allocates more effort has a higher probability of winning than one who allocates lesser effort. In real-world scenarios, the lottery contest is highly applicable due to the stochastic factors that may impact the outcome. More specifically, despite allocating greater effort towards a given issue, winning is not always certain due to the unpredictability of such factors.


Current research on pairwise and lottery contests tends to center around studying the equilibrium behaviors of contestants or optimizing lottery functions to achieve certain objectives. However, little attention has been paid to investigating competing designers. According to a recent survey on contest theory \cite{S20}, exploring the economy of competitions among designers poses several challenges, particularly in analyzing their equilibrium behavior. In a single contest or a fixed number of contests, the focus is primarily on the strategic behavior of contestants. However, designers also have strategic behavior that needs to be taken into account, including how contestants allocate their efforts and how designers compete with one another.


In this paper, we concentrate on the pairwise lottery contests (PLC), where two contestants compete for a prize with a winning probability determined by the lottery rule that is based on their (weighted) effort. 
Designers are allowed to hold one or several PLCs. 
Each designer's goal is to maximize the total exerted effort of the participants in all her held contests. Each contestant pursues maximizing the expected prize from the contests she joins. There is a two-stage game in our model: one is among contest designers, who decide the number of held contests and design the configuration (including prize, participants and biases) of held contests. The other is among contestants who decide how to allocate effort. 

\subsection{Our Contributions}
Our model introduces several innovative features that enrich the current discourse in contest theory. Firstly, much of the existing literature predominantly concentrates on single contest design or contestants' equilibrium analysis within prescribed multi-contest frameworks. In contrast, we cast sight into the case of multiple contests held by different strategic designers. Therefore, the designers' strategic behaviors are integrally addressed and analyzed. Secondly, traditional models, exemplified by the Colonel Blotto games, typically focus on pairwise contests involving just two contestants and mainly study equilibrium behaviors of these two contestants. We expand this framework, allowing for $n$ potential participants, thereby granting strategic designers the latitude to select any two from this candidate pool. While this expansion offers a more richer and realistic representation, it increases the analytical difficulty of the model.

Our contributions and results can be summarized as follows:
\begin{itemize}
    \item Given the configurations of all contests, for the game among contestants (the second stage in our model), we propose a concept called equilibrium multiplier vector (EMV) which represents marginal utilities of contestants in equilibrium, as our main analytical tool characterize the contestant equilibrium. We prove the existence of EMV utilizing Brouwer's fixed-point theorem, and show the uniqueness of EMV leveraging a monotone property. By establishing the connection between the EMV and equilibrium strategy of contestants, we fully characterize the contestant equilibria. Furthermore, we design a polynomial time algorithm to compute an $\epsilon$-approximate contestant equilibrium.

    \item For the game of designers (the first stage in our model), when each designer is allowed to hold one contest only, we first show the non-existence of subgame perfect equilibrium (SPE) in certain cases, due to the complicated deviation of the two-dimension strategy (choosing participants and biases simultaneously). However, if designers choose participants and set biases in two separated substages, we can always find an SPE. Specifically, when the participants are fixed, considering the designers' strategies to set the biases, we show that it forms an equilibrium when all designers set balancing biases such that every participants in each contest have the same winning probability (i.e., 1/2). Under this situation, the equilibrium effort exerted by a contestant into each contest is proportional to the contest prize. We observe that the designers’ participants selection is actually equivalent to a variant of weighted congestion game, where a pure Nash equilibrium always exists, implying the existence of a sequential equilibrium in our model.
    \item When each designer can divide her budget to hold several contests, although the strategy space of designers seems to become more complicated, surprisingly, we show that an SPE always exists even if the participants and biases are decided simultaneously, under a very mild condition that the maximum total effort of an individual contestant does not exceed the total effort of all other contestants. In this SPE, each contestant's or designer's utility will be proportional to her total effort  or prize budget, respectively.
\end{itemize}
Due to space limitations, all missing proofs appear in the appendix.
\subsection{Related Works}
Our paper contributes to the literature in the field of economics and computer science, particularly in topics of multiple contests competition and pairwise contest design. Our
work is closely related to the following several studies.
\cite{LZ22} focus on the analysis of pure strategy Nash
equilibrium on 2-contestant lottery Colonel Blotto games.
However, our paper extends the total number of contestants from 2 to $n$, which leads to that each contestant may have 
different opponents in different pairwise contests. 
In addition, \cite{FW18} study designing the optimal lottery contest by setting biases in the setting of single contest to achieve different objectives. \cite{WWX23} consider a setting of multi-battle contests where the same two contestants battle with each other in every contest and every designer sets biases to attract more effort. Our paper can be viewed as a generalized model of these two papers,
where the designer of each contest picks up two contestants from $n$ candidates and sets the biases.


Our research focuses on two aspects: equilibrium analysis and contest design. We summarize related works in three fields: lottery contests, Colonel Blotto games, and competition among contests.

\subsubsection{Lottery Contests} 
The lottery-form contest is introduced by \cite{S96} and \cite{CR98}, where contestants' winning probability is determined by a contest success function (CSF). 
\cite{DN98} consider the optimal CSF with $n$ symmetric contestants. 
\cite{N04} studies the optimal CSF in two-contestant symmetric contest. When employing a certain form of CSF, lottery contest is classified as a specific type of Tullock contest \cite{N99,T01,S02}. 
\cite{CR98} examine the contest performance affected by the different parameters of Tullock CSF. Multiple equilibria in Tullock contest is studied in \cite{CS11}. When there are only two contestants, the optimal contests obtained by optimizing the parameters of Tullock CSF are investigated \cite{W10,EMN11}. 
\cite{FKLS13} provide the optimal biases for an $n$-player Tullock contest. Besides, the lottery contests with multi-prize is discussed in \cite{FWZ22}. Additionally, the best response dynamics of contestants is investigated by \cite{E17,GG23}.

\subsubsection{Colonel Blotto Games} Colonel Blotto Games \cite{B21} characterize the competition between two players across several contests (aka., battle-fields), which has some similarity to the game among contestants in our model. Many classic papers in this topic mainly focus on the deterministic CSFs, e.g., \cite{GW50,R06,MM15,KR21}.  
\cite{F58} first introduces lottery CSFs into a two-contestant symmetric Blotto game and shows the uniqueness of equilibrium. 
\cite{DM15} generalize the results to the case with more than two contestants. Some works \cite{R05,XZ18} study the two-contestant Blotto game under the general Tullock CSFs. 

\subsubsection{Competition among Contests} The topic of 
competition among contests has received increasing attention in the recent decade. Initially, 
\cite{AM09} examine the two identical Tullock contests setting and investigate the prize structure in different goals. 
\cite{DV09} study multiple auction-based crowdsourcing contests and give the contestants' equilibrium in symmetric and asymmetric settings. Later, many sutdies \cite{B16,AM18,JSY20} focus on comparing the performance of two parallel contests with different types. 
Recently, 
\cite{DGLLL21} investigate that the optimal CSF in the monopolistic setting is also the equilibrium strategy in the competitive setting when designers aim to maximize the total effort.
\cite{KK22} show that when a contestant can join several contests but the output in each contest is affected by an uncertainty variable, increasing the number of contests one contestant participates in  improves the utility of contest organizer. 
\cite{DLLQ22} focus on the environment of parallel contest. They analyze the equilibrium of contestants' participation and design the prize policies of contests in different settings.

\section{Model and Preliminaries}
\label{sec:pre}


There are $n$ contestants and $m$ designers. We use the notation $i\in[n]$ and $j\in[m]$ to denote a contestant and a designer, respectively. We assume that each contestant $i$ has a limited total effort $T_i\in\mathbb{R}_{> 0}$ to exert in contests and each designer $j$ has a limited budget $B_j\in\mathbb{R}_{> 0}$.

In this work, we focus on pairwise general lottery contests, in which the designer invites two contestants as the participants, and sets a multiplicative bias for each participant to incentivize their effort. Each participant's winning probability depends on the product of her bias and her effort exerted into this contest.

Formally, a pairwise general lottery contest $C$ is defined as a tuple $C=(S_C,R_C,\alpha_C)$: $S_C$ denotes two contestants selected as participants of contest $C$, satisfying that $S_C\subseteq [n]$ and $|S_C|=2$; $R_C\in \mathbb{R}_{>0}$ denotes the prize prepared for the winner in the contest; and $\alpha_C=(\alpha_{C,i})_{i\in S_C}$, where $\alpha_{C,i}\in\mathbb{R}_{>0}$ denotes the bias selected for the participant $i$. 

Suppose $S_C=\{i_1,i_2\}$, and let $x_{i_1,C}$ and $x_{i_2,C}$ be the effort that these two contestants exert in contest $C$. Each contestant's effort is multiplied by her bias to get $\alpha_{C,i_1}\cdot x_{i_1,C}$ and $\alpha_{C,i_2}\cdot x_{i_2,C}$. The winning probabilities of contestants $i_1$ and $i_2$ are $f(\alpha_{C,i_1} \cdot x_{i_1,C};\alpha_{C,i_2} \cdot x_{i_2,C})$ and $f(\alpha_{C,i_2} \cdot x_{i_2,C};\alpha_{C,i_1} \cdot x_{i_1,C})$ respectively, where $f$ is the lottery CSF defined as follows:
\begin{align*}
f(x;y)=\begin{cases}
    \frac{x}{x+y}, &\text{if } x>0 \vee y>0,\\
    \frac{1}{2}, & \text{if } x=y=0.
    \end{cases}
\end{align*}
Note that $f(x;y)+f(y;x)=1$.

We study two models of the designers, varying in whether a designer can divide her budget to hold multiple contests.
\begin{enumerate}
    \item In the \textbf{divisible prize} model (DPM), each designer $j$ is allowed to distribute her prize budget $B_j$ to hold an arbitrary number of contests, denoted by $\mathcal{C}_j=\{C_{j,1},\cdots,C_{j,K_j}\}$, satisfying that $\sum_{C\in \mathcal{C}_j}R_{C}\leq B_j$.
    \item In the \textbf{indivisible prize} model (IPM), each designer $j$ can hold only one pairwise general lottery contest, denoted by $C_j$, and $R_{C_j}\leq B_j$. In this case we define $\mathcal{C}_j=\{C_j\}$.
\end{enumerate}

In both models, each designer $j$ can arbitrarily design the configuration of every contest $C\in\mathcal{C}_j$, i.e., the invited participants $S_C$, the reward $R_C$, and the bias $\alpha_C$, within her budget.
We call $\mathcal{C}_j$ the strategy of designer $j$  and define $\vec{\mathcal{C}}=(\mathcal{C}_1,\cdots,\mathcal{C}_m)$ as the strategy profile of designers. Sometimes we use the notation $C\in\vec{\mathcal{C}}$ to denote that $C\in\cup_{j\in[m]}\mathcal{C}_j$.

Given designers' strategy profile $\vec{\mathcal{C}}$, for any contestant $i$, let $\mathcal{A}(i,\vec{\mathcal{C}})=\{C\in\cup_{j\in[m]}\mathcal{C}_j:i\in S_C\}$ be the set of contests that $i$ is invited to participate in. Each contestant $i$ decides non-negative amounts of effort to exert in those contests inviting her, denoted by $x_i=(x_{i,C})_{C\in \mathcal{A}(i,\vec{\mathcal{C}})}$, which satisfies $\sum_{C\in\mathcal{A}(i,\vec{\mathcal{C}})}x_{i,C}\leq T_i$. We call $x_i$ the strategy of contestant $i$, and $\vec{x}=(x_1,\cdots,x_n)$ is called the strategy profile of contestants. Sometimes we use $\vec{\mathcal{C}}_{-j}$ and $\vec{x}_{-i}$ to denote the strategy profile of all designers except designer $j$ and the strategy profile of all contestants except contestant $i$, respectively. 

Given $\vec{\mathcal{C}}$ and $\vec{x}$, for any contestant $i$ and any contest $C\in \mathcal{A}(i,\vec{\mathcal{C}})$, let $\op_{i,C}$ denote her opponent in contest $C$, that is, $S_C=\{i,\op_{i,C}\}$. Then, her winning probability in $C$ is denoted by
\begin{align*}
    p_{i,C}(\vec{x})=f(\alpha_{C,i}\cdot x_{i,C};\alpha_{C,\op_{i,C}}\cdot x_{\op_{i,C},C}).
\end{align*}

The utility of contestant $i$ is defined as her expected total prize, $u_i^{Contestant}(\vec{\mathcal{C}},\vec{x})=\sum_{C\in\mathcal{A}(i,\vec{\mathcal{C}})} R_C\cdot p_{i,C}(\vec{x}).$
And the utility of a designer is the total effort exerted by the participants in her all contests, $u_j^{Designer}(\vec{\mathcal{C}},\vec{x})=\sum_{C\in\mathcal{C}_j}\sum_{i\in S_C}x_{i,C}$.

With these definitions, we study a two-stage game model of the competition among pairwise lottery contests.

\begin{definition}
An instance of Pairwise Lottery Contest Competition Game (PLCCG) is defined as the tuple $(n,m,(T_i)_{i\in[n]},(B_j)_{j\in[m]})$. The game has two stages:
\begin{enumerate}
\item In the first stage (called the stage of designers), all designers simultaneously select their strategies. In other words, each designer $j\in [m]$ decides the number $K_j=|\mathcal{C}_j|$ (under the indivisible prize model, $K_j$ always equals to $1$.) of contests to hold, and the configuration of each contest $C\in \mathcal{C}_j$, within her total budget $B_j$. 
\item In the second stage (called the stage of contestants), having observed $\mathcal{C}_1,\cdots,\mathcal{C}_m$, all contestants simultaneously select their strategies, i.e., each contestant $i\in[n]$ decides her effort $x_i=(x_{i,C})_{C\in \mathcal{A}(i,\mathcal{C})}$, within her total effort $T_i$.
\end{enumerate}
\end{definition}

Our work mainly focuses on the sequential equilibrium, i.e., subgame perfect equilibrium (SPE), of PLCCG. Before giving the definition of SPE, we first define the contestant equilibrium, i.e., the pure Nash equilibrium among contestants in the second stage, when a strategy profile $\vec{\mathcal{C}}$ of designers is given.
\begin{definition}\label{def:contestant equilibrium}
Given designers' strategy profile $\vec{\mathcal{C}}$, we say a contestant strategy profile $\vec{x}$ is a contestant equilibrium under $\vec{\mathcal{C}}$ , if for any $i\in[n]$ and any feasible strategy $x_i'$, it holds that
$$u_i^{Contestant}(\vec{\mathcal{C}},\vec{x})\geq u_i^{Contestant}(\vec{\mathcal{C}},(x_i',\vec{x}_{-i})).$$
Define $\mathcal{E}_{\vec{\mathcal{C}}}$ as the set of all contestant equilibria under $\vec{\mathcal{C}}$.
\end{definition}
Next, we define the subgame perfect equilibrium and designer equilibrium.
\begin{definition}\label{def:SPE}
$(\vec{\mathcal{C}},\vec{x})$ is a subgame perfect equilibrium, if the following two conditions hold:
\begin{enumerate}
    \item $\vec{x}$ is a contestant equilibrium under $\vec{\mathcal{C}}$, i.e., $\vec{x}\in \mathcal{E}_{\vec{\mathcal{C}}}$.
    \item For any designer $j$, any feasible strategy $\mathcal{C}_j'$ and any $\vec{x}'\in\mathcal{E}_{(\mathcal{C}_j',\vec{\mathcal{C}}_{-j})}$, it holds that \footnote{Note that this definition is slightly stronger than the standard definition of SPE since it requires that for any designer $j$, $\vec{x}$ is better for the best $\vec{x}'\in\mathcal{E}_{(\mathcal{C}_j',\vec{\mathcal{C}}_{-j})}$, while the standard definition only requires that $\vec{x}$ is better for some $\vec{x}'\in\mathcal{E}_{(\mathcal{C}_j',\vec{\mathcal{C}}_{-j})}$. However, this is not an essential difference since the contestant equilibrium will be unique in some sense as shown later.}
    $$u_j^{designer}(\vec{\mathcal{C}},\vec{x})\geq u_j^{designer}((\mathcal{C}_j',\vec{\mathcal{C}}_{-j}),\vec{x}').$$
\end{enumerate}
We say $\vec{\mathcal{C}}$ is a designer equilibrium if there is some $\vec{x}\in \mathcal{E}_{\vec{\mathcal{C}}}$ such that $(\vec{\mathcal{C}},\vec{x})$ is a subgame perfect equilibrium.
\end{definition}

\section{Contestant Equilibrium}
\label{sec:contestant equilibrium}
In this section, we study the equilibrium behavior of contestants in the contestants' stage of PLCCG, i.e., the contestant equilibrium, when the designers' strategy profile 
is given.
In subsection \ref{subsec:EMV}, as a key tool for analyzing and characterizing contestant equilibrium, we propose a concept called \emph{equilibrium multiplier vector} (EMV), which represents each contestant's equilibrium strategy by a multiplier variable, indicating the contestant's marginal utility under the contestant equilibrium. We also show the close connection between contestant equilibrium and EMV. This simplifies the contestant's multi-dimensional strategy into a single-dimensional number.
In subsection \ref{subsec:ExistUniqueEMV}, we prove the existence and uniqueness of equilibrium multiplier vector, which enables us to fully characterize the set of all contestant equilibria.
Additionally, in subsection \ref{subsec:Algorithm}, we show that an $\epsilon$-approximate contestant equilibrium can be found in polynomial time through an iterative updating process of the multiplier vector, which draws inspiration from the t\^atonnement algorithm used in the field of market equilibrium.

Given any designers' strategy profile $\vec{\mathcal{C}}$, since the contestants do not care about the holder of each contest, we can simplify some notions. We use the notation $\mathcal{C}=\cup_{j\in[m]}\mathcal{C}_j$ to denote the set of all contests, and define $\mathcal{A}(i,\mathcal{C})=\{C\in\mathcal{C}:i\in S_C\}$ and $u_i(\mathcal{C},\vec{x})=\sum_{C\in\mathcal{A}(i,\mathcal{C})}R_C\cdot p_{i,C}(\vec{x})$. W.l.o.g, we assume that for any contestant $i\in[n]$, $\mathcal{A}(i,\mathcal{C})\neq \emptyset$.

\subsection{Equilibrium Multiplier Vector}
In this subsection, we propose equilibrium multiplier vector as a representation of contestant equilibrium. We first give the motivation and definition of EMV by \Cref{lemma:lagrange-dual} and \Cref{def:contestant-dual-equilibrium}. Then we characterize the contestant equilibrium with the help of EMV. We derive a necessary and sufficient condition for a vector being an EMV in \Cref{lemma:dual-solution-condition}, and then characterize the set of all contestant equilibria corresponding to an EMV as shown in \Cref{lemma:characterize-contestant-equilibrium-by-multiplier}. Combining \Cref{lemma:characterize-contestant-equilibrium-by-multiplier} and the uniqueness of EMV proved in the next subsection, we can fully characterize the set of all contestant equilibrium.

\label{subsec:EMV}
If $\vec{x}$ is a contestant equilibrium, for each contestant $i$, $x_i$ is a best response to $\vec{x}_{-i}$. In other words, $x_i$ is an optimal solution to the following optimization problem:
\begin{align}
\max_{x_{i,C}\geq0 \text{ for }C\in\mathcal{A}(i,\mathcal{C})} \quad & \sum_{C\in\mathcal{A}(i,\mathcal{C})}R_C\cdot p_{i,C}(x_i,\vec{x}_{-i}),\label{eq:contestant-bestresponse-optimization-problem}\\
\mbox{s.t.} \quad & \sum_{C\in \mathcal{A}(i,\mathcal{C})}x_{i,C}\leq T_i.\nonumber
\end{align}
Intuitively, if we use the Lagrange multiplier method, there will be a Lagrange multiplier $\lambda_i\geq0$ so that $x_i$ maximizes the Lagrangian function $\sum_{C\in\mathcal{A}(i,\mathcal{C})}R_C\cdot p_{i,C}(x_i,\vec{x}_{-i})-\lambda_i(T_i-\sum_{C\in \mathcal{A}(i,\mathcal{C})}x_{i,C})$. However, due to the discontinuity of $p_{i,C}(x_i,\vec{x}_{-i})$ at the point with $x_{i,C}=x_{\op_{i,C},C}=0$, the Lagrange multiplier method cannot be applied directly. Thus, we establish the existence of such $\lambda_i$ for each contestant $i$ through some analysis, to obtain the following lemma.

\begin{restlem}{CEEquilibriumMultiplierVector}
\label{lemma:lagrange-dual}
If $\vec{x}$ is a contestant equilibrium under strategy profile $\mathcal{C}$, there exist $\lambda_1,\cdots,\lambda_n\in \mathbb{R}_{\geq 0}$ such that, for any contestant $i$ and any contest $C\in\mathcal{A}(i,\mathcal{C})$, $R_C\cdot\frac{\partial p_{i,C}(\vec{x})}{\partial x_{i,C}}\leq\lambda_i$, where the equation holds when $x_{i,C}>0$.
\end{restlem}

By \Cref{lemma:lagrange-dual}, we know that every contestant equilibrium $\vec{x}$ corresponds to a vector $\vec{\lambda}=(\lambda_1,\cdots,\lambda_n)$, which can be viewed as the vector of contestants' Lagrange multipliers in Optimization \ref{eq:contestant-bestresponse-optimization-problem}. We refer to such $\vec{\lambda}$ as an EMV.
\begin{definition}\label{def:contestant-dual-equilibrium}
A vector $\vec{\lambda}=(\lambda_1,\cdots,\lambda_n)\in R_{\geq 0}^n$ is an equilibrium multiplier vector, if there exists a contestant equilibrium $\vec{x}$ such that $\vec{x}$ and $\vec{\lambda}$ satisfies the conditions in \Cref{lemma:lagrange-dual}.
We call $\vec{x}$ a contestant equilibrium corresponding to $\vec{\lambda}$.
\end{definition}


First, we present a necessary and
sufficient condition to decide whether a vector $\vec{\lambda}$ is an EMV. We say a vector $\vec{\lambda}\in\mathbb{R}_{\geq 0}^n$ is \emph{valid} if for any contest $C\in\mathcal{C}$, it holds $\sum_{i\in S_C}\lambda_i>0$. Then, for any valid vector $\vec{\lambda}\in\mathbb{R}_{\geq 0}^n$, we define  $$\hat{x}_{i,C}(\vec{\lambda})=R_C \cdot \frac{
\alpha_{C,i}\alpha_{C,\op_{i,C}}\lambda_{\op_{i,C}}
}{
(\alpha_{C,\op_{i,C}}\lambda_i+\alpha_{C,i}\lambda_{\op_{i,C}})^2}$$
for any contestant $i$ and any contest $C\in\mathcal{A}(i,\mathcal{C})$.
For any contestant $i$, we also define $\hat{T}_i(\vec{\lambda})=\sum_{C\in\mathcal{A}(i,\mathcal{C})}\hat{x}_{i,C}(\vec{\lambda})$, which can be viewed as the demand of contestant $i$'s effort induced by $\vec{\lambda}$. Before giving the characterization of  EMV, we first give a lemma to show that $\hat{x}_{i,C}(\vec{\lambda})$ is the lowest exerted effort in a contestant equilibrium corresponding to $\vec{\lambda}$.

\begin{restlem}{CEhatxfromEMV}
\label{lemma:reconstruct-effort-from-EMV}
If $\vec{x}$ is a contestant equilibrium corresponding to an equilibrium multiplier vector $\vec{\lambda}$, for any contestant $i\in[n]$ and any contest $C\in\mathcal{A}(i,\mathcal{C})$, it holds that $x_{i,C}\geq \hat{x}_{i,C}(\vec{\lambda})$, where the equation holds when $\lambda_i>0$. 
\end{restlem}

Now, we can give a necessary and sufficient condition for a vector $\vec{\lambda}$ to be an EMV, which enables us to identify an EMV directly.
\begin{restthm}{CEconditionforEMV}
\label{lemma:dual-solution-condition}
For any $\vec{\lambda}\in\mathbb{R}_{\geq 0}^n$, $\vec{\lambda}$ is an equilibrium multiplier vector if and only if the following statements hold:
\begin{enumerate}
    \item $\vec{\lambda}$ is valid;
    \item For any contest $i$ with $\lambda_i>0$, $T_i=\hat{T}_i(\vec{\lambda})$;
    \item For any contest $i$ with $\lambda_i=0$, $T_i\geq\hat{T}_i(\vec{\lambda})$.
\end{enumerate}
\end{restthm}

Next, we show that, when given an EMV $\vec{\lambda}$, the set of all contestant equilibria corresponding to $\vec{\lambda}$ is also uniquely determined.
\begin{restthm}{CEcharacterizeAllCEbyEMV}
\label{lemma:characterize-contestant-equilibrium-by-multiplier}
If $\vec{\lambda}$ is an equilibrium multiplier vector, then a contestant strategy profile $\vec{x}$ is a contestant equilibrium corresponding to $\vec{\lambda}$ if and only if $\vec{x}\in\mathcal{X}(\vec{\lambda})$, 
where
\begin{align*}
    \mathcal{X}(\vec{\lambda})=\{&(x_{i,C})_{i\in[n],C\in\mathcal{A}(i,\mathcal{C})}:\\
    &\forall i\in[n],\sum_{C\in\mathcal{A}(i,\mathcal{C})}x_{i,C}\leq T_i\land\\
    &\forall C\in\mathcal{A}(i,\mathcal{C}),x_{i,C}\geq\hat{x}_{i,C}(\vec{\lambda})\}.
\end{align*}
\end{restthm}

It is notable that in the next subsection, we will prove that for any $\mathcal{C}$, there always exists a unique EMV $\vec{\lambda}$. Combined with this, \Cref{lemma:characterize-contestant-equilibrium-by-multiplier} fully characterizes the set of all contestant equilibria, which is exactly $\mathcal{X}(\vec{\lambda})$.

\subsection{Existence and Uniqueness}\label{subsec:ExistUniqueEMV}
In this subsection, we mainly discuss the existence and uniqueness of EMV. We prove that EMV always exists (\Cref{thm:contestant-equilibrium-existence}) and is unique (\Cref{thm:unique-contestant-equilibrium-multiplier}) for any strategy profile of designers. Although the existence of contestant equilibrium follows immediately, there may exist multiple contestant equilibria. Nonetheless, as mentioned before, the set of all contestant equilibria is fully characterized by the unique EMV through \Cref{lemma:characterize-contestant-equilibrium-by-multiplier}.


A conventional approach to prove the existence of a contestant equilibrium is to consider the best response updating process of the strategy profile $\vec{x}$ and show the existence of a fixed point by Kakutani fixed-point theorem \cite{K41}. However, due to the discontinuity of the lottery CSF $f(x;y)$ at the point $x=y=0$, the set of contestant $i$'s best response is sometimes empty and 
the condition of Kakutani fixed-point theorem is not satisfied.
To address this problem, we turn to the space of multiplier vectors. We carefully design a continuous mapping of the multiplier vector $\vec{\lambda}$ such that the fixed point is an EMV, and prove the existence of such a fixed point by Brouwer's fixed-point theorem.
\begin{restthm}{CEthmExistenceofEMV}
\label{thm:contestant-equilibrium-existence}
For any designers' strategy profile $\vec{\mathcal{C}}$, there exists an equilibrium multiplier vector $\vec{\lambda}$.
\end{restthm}

Next, we prove the uniqueness of EMV. Recall that, for any valid $\vec{\lambda}$, $\hat{T}_i(\vec{\lambda})$ can be viewed as the demand of contestant $i$'s effort induced by $\vec{\lambda}$, and the conditions in \Cref{lemma:dual-solution-condition} can be interpreted as a complementary-slackness condition for the demands $\hat{T}_1(\vec{\lambda}),\cdots,\hat{T}_n(\vec{\lambda})$. We view these demands as a vector function $\hat{T}(\vec{\lambda})=(\hat{T}_1(\vec{\lambda}),\cdots,\hat{T}_n(\vec{\lambda}))$. An important observation is that, $\hat{T}(\vec{\lambda})$ satisfies a monotone property in $\vec{\lambda}$.

\begin{restlem}{CEMonotoneProperty}
\label{lemma:monotone}
For any two valid multiplier vectors $\vec{\lambda}$ and $\vec{\lambda}'$, it holds that
\begin{align*}
    \sum_{i=1}^n(\lambda_i'-\lambda_i)(\hat{T}_i(\vec{\lambda}')-\hat{T}_i(\vec{\lambda}))\leq0.
\end{align*}
Moreover, the strict inequality holds when there exists some $i$ such that $\lambda_i'\neq\lambda_i$ and $\max_{C\in\mathcal{A}(i,\mathcal{C})}$ $\max\{\lambda_{\op_{i,C}},\lambda_{\op_{i,C}}'\}>0$.
\end{restlem}

With this monotone property, we can prove that the EMV is unique
\footnote{We remark that the uniqueness of EMV relies on the assumption that for any contestant $i$, $\mathcal{A}(i,\mathcal{C})\neq \emptyset$. When there is some contestant $i$ with $\mathcal{A}(i,\mathcal{C})= \emptyset$, it means that contestant  $i$ does not participate in any contest, and we can assume that $\lambda_i$ can take arbitrary value. In this case, however, for any other contestant with $\mathcal{A}(i,\mathcal{C})\neq\emptyset$, the equilibrium multiplier is still unique.}. Intuitively, if there are two distinct EMVs, $\vec{\lambda}$ and $\vec{\lambda}'$, by \Cref{lemma:monotone} they will induce different demand of efforts, i.e., $\hat{T}(\vec{\lambda})\neq \hat{T}(\vec{\lambda}')$, which will contradict with \Cref{lemma:dual-solution-condition}.
\begin{restthm}{CEthmUniquenessofEMV}
\label{thm:unique-contestant-equilibrium-multiplier}
Given any designers' strategy profile $\vec{\mathcal{C}}$, there is a unique equilibrium multiplier vector.
\end{restthm}

\subsection{Computation of Contestant Equilibrium}\label{subsec:Algorithm}
In this subsection, we study the computation of the contestant equilibrium. We design an algorithm which computes an $\epsilon$-contestant equilibrium in polynomial time given any strategy profile of designers $\vec{\mathcal{C}}$. \Cref{lemma:monotone} provides the insight that, we can roughly adjust $\hat{T}(\vec{\lambda})$ towards some direction by adjusting $\vec{\lambda}$ in the the opposing direction. Building upon this, we firstly find an approximate EMV through an iterative updating process inspired by the t\^atonnement algorithm, and then construct an approximate contestant equilibrium based on this approximate EMV. 
\begin{definition}\label{def:approx ce}
A strategy profile $\vec{x}$ is an $\epsilon$-approximate contestant equilibrium, if for any $i$ and any feasible strategy $x_i'$, $u_i^{Contestant}(\vec{\mathcal{C}},\vec{x})\geq (1-\epsilon)u_i^{Contestant}(\vec{\mathcal{C}},(x_i',\vec{x}_{-i}))$ holds.
\end{definition}
\begin{restthm}{CEthmAlgorithm}
\label{thm:contestant-equilibrium-algorithm}
Given any strategy profile $\vec{\mathcal{C}}$, for any $\epsilon>0$, there exists an algorithm to compute an $\epsilon$-approximate contestant equilibrium in polynomial time in $\frac1{\epsilon}$ and the input sizes, namely $n$, $m$, and $|\cup_{j\in[m]}\mathcal{C}_j|$.
\end{restthm}
\section{Indivisible Prize Model}
\label{sec:indivisible}
Starting from this section, we investigate the equilibrium behavior of designers. We study the indivisible prize model (IPM) in this section and the divisible prize model (DPM) in Section \ref{sec:divisible}. 

In this section, we first show that the designer equilibrium (defined in Definition \ref{def:SPE}) may not exist in some instances of IPM. Thus, we consider a weaker concept called weak designer equilibrium (WDE), based on a setting where the stage of designers is divided into two substages. 
By analyzing the equilibrium of two substages in reverse order, we prove that WDE always exists, in which all designers will adopt \emph{balancing} biases such that both sides of any contest have an equal winning probability of $1/2$ under a contestant equilibrium.


\subsection{Weak Designer Equilibrium}
We first use a counterexample to show that the SPE may not exist under IPM, even in a very simple instance with $3$ identical contestants and $2$ identical designers.
\begin{restthm}{INDSPEnotexist}
\label{thm:ind-SPE-not-exist}
In some instances of indivisible prize model, the designer equilibrium does not exist.
\end{restthm}


Roughly speaking, the main reason of the nonexistence of SPE is that modifying the choice of participants in some contest can cause significant change in the optimal choice of biases, which again leads to another better choice of participants. Therefore, we relax the requirement of designer equilibrium by separating the stage of designers into two substages\footnote{This setting is justified by the common fact that the list of participants is often announced before the contest beginning, and modifying the judging criteria for contestants' performance is relatively less costly than withdrawing the invitation to participants.}: in the first substage, each designer decides the amount of prize and participants of her contest; and in the second stage, each designer decides the biases of her contest.  

Now we provide the definition of WDE formally. For each designer $j$, we call $(R_{C_j},S_{C_j})$ her first-stage strategy, and $(\alpha_{C_j,i})_{i\in S_{C_j}}$ her second-stage strategy. Let $\mathrm{BiasDev}(C_j)=\{C_j':R_{C_j'}=R_{C_j}\land S_{C_j'}=S_{C_j}\}$ denote all strategies of designer $j$ whose first-stage strategy is the same as that of $C_j$. The WDE can be defined as follows.
\begin{definition}
In the IPM, we say a strategy profile $\vec{\mathcal{C}}$ is a second-substage equilibrium, if there exists $\vec{x}\in\mathcal{E}_{\vec{\mathcal{C}}}$ such that, for any designer $j$, any $C_j'\in\mathrm{BiasDev}(C_j)$ and for any $\vec{x}'\in\mathcal{E}_{(\mathcal{C}'_j,\vec{\mathcal{C}}_{-j})}$, it holds that $u^{designer}_j(\vec{\mathcal{C}};\vec{x})\geq u^{designer}_j((\mathcal{C}'_j,\vec{\mathcal{C}}_{-j});\vec{x}')$.

We say a strategy profile $\vec{\mathcal{C}}$ is a first-substage equilibrium, if the following holds:
\begin{enumerate}
\item $\vec{\mathcal{C}}$ is a second-substage equilibrium.
\item There exists $\vec{x}\in\mathcal{E}_{\vec{\mathcal{C}}}$ such that, for any designer $j$ and any strategy $\mathcal{C}_{j}'$, there is $\vec{\mathcal{C}}_{-j}'$ such that
\begin{itemize}
\item $C'_{j'}\in\mathrm{BiasDev}(C_{j'})$ for any $j'\neq j$,
\item $\vec{\mathcal{C}}'=(\mathcal{C}_j',\vec{\mathcal{C}}_{-j}')$ is a second-substage equilibrium,
\item $u^{designer}_j(\vec{\mathcal{C}};\vec{x})\geq u^{designer}_j(\vec{\mathcal{C}}';\vec{x}')$ for all $\vec{x}'\in\mathcal{E}_{\vec{\mathcal{C}}'}$.
\end{itemize}
\end{enumerate}
A strategy profile $\vec{\mathcal{C}}$ is called a weak designer equilibrium if it is a first-substage equilibrium.
\end{definition}

It is not hard to find that WDE is a weaker concept than designer equilibrium, since any beneficial deviation in either substage leads to a beneficial deviation in the original designer stage.

\subsection{Equilibrium in the Second Substage}
To analyze the weak designer equilibrium, we firstly study the second-substage equilibrium, i.e., how the designers set the biases when their first-stage strategies are fixed.

We extend an approach from the previous works to our model, which considers the winning probability of a participant under contestant equilibrium as designer's decision variable, instead of directly deciding the biases in the contest. We establish the validity of this approach in our model by \Cref{lemma:designer-control-probability-as-bias}. Although existing literature suggests that a designer's dominate strategy is to set a balancing bias which results in her participants having an equal winning probability of $1/2$, we show that this claim does not unconditionally hold in our model in \Cref{thm:example-balancing-is-not-best}. Nonetheless,  in \Cref{thm:indiv-1/2probability-is-equilibrium} we prove that it still forms an second-substage equilibrium when all designers are using the balancing biases.

Firstly we show that the winning probability is uniquely determined by the designers' strategy profile.
\begin{restlem}{INDhatplambda}
\label{lemma:winning-probability-from-lambda}
Given the strategy profile $\vec{\mathcal{C}}$, let $\vec{\lambda}$ be the unique equilibrium multiplier vector with respect to $\vec{\mathcal{C}}$. For any contest $C$ and any contestant $i\in S_C$, define $\hat{p}_{i,C}(\vec{\lambda})=\frac{\alpha_{C,i}\lambda_{\op_{i,C}}}{\alpha_{C,i}\lambda_{\op_{i,C}}+\alpha_{C,\op_{i,C}}\lambda_{i}}$.
Then, for any contestant equilibrium $\vec{x}$, it holds that $p_{i,C}(\vec{x})=\hat{p}_{i,C}(\vec{\lambda})$.
\end{restlem}


The following technical lemma shows that the designers are able to manipulate the equilibrium winning probabilities in their contests by adjusting the biases. This allows us to consider the winning probability in a contest as the designer's decision variable in the second substage.
\begin{restlem}{INDControlProbabilityAsBias}
\label{lemma:designer-control-probability-as-bias}
Suppose the set of all contests is partitioned as $\mathcal{C}=\mathcal{C}^{fix}\cup\mathcal{C}^{var}$, such that every $C\in\mathcal{C}^{fix}$'s configuration is fixed, while every $C\in\mathcal{C}^{var}$ only has fixed $S_C$ and $R_C$, and the biases $\alpha_C$ need to be assigned. Given any target of winning probabilities for these contests $(\tilde{p}_{i,C})_{C\in\mathcal{C}^{var},i\in S_C}$ satisfying that $\tilde{p}_{i,C}\in(0,1)$ and $\sum_{i\in S_C}\tilde{p}_{i,C}=1$, there exists an assignment of biases $(\alpha_{C,i})_{C\in\mathcal{C}^{var},i\in S_C}$, under which it holds for all $C\in\mathcal{C}^{var}$ and $i\in S_C$ that $\hat{p}_{i,C}(\vec{\lambda})=\tilde{p}_{i,C}$, where $\vec{\lambda}$ is the EMV under $\mathcal{C}$ after assigning the biases to contests in $\mathcal{C}^{var}$. Moreover, such assignment of bias is unique when normalized such that $\alpha^*_{C,i}+\alpha^*_{C,\op_{i,C}}=1$.


\end{restlem}
Viewing $\hat{p}_{i,C}(\vec{\lambda})$ as the decision varaible is an effective approach, since it affects the contestants' effort exertion more directly.
Define $Q_{C}(\vec{\lambda})=\hat{p}_{i,C}(\vec{\lambda})\cdot \hat{p}_{\op_{i,C},C}(\vec{\lambda})=\hat{p}_{i,C}(\vec{\lambda})(1-\hat{p}_{i,C}(\vec{\lambda}))$ for arbitrary $i\in S_C$. Recall the definition of $\hat{x}_{i,C}(\vec{\lambda})$ in Section \ref{sec:contestant equilibrium}, for any contestant $i$ with $\lambda_i>0$, we can find that for any contest $C\in\mathcal{A}(i,\vec{\mathcal{C}})$,
$\hat{x}_{i,C}(\vec{\lambda})=\frac{R_CQ_C(\vec{\lambda})}{\lambda_{i}}$.
Observe that $Q_C(\vec{\lambda})$ is maximized when the bias is adjusted such that $\hat{p}_{i,C}(\vec{\lambda})=\frac12$ for both contestants $i\in S_C$, which we call the \emph{balancing} bias. Consequently, using the balancing bias in $C$ intuitively maximizes $x_{i,C}$ as long as the indirect influence on $\lambda_i$ is limited. 
Previous works \cite{WWX23} also suggest that, when there are only two candidate contestants, i.e., $n=2$, the optimal choice for a designer under any strategies of the other designers is to use the balancing bias. However, surprisingly, this is not a dominant strategy in the second substage of designers in our model.


\begin{restthm}{INDthmExampleBalanceNotBest}
\label{thm:example-balancing-is-not-best}
In some instances of IPM, setting the balancing bias may not be the best response strategy for a designer in the second substage of designers.
\end{restthm}

Nonetheless, we can prove that, when all designers simultaneously use the balancing biases, it forms an equilibrium. Therefore, it is still reasonable to assume that all designers will use the balancing biases.
\begin{restthm}{INDthmEqualProbabilitySecondStageEquilibrium}
\label{thm:indiv-1/2probability-is-equilibrium}
In the IPM, for a strategy profile $\vec{\mathcal{C}}$, let $\vec{\lambda}$ be the unique equilibrium multiplier vector. If it holds that $\hat{p}_{i,C_j}(\vec{\lambda})=\frac12$ for any contest $C_j$ and any contestant $i\in S_{C_j}$,  the biases of all contests in $\vec{\mathcal{C}}$ form an equilibrium in the second substage of designers.
\end{restthm}

\subsection{Equilibrium in the First Substage}
Assuming that all designers use the balancing biases in the second substage, with a little calculation, we can find that the contestants' efforts are in proportion to the prizes of contests. Therefore, the first substage of designers is strategically equivalent to a variant of weighted congestion game \cite{BGR14}, which has a pure Nash equilibrium. This guarantees the existence of WDE in the IPM.
\begin{restthm}{INDweakSPEexists}
\label{thm:weak-SPE-existence}
In the IPM, there exists at least one weak designer equilibrium.
\end{restthm}
\section{Divisible Reward Model}
\label{sec:divisible}
In this section, we concentrate on DPM, in which each designer is allowed to divide her budget to hold multiple contests. Compared to IPM, the strategy space of a designer under DPM is more complicated due to the involvement of multiple contests, but at the same time, it also become more flexible since the prize amount can be continuously adjusted across different contests to achieve some balanced state. Consequently, our result on DPM is two-fold: On the one hand, we show by an counterexample that \Cref{thm:indiv-1/2probability-is-equilibrium} cannot be extended to DPM (\Cref{thm:example-divisible-1/2notbest}), which means that using the balancing bias is sometimes no longer the best choice, even if all other designers do so. On the other hand, in contrast to IPM, we establish the existence of the designer equilibrium in DPM (\Cref{thm:div-proportional-allocation-is-SPE} \& \ref{thm:divisible-SPE-existence}), under a mild condition that $\max_{i\in[n]} T_i\leq \frac12\sum_{i\in[n]}T_i$.



We first show that \Cref{thm:indiv-1/2probability-is-equilibrium} cannot be extended to the DPM. That is, even when all designers use balancing bias simultaneously, it may not be an second-substage equilibrium.
\begin{restthm}{DIVExamplebanlanceNotBest}
\label{thm:example-divisible-1/2notbest}
In some instances of DPM, there exists some strategy profile $\vec{\mathcal{C}}$ such that:
\begin{itemize}
    \item Suppose $\vec{\lambda}$ is the EMV, it holds that $\hat{p}_{i,C}(\vec{\lambda})=\frac12$, for any contest $C\in\cup_{j\in[m]}\mathcal{C}_j$ and any participant $i\in S_C$, 
    \item However, there is some designer who has the incentive to change the biases of her contests.
\end{itemize}
\end{restthm}

However, interestingly, if every designer distributes her budget of prize proportional to the total effort of each participant and sets the balancing bias in each contest, it will be a designer equilibrium.
\begin{restthm}{DIVProportionalAllocationisSPE}
\label{thm:div-proportional-allocation-is-SPE}
In the DPM, given designers' strategy profile $\vec{\mathcal{C}}$, let $\vec{\lambda}$ be the EMV under $\vec{\mathcal{C}}$. If the following two conditions hold:
\begin{enumerate}
    \item For any designer $j$ and contestant $i$, it holds that $\sum_{C\in\mathcal{A}(i,\mathcal{C}_j)}R_C=2B_j\frac{T_i}{\sum_{k\in[n]}T_{k}}$;
    \item For any contest $C\in\cup_{j\in[m]}\mathcal{C}_j$ and any participant $i\in S_C$, it holds that $\hat{p}_{i,C}(\vec{\lambda})=\frac12$;
\end{enumerate}
then $\vec{\mathcal{C}}$ is a designer equilibrium.
\end{restthm}


Under the mild condition that the maximum effort of an individual contestant is not too large, we can show the existence of a designer equilibrium by constructing a strategy profile satisfying the condition of \Cref{thm:div-proportional-allocation-is-SPE}.
\begin{restthm}{DIVthmSPEExists}
\label{thm:divisible-SPE-existence}
In the DPM, if $\max_{i\in[n]}T_i\leq\frac12\sum_{i\in[n]}T_i$, there exists a designer equilibrium.
\end{restthm}

It's worth noting that, the designer equilibrium $\vec{\mathcal{C}}$ shown in \Cref{thm:divisible-SPE-existence} and its corresponding contestant equilibrium $\vec{x}$ exhibits a kind of balance: each contestant gets a utility proportional to her total effort, and each designer gets a utility proportional to her budget. Formally, it holds that
\begin{align*}
u_j^{contestant}(\vec{\mathcal{C}};\vec{x})&=\frac{T_i}{\sum_{i'\in[n]}T_{i'}}\sum_{j\in[m]}B_j,\\
u_j^{designer}(\vec{\mathcal{C}};\vec{x})&=\frac{B_j}{\sum_{j'\in[m]}B_{j'}}\sum_{i\in[n]}T_i,
\end{align*}
for all contestants $i\in[n]$ and all designers $j\in[m]$.

\section{Conclusion and Future Work}
This paper examines the competitive environment of multiple pairwise lottery contests, focusing on the equilibrium behavior of contest designers and contestants. Designers determine the prize amount, participants, and biases of their contests, while contestants allocate their effort across contests. We fully characterize the contestant equilibrium using the equilibrium multiplier vector. When designers can hold one or multiple contests, we demonstrate the designer equilibrium under mild conditions.

We suggest two directions for future research. The first is to extend our results to the general Tullock model with a more complex contest success function. The second is to analyze the equilibrium strategy of contestants and designers when there are more than two participants in a contest.

\newpage
\bibliography{main}

\begin{thebibliography}{39}
\providecommand{\natexlab}[1]{#1}

\bibitem[{Ackermann, Röglin, and Vöcking(2009)}]{ACKERMANN20091552}
Ackermann, H.; Röglin, H.; and Vöcking, B. 2009.
\newblock Pure Nash equilibria in player-specific and weighted congestion
  games.
\newblock \emph{Theoretical Computer Science}, 410(17): 1552--1563.
\newblock Internet and Network Economics.

\bibitem[{Azmat and M{\"o}ller(2009)}]{AM09}
Azmat, G.; and M{\"o}ller, M. 2009.
\newblock Competition among contests.
\newblock \emph{The RAND Journal of Economics}, 40(4): 743--768.

\bibitem[{Azmat and M{\"o}ller(2018)}]{AM18}
Azmat, G.; and M{\"o}ller, M. 2018.
\newblock The distribution of talent across contests.
\newblock \emph{The economic journal}, 128(609): 471--509.

\bibitem[{Bhawalkar, Gairing, and Roughgarden(2014)}]{BGR14}
Bhawalkar, K.; Gairing, M.; and Roughgarden, T. 2014.
\newblock Weighted congestion games: the price of anarchy, universal worst-case
  examples, and tightness.
\newblock \emph{ACM Transactions on Economics and Computation (TEAC)}, 2(4):
  1--23.

\bibitem[{Borel(1921)}]{B21}
Borel, E. 1921.
\newblock La th{\'e}orie du jeu et les {\'e}quations int{\'e}gralesa noyau
  sym{\'e}trique.
\newblock \emph{Comptes rendus de l’Acad{\'e}mie des Sciences},
  173(1304-1308): 58.

\bibitem[{B{\"u}y{\"u}kboyac{\i}(2016)}]{B16}
B{\"u}y{\"u}kboyac{\i}, M. 2016.
\newblock A Designer'S Choice between Single-Prize and Parallel Tournaments.
\newblock \emph{Economic Inquiry}, 54(4): 1774--1789.

\bibitem[{Chowdhury and Sheremeta(2011)}]{CS11}
Chowdhury, S.~M.; and Sheremeta, R.~M. 2011.
\newblock Multiple equilibria in Tullock contests.
\newblock \emph{Economics Letters}, 112(2): 216--219.

\bibitem[{Clark and Riis(1998)}]{CR98}
Clark, D.~J.; and Riis, C. 1998.
\newblock Contest success functions: an extension.
\newblock \emph{Economic Theory}, 201--204.

\bibitem[{Dasgupta and Nti(1998)}]{DN98}
Dasgupta, A.; and Nti, K.~O. 1998.
\newblock Designing an optimal contest.
\newblock \emph{European Journal of Political Economy}, 14(4): 587--603.

\bibitem[{Deng et~al.(2023)Deng, Gafni, Lavi, Lin, and Ling}]{DGLLL21}
Deng, X.; Gafni, Y.; Lavi, R.; Lin, T.; and Ling, H. 2023.
\newblock From Monopoly to Competition: Optimal Contests Prevail.
\newblock In \emph{Proceedings of the AAAI Conference on Artificial
  Intelligence}, volume 37(5), 5608--5615.

\bibitem[{Deng et~al.(2022)Deng, Li, Li, and Qi}]{DLLQ22}
Deng, X.; Li, N.; Li, W.; and Qi, Q. 2022.
\newblock Competition Among Parallel Contests.
\newblock In \emph{Web and Internet Economics: 18th International Conference,
  WINE 2022, Troy, NY, USA, December 12--15, 2022, Proceedings}, volume 13778,
  357. Springer Nature.

\bibitem[{DiPalantino and Vojnovic(2009)}]{DV09}
DiPalantino, D.; and Vojnovic, M. 2009.
\newblock Crowdsourcing and all-pay auctions.
\newblock In \emph{Proceedings of the 10th ACM conference on Electronic
  commerce}, 119--128.

\bibitem[{Duffy and Matros(2015)}]{DM15}
Duffy, J.; and Matros, A. 2015.
\newblock Stochastic asymmetric Blotto games: Some new results.
\newblock \emph{Economics Letters}, 134: 4--8.

\bibitem[{Epstein, Mealem, and Nitzan(2011)}]{EMN11}
Epstein, G.~S.; Mealem, Y.; and Nitzan, S. 2011.
\newblock Political culture and discrimination in contests.
\newblock \emph{Journal of Public Economics}, 95(1-2): 88--93.

\bibitem[{Ewerhart(2017)}]{E17}
Ewerhart, C. 2017.
\newblock The lottery contest is a best-response potential game.
\newblock \emph{Economics Letters}, 155: 168--171.

\bibitem[{Franke et~al.(2013)Franke, Kanzow, Leininger, and Schwartz}]{FKLS13}
Franke, J.; Kanzow, C.; Leininger, W.; and Schwartz, A. 2013.
\newblock Effort maximization in asymmetric contest games with heterogeneous
  contestants.
\newblock \emph{Economic Theory}, 52: 589--630.

\bibitem[{Friedman(1958)}]{F58}
Friedman, L. 1958.
\newblock Game-theory models in the allocation of advertising expenditures.
\newblock \emph{Operations research}, 6(5): 699--709.

\bibitem[{Fu and Wu(2018)}]{FW18}
Fu, Q.; and Wu, Z. 2018.
\newblock On the optimal design of lottery contests.
\newblock \emph{Available at SSRN 3291874}.

\bibitem[{Fu, Wu, and Zhu(2022)}]{FWZ22}
Fu, Q.; Wu, Z.; and Zhu, Y. 2022.
\newblock On equilibrium existence in generalized multi-prize nested lottery
  contests.
\newblock \emph{Journal of Economic Theory}, 200: 105377.

\bibitem[{Ghosh and Goldberg(2023)}]{GG23}
Ghosh, A.; and Goldberg, P.~W. 2023.
\newblock Best-Response Dynamics in Lottery Contests.
\newblock \emph{arXiv preprint arXiv:2305.10881}.

\bibitem[{Gross and Wagner(1950)}]{GW50}
Gross, O.; and Wagner, R. 1950.
\newblock A continuous Colonel Blotto game.
\newblock Technical report, Rand Project Air Force Santa Monica Ca.

\bibitem[{Juang, Sun, and Yuan(2020)}]{JSY20}
Juang, W.-T.; Sun, G.-Z.; and Yuan, K.-C. 2020.
\newblock A model of parallel contests.
\newblock \emph{International Journal of Game Theory}, 49(2): 651--672.

\bibitem[{Kakutani(1941)}]{K41}
Kakutani, S. 1941.
\newblock A generalization of Brouwer’s fixed point theorem.

\bibitem[{K{\"o}rpeo{\u{g}}lu, Korpeoglu, and Hafal{\i}r(2022)}]{KK22}
K{\"o}rpeo{\u{g}}lu, E.; Korpeoglu, C.~G.; and Hafal{\i}r, {\.I}.~E. 2022.
\newblock Parallel innovation contests.
\newblock \emph{Operations Research}, 70(3): 1506--1530.

\bibitem[{Kovenock and Roberson(2021)}]{KR21}
Kovenock, D.; and Roberson, B. 2021.
\newblock Generalizations of the general lotto and colonel blotto games.
\newblock \emph{Economic Theory}, 71: 997--1032.

\bibitem[{Li et~al.(2019)Li, Yan, Deng, Qi, Chu, Song, Qiao, He, and
  Xiong}]{LYDQCS19}
Li, C.; Yan, X.; Deng, X.; Qi, Y.; Chu, W.; Song, L.; Qiao, J.; He, J.; and
  Xiong, J. 2019.
\newblock Latent dirichlet allocation for internet price war.
\newblock In \emph{Proceedings of the AAAI Conference on Artificial
  Intelligence}, volume 33(01), 639--646.

\bibitem[{Li and Zheng(2022)}]{LZ22}
Li, X.; and Zheng, J. 2022.
\newblock Pure strategy Nash Equilibrium in 2-contestant generalized lottery
  Colonel Blotto games.
\newblock \emph{Journal of Mathematical Economics}, 103: 102771.

\bibitem[{Macdonell and Mastronardi(2015)}]{MM15}
Macdonell, S.~T.; and Mastronardi, N. 2015.
\newblock Waging simple wars: a complete characterization of two-battlefield
  Blotto equilibria.
\newblock \emph{Economic Theory}, 58(1): 183--216.

\bibitem[{Nti(1999)}]{N99}
Nti, K.~O. 1999.
\newblock Rent-seeking with asymmetric valuations.
\newblock \emph{Public Choice}, 98(3-4): 415--430.

\bibitem[{Nti(2004)}]{N04}
Nti, K.~O. 2004.
\newblock Maximum efforts in contests with asymmetric valuations.
\newblock \emph{European journal of political economy}, 20(4): 1059--1066.

\bibitem[{Roberson(2006)}]{R06}
Roberson, B. 2006.
\newblock The colonel blotto game.
\newblock \emph{Economic Theory}, 29(1): 1--24.

\bibitem[{Robson et~al.(2005)}]{R05}
Robson, A.~R.; et~al. 2005.
\newblock Multi-item contests.

\bibitem[{Segev(2020)}]{S20}
Segev, E. 2020.
\newblock Crowdsourcing contests.
\newblock \emph{European Journal of Operational Research}, 281(2): 241--255.

\bibitem[{Skaperdas(1996)}]{S96}
Skaperdas, S. 1996.
\newblock Contest success functions.
\newblock \emph{Economic theory}, 7: 283--290.

\bibitem[{Stein(2002)}]{S02}
Stein, W.~E. 2002.
\newblock Asymmetric rent-seeking with more than two contestants.
\newblock \emph{Public Choice}, 113(3-4): 325--336.

\bibitem[{Tullock(2001)}]{T01}
Tullock, G. 2001.
\newblock Efficient rent seeking.
\newblock \emph{Efficient rent-seeking: Chronicle of an intellectual quagmire},
  3--16.

\bibitem[{Wang(2010)}]{W10}
Wang, Z. 2010.
\newblock The optimal accuracy level in asymmetric contests.
\newblock \emph{The BE Journal of Theoretical Economics}, 10(1).

\bibitem[{Wang, Wu, and Xing(2023)}]{WWX23}
Wang, Z.; Wu, Z.; and Xing, Z. 2023.
\newblock Multi-battle Contests with Competing Battlefields.
\newblock \emph{Working Paper}.

\bibitem[{Xu and Zhou(2018)}]{XZ18}
Xu, J.; and Zhou, J. 2018.
\newblock Discriminatory power and pure strategy Nash equilibrium in the
  lottery Blotto game.
\newblock \emph{Operations Research Letters}, 46(4): 424--429.

\end{thebibliography}

\newpage
\appendix
\section*{Appendix}
\section{Missing Proofs in \Cref{sec:contestant equilibrium}}

\subsection{Proof of Lemma \ref{lemma:lagrange-dual}}
\CEEquilibriumMultiplierVector*

\begin{proof}
    We show the \Cref{lemma:lagrange-dual} by giving a sufficient and necessary condition on the best response of any contestant, which is proven by the following lemma. 

\begin{restlem}{CEnecessaryCondition}
\label{lemma:contestantequilibrium-no-both-0}
For any feasible strategy profile $\vec{x}$ and any contestant $i$, if $x_i$ is contestant $i$'s best response to $\vec{x}_{-i}$, the following statements hold:
\begin{enumerate}
    \item For any $C\in \mathcal{A}(i,\mathcal{C})$, $x_{i,C}+x_{\op_{i,C},C}>0$.
    \item If there exists $C\in \mathcal{A}(i,\mathcal{C})$ with $x_{\op_{i,C},C}>0$, then $\sum_{C\in\mathcal{A}(i,\mathcal{C})}x_{i,C}=T_i$.
\end{enumerate}
\end{restlem}
\begin{proof}
We prove Statement 1 by contradiction. Intuitively, a contestant's best response should always exert positive effort (rather than zero) in any contest where the opponent exerts zero effort, which can increase her winning probability from $1/2$ to $1$.
Suppose for contradiction that there exists $C\in \mathcal{A}(i,\mathcal{C})$ with $x_{i,C}=x_{\op_{i,C},C}=0$. 
Since $x_{\op_{i,C},C}=0$, the winning probability of contestant $i$ in $C$ is 
\begin{align*}
f(\alpha_{C,i}x_{i,C};0)=\begin{cases}1,&\text{ if }x_{i,C}>0;\\\frac12,&\text{ if }x_{i,C}=0. \end{cases}
\end{align*}
Therefore, contestant $i$ will get a utility more than $R_C\cdot(1-\frac12)=\frac12 R_C>0$ from $C$ if she modifies $x_{i,C}$ to any $\epsilon>0$. We discuss two cases:

If there exists a contest $C'\in\mathcal{A}(i,\mathcal{C})$ such that $x_{i,C'}>0$, since $f(\alpha_{C',i}x_{i,C'};\alpha_{C',\op_{i,C'}}x_{\op_{i,C'},C'})$ is continuous when $x_{i,C'}\in(0,+\infty)$, there exists $\epsilon\in(0,x_{i,C'})$ such that $f(\alpha_{C',i}(x_{i,C'}-\epsilon);\alpha_{C',\op_{i,C'}}x_{\op_{i,C'},C'})>f(\alpha_{C',i}x_{i,C'};\alpha_{C',\op_{i,C'}}x_{\op_{i,C'},C'})-\frac{\frac12 R_C}{R_{C'}}$. Therefore, we can construct strategy $x_{i}'$ where $x_{i,C}'=\epsilon,x_{i,C'}'=x_{i,C'}-\epsilon$, and $x_{i,C''}'=x_{i,C''}$ for all $C''\in\mathcal{A}(i,\mathcal{C})\setminus\{C,C'\}$, which gets strictly more utility for contestant $i$ than $x_i$. This contradicts with the assumption that $x_{i}$ is a best response.

If for any contest $C'\in\mathcal{A}(i,\mathcal{C})$, it has $x_{i,C'}=0$. We can construct strategy $x_{i}'$ where $x_{i,C}'=T_{i}>0$, $x_{i,C'}'=x_{i,C'}=0$ for all $C'\in\mathcal{A}(i,\mathcal{C})\setminus\{C\}$, which also gets a utility more than $\frac12 R_C$ from contest $\mathcal{C}$ and has a more total utility. It is a contradiction. 
In summary, we show that Statement 1 holds.

For Statement 2, we still prove it by contradiction. Suppose there exists $C\in\mathcal{A}(i,\mathcal{C})$ such that $x_{\op_{i,C},C}>0$, but $\sum_{C\in\mathcal{A}(i,\mathcal{C})}x_{i,C}<T_i$. Since $f(\alpha_{C,i}x_{i,C};\alpha_{C,\op_{i,C}}x_{\op_{i,C}})$ is strictly increasing in $x_{i,C}$ for $x_{i,C}\in[0,+\infty)$, we can construct a feasible strategy $x_{i}'$ where $x_{i,C}'=x_{i,C}+T_{i}-\sum_{C\in\mathcal{A}(i,\mathcal{C})}x_{i,C}>x_{i,C}$, $x_{i,C'}'=x_{i,C'}$ for all $C'\in\mathcal{A}(i,\mathcal{C})\setminus\{C\}$, which gets a utility for contestant $i$ strictly more than that of $x_i$, contradicting with the assumption that $x_{i}$ is a best response. Therefore, Statement 2 also holds. 
\end{proof}

With this necessary condition on the best response, we present a sufficient and necessary condition on best response, shown in \Cref{lemma:bestresponse-lagrange}.
\begin{restlem}{CEbestresponseLagrange}
\label{lemma:bestresponse-lagrange}
In a feasible strategy profile $\vec{x}$, for any $i\in[n]$, a sufficient and necessary condition of that $x_i$ is contestant $i$'s best response to $\vec{x}_{-i}$ is: there exists a $\lambda_i\in\mathbb{R}_{\geq 0}$, such that for any $C\in\mathcal{A}(i,\mathcal{C})$, if $x_{i,C}>0$, then $R_C\cdot \frac{\partial p_{i,C}(\vec{x})}{\partial x_{i,C}}=\lambda_i$; if $x_{i,C}=0$, then $R_C\cdot \frac{\partial p_{i,C}(\vec{x})}{\partial x_{i,C}}\leq\lambda_i$.\footnote{Note that the existence of $\frac{\partial p_{i,C}(\vec{x})}{\partial x_{i,C}}$ implicitly implies that $x_{i,C}+x_{\op_{i,C},C}>0$.}
\end{restlem}



\begin{proof}
We first recall that
$$p_{i,C}(\vec{x})=\begin{cases}
    \frac{\alpha_{C,i} x_{i,C}}{\alpha_{C,i} x_{i,C}+\alpha_{C,\op_{i,C}} x_{\op_{i,C},C}}, &\text{if } x_{i,C}+x_{\op_{i,C},C}>0,\\
    \frac{1}{2}, & \text{if }x_{i,C}=x_{\op_{i,C},C}=0.
\end{cases}$$
We calculate that when $x_{i,C}+x_{\op_{i,C}}>0$, it holds that
$$\frac{\partial p_{i,C}(\vec{x})}{\partial x_{i,C}}=\frac{\alpha_{C,i}\alpha_{C,\op_{i,C}}x_{\op_{i,C},C}}{(\alpha_{C,i}x_{i,C}+\alpha_{C,\op_{i,C}}x_{\op_{i,C},C})^2}.$$

We also recall that for each contestant $i$, a best response $x_i$ is an optimal solution to the following optimization problem:
\begin{align}
\max_{x_{i,C}\geq0\text{ for }C\in\mathcal{A}(i,\mathcal{C})}&~\sum_{C\in\mathcal{A}(i,\mathcal{C})}R_C\cdot p_{i,C}(x_i,\vec{x}_{-i}),\label{eq:xi-best-response-optimization} \\
\mbox{s.t.} &\sum_{C\in \mathcal{A}(i,\mathcal{C})}x_{i,C}\leq T_i.\nonumber
\end{align}

\subsubsection{Necessity:} 

Take $\lambda_i=\max_{C\in\mathcal{A}(i,\mathcal{C})}R_{C} \cdot \frac{\partial p_{i,C}(\vec{x})}{\partial x_{i,C}}$.
For any $C_1\in\mathcal{A}(i,\mathcal{C})$ such that $x_{i,C_1}>0$, if there exists $C_2\in\mathcal{A}(i,\mathcal{C})$ such that $R_{C_1} \cdot \frac{\partial p_{i,C_1}(\vec{x})}{\partial x_{i,C_1}}<R_{C_2} \cdot \frac{\partial p_{i,C_2}(\vec{x})}{\partial x_{i,C_2}}$, then contestant $i$ can deviate to $x_i''$ where $x_{i,C_1}''=x_{i,C_1}-\epsilon$, $x_{i,C_2}''=x_{i,C_2}+\epsilon$ for some small enough $\epsilon>0$, and $x_{i,C'}''=x_{i,C'}$ for all other $C'$,  so that $\sum_{C\in\mathcal{A}(i,\mathcal{C})}R_C\cdot p_{i,C}(x_i'',\vec{x}_{-i})>\sum_{C\in\mathcal{A}(i,\mathcal{C})}R_C\cdot p_{i,C}(\vec{x})$, which contradicts with the assumption that $x_i$ is a best response. Therefore it holds that $R_{C_1} \cdot \frac{\partial p_{i,C_1}(\vec{x})}{\partial x_{i,C_1}}=\max_{C\in\mathcal{A}(i,\mathcal{C})}R_{C} \cdot \frac{\partial p_{i,C}(\vec{x})}{\partial x_{i,C}}=\lambda_i$, which implies the necessity.

\subsubsection{Sufficiency:} We discuss the sufficiency in two cases:

If there exists some $C\in\mathcal{A}(i,\mathcal{C})$ that $x_{\op_{i,C},C}=0$, then $x_{i,C}>0$, and it holds that $\lambda_i=R_C\cdot \frac{\partial p_{i,C}(\vec{x})}{\partial x_{i,C}}=0$. Therefore, for all $C\in\mathcal{A}(i,\mathcal{C})$,  we have $R_C\cdot \frac{\partial p_{i,C}(\vec{x})}{\partial x_{i,C}}=0$, implying that $x_{\op_{i,C},C}=0$ and $p_{i,C}(\vec{x})=1$. Therefore $x_i$ is a best response.

If $x_{\op_{i,C},C}>0$ for all $C\in\mathcal{A}(i,\mathcal{C}))$, we can observe that for each $C\in\mathcal{A}(i,\mathcal{C}))$, $p_{i,C}(x_i,\vec{x}_{-i})$ is a strictly increasing, strictly concave, and differentiable function in $x_{i,C}$ for $x_{i,C}\in[0,+\infty)$. Therefore, the objective function $\sum_{C\in\mathcal{A}(i,\mathcal{C})}R_C\cdot p_{i,C}(x_i,\vec{x}_{-i})$ of Optimization (\ref{eq:xi-best-response-optimization}) is a strictly concave and differentiable function in $x_i$ on the convex and compact feasible region $\{x_i=(x_{i,C})_{C\in\mathcal{A}(i,\mathcal{C})}:\sum_{C\in\mathcal{A}(i,\mathcal{C})}x_{i,C}\leq T_i\land \forall C\in\mathcal{A}(i,\mathcal{C}),x_{i,C}\geq 0\}$. Thus, Optimization (\ref{eq:xi-best-response-optimization}) can be viewed as a convex optimization problem with affine constraints.
We can prove that $x_i$ is the optimal solution to problem (\ref{eq:xi-best-response-optimization}) through the KKT conditions. We view $\lambda_i$ as the dual variable for the constraint $\sum_{C\in\mathcal{A}(i,\mathcal{C})}x_{i,C}\leq T_i$. And for each $C\in\mathcal{A}(i,\mathcal{C})$, let $\mu_C=\lambda_i-R_C\frac{\partial p_{i,C}(\vec{x})}{\partial x_{i,C}}\geq 0$ be the dual variable for the constraint $x_{i,C}\geq 0$. One can easily verify that $(x_{i,C})_{C\in\mathcal{A}(i,\mathcal{C})},\lambda_i,(\mu_C)_{C\in\mathcal{A}(i,\mathcal{C})}$ satisfy the KKT conditions of Optimization (\ref{eq:xi-best-response-optimization}), and the strong duality holds by Slater's condition. Therefore, it follows that $x_i$ is the optimal solution to problem (\ref{eq:xi-best-response-optimization}).
In summary, we show the sufficiency.
\end{proof}


Turn back to proving \Cref{lemma:lagrange-dual}. Since $\vec{x}$ is a contestant equilibrium under $\mathcal{C}$, for each contestant $i\in[n]$, $x_i$ is contestant $i$'s best response to $\vec{x}_{-i}$. Therefore, there exists $\lambda_i\in \mathbb{R}_{\geq 0}$ satisfying the condition of \Cref{lemma:bestresponse-lagrange} and we prove the \Cref{lemma:lagrange-dual}.
\end{proof}

\subsection{Proof of \Cref{lemma:reconstruct-effort-from-EMV}}
\CEhatxfromEMV*
\begin{proof}
We know that $\vec{x}$ and $\vec{\lambda}$ satisfies the conditions in \Cref{lemma:lagrange-dual}. Recall that $\hat{x}_{i,C}(\vec{\lambda})=\frac{
R_C\alpha_{C,i}\alpha_{C,\op_{i,C}}\lambda_{\op_{i,C}}
}{
(\alpha_{C,\op_{i,C}}\lambda_i+\alpha_{C,i}\lambda_{\op_{i,C}})^2
}$.
For any contestant $i$ and any $C\in\mathcal{A}(i,\mathcal{C})$, we discuss this lemma in three cases:

(a) If $\lambda_i>0$ and $x_{i,C}>0$, by \Cref{lemma:lagrange-dual}, we have $\lambda_i=R_C\cdot \frac{\alpha_{C,i}\alpha_{C,\op_{i,C}}x_{\op_{i,C},C}}{(\alpha_{C,i}x_{i,C}+\alpha_{C,\op_{i,C}}x_{\op_{i,C},C})^2}$, which implies that $x_{\op_{i,C},C}>0$ and $\lambda_{\op_{i,C}}=R_C\cdot \frac{\alpha_{C,i}\alpha_{C,\op_{i,C}}x_{i,C}}{(\alpha_{C,i}x_{i,C}+\alpha_{C,\op_{i,C}}x_{\op_{i,C},C})^2}$. With a little calculation, we get
\begin{align*}
&(\alpha_{C,\op_{i,C}}\lambda_i+\alpha_{C,i}\lambda_{\op_{i,C}})^2\\
=&(\frac{R_C\alpha_{C,i}\alpha_{C,\op_{i,C}}}{\alpha_{C,i}x_{i,C}+\alpha_{C,\op_{i,C}}x_{\op_{i,C},C}})^2\\
=&\lambda_{\op_{i,C}}\cdot \frac{R_C\alpha_{C,i}\alpha_{C,\op_{i,C}}}{x_{i,C}}.
\end{align*}
It follows that $x_{i,C}=\frac{
R_C\alpha_{C,i}\alpha_{C,\op_{i,C}}\lambda_{\op_{i,C}}
}{
(\alpha_{C,\op_{i,C}}\lambda_i+\alpha_{C,i}\lambda_{\op_{i,C}})^2
}=\hat{x}_{i,C}(\vec{\lambda})$.

(b) If $\lambda_i>0$ and $x_{i,C}=0$, by \Cref{lemma:contestantequilibrium-no-both-0}, we know that $x_{\op_{i,C},C}>0$, which means that  $\lambda_{\op_{i,C}}=R_C\cdot \frac{\alpha_{C,i}\alpha_{C,\op_{i,C}}x_{i,C}}{(\alpha_{C,i}x_{i,C}+\alpha_{C,\op_{i,C}}x_{\op_{i,C},C})^2}=0$ by \Cref{lemma:lagrange-dual}. It still holds that $x_{i,C}=0=\frac{
R_C\alpha_{C,i}\alpha_{C,\op_{i,C}}\lambda_{\op_{i,C}}
}{
(\alpha_{C,\op_{i,C}}\lambda_i+\alpha_{C,i}\lambda_{\op_{i,C}})^2
}=\hat{x}_{i,C}(\vec{\lambda})$.

(c) If $\lambda_i=0$, we have 
$$0=\lambda_i\geq R_C\cdot \frac{\alpha_{C,i}\alpha_{C,\op_{i,C}}x_{\op_{i,C},C}}{(\alpha_{C,i}x_{i,C}+\alpha_{C,\op_{i,C}}x_{\op_{i,C},C})^2},$$
which leads to that $x_{\op_{i,C},C}=0$. Therefore, $\lambda_{\op_{i,C}}\geq R_C\cdot \frac{\alpha_{C,i}\alpha_{C,\op_{i,C}}x_{i,C}}{(\alpha_{C,i}x_{i,C})^2}=R_C\cdot \frac{\alpha_{C,\op_{i,C}}}{\alpha_{C,i}x_{i,C}}$. This means that $x_{i,C}\geq \frac{R_C\alpha_{C,\op_{i,C}}}{\alpha_{C,i}\lambda_{\op_{i,C}}}=\frac{
R_C\alpha_{C,i}\alpha_{C,\op_{i,C}}\lambda_{\op_{i,C}}
}{
(\alpha_{C,\op_{i,C}}\lambda_i+\alpha_{C,i}\lambda_{\op_{i,C}})^2
}=\hat{x}_{i,C}(\vec{\lambda})$.

In summary, the lemma holds for any $i\in[n]$ and any $C\in\mathcal{A}(i,\mathcal{C})$.
\end{proof}

\subsection{Proof of \Cref{lemma:dual-solution-condition}}
\CEconditionforEMV*
\begin{proof}
We first show the necessity. 
\subsubsection{Necessity:} If $\vec{\lambda}$ is an equilibrium multiplier vector, there is a contestant equilibrium $\vec{x}$, satisfying the conditions stated in \Cref{lemma:lagrange-dual}. 

For Statement 1, for any $C\in\mathcal{C}$, suppose $S_C=\{i_1,i_2\}$, by \Cref{lemma:contestantequilibrium-no-both-0} we know that $x_{i_1,C}+x_{i_2,C}>0$. Without loss of generality, assume $x_{i_2,C}>0$. By \Cref{lemma:lagrange-dual}, we have
\begin{align*}\sum_{i\in S_C}\lambda_i\geq\lambda_{i_1}&\geq R_C\cdot \frac{\partial p_{i_1,C}(\vec{x})}{\partial x_{i_1,C}}\\
&=R_C\cdot \frac{\partial f(\alpha_{C,i_1}x_{i_1,C};\alpha_{C,i_2}x_{i_2,C})}{\partial x_{i_1,C}}\\
&=R_C\cdot \frac{\alpha_{C,i_1}\alpha_{C,i_2}x_{i_2,C}}{(\alpha_{C,i_1}x_{i_1,C}+\alpha_{C,i_2}x_{i_2,C})^2}>0.
\end{align*}
This holds for any $C\in\mathcal{C}$, which implies that  $\vec{\lambda}$ is valid.

For Statement 2 and Statement 3, for any contestant $i$, by \Cref{lemma:reconstruct-effort-from-EMV}, we know that for any contest $C\in\mathcal{A}(i,\mathcal{C})$, $x_{i,C}\geq \hat{x}_{i,C}(\vec{\lambda})$, and $\hat{T}_i(\vec{\lambda})=\sum_{C\in\mathcal{A}(i,\mathcal{C})}\hat{x}_{i,C}(\vec{\lambda})\leq \sum_{C\in\mathcal{A}(i,\mathcal{C})}x_{i,C}\leq T_i$.
Moreover, when $\lambda_i>0$, we know that $x_{i,C}=\hat{x}_{i,C}(\vec{\lambda})$, and $x_{\op_{i,C},C}\geq \hat{x}_{\op_{i,C},C}(\vec{\lambda})=\frac{
R_C\alpha_{C,i}\alpha_{C,\op_{i,C}}\lambda_{i}
}{
(\alpha_{C,\op_{i,C}}\lambda_i+\alpha_{C,i}\lambda_{\op_{i,C}})^2}>0$. Therefore, by Statement 2 of \Cref{lemma:contestantequilibrium-no-both-0}, we have $T_i=\sum_{C\in\mathcal{A}(i,\mathcal{C})}x_{i,C}(\vec{\lambda})=\sum_{C\in\mathcal{A}(i,\mathcal{C})}\hat{x}_{i,C}(\vec{\lambda})=\hat{T}_i(\vec{\lambda})$.

\subsubsection{Sufficiency:} 
Suppose $\vec{\lambda}$ satisfies all three statements, we construct $\vec{x}$ such that $x_{i,C}=\hat{x}_{i,C}(\vec{\lambda})$. We prove that $\vec{x}$ is a contestant equilibrium corresponding to $\vec{\lambda}$, so that $\vec{\lambda}$ is an equilibrium multiplier vector.

Firstly, from Statement 2 and Statement 3 we know that $\sum_{C\in\mathcal{A}(i,\mathcal{C})}x_{i,C}$ $=\hat{T}_i(\vec{\lambda})\leq T_i$ holds for any contestant  $i\in[n]$, so $\vec{x}$ is a feasible strategy profile of contestants. Next, we prove that $\vec{x}$ is a contestant equilibrium with the help of \Cref{lemma:bestresponse-lagrange}. 

We only need to verify that for any contest $C\in\mathcal{C}$ and any contestant $i\in S_C$, it holds that $R_C\cdot\frac{\partial p_{i,C}(\vec{x})}{\partial x_{i,C}}=\lambda_i$ if $x_{i,C}>0$, and $R_C\cdot\frac{\partial p_{i,C}(\vec{x})}{\partial x_{i,C}}\leq\lambda_i$ if $x_{i,C}=0$, which implies that the condition of \Cref{lemma:bestresponse-lagrange} is satisfied for all contestants, i.e., each $x_i$ is a best response to $\vec{x}_{-i}$, and consequently $\vec{x}$ is a contestant equilibrium.

For any contest $C\in\mathcal{A}$, suppose $S_C=\{i_1,i_2\}$. Since $\vec{\lambda}$ is valid, we consider the following two cases:

(a) If both $\lambda_{i_1}$ and $\lambda_{i_2}$ are positive, we have 
$$x_{i_1,C}=\hat{x}_{i_1,C}(\vec{\lambda})=
\frac{
R_C\alpha_{C,i_1}\alpha_{C,i_2}\lambda_{i_2}
}{
(\alpha_{C,i_2}\lambda_{i_1}+\alpha_{C,i_1}\lambda_{i_2})^2}$$ and 
$$x_{i_2,C}=\hat{x}_{i_2,C}(\vec{\lambda})=
\frac{
R_C\alpha_{C,i_1}\alpha_{C,i_2}\lambda_{i_1}
}{
(\alpha_{C,i_2}\lambda_{i_1}+\alpha_{C,i_1}\lambda_{i_2})^2}.$$ Observe that $\alpha_{C,i_1}x_{i_1,C}+\alpha_{C,i_2}x_{i_2,C}=\frac{
R_C\alpha_{C,i_1}\alpha_{C,i_2}
}{
\alpha_{C,i_2}\lambda_{i_1}+\alpha_{C,i_1}\lambda_{i_2}}
$. We can calculate
\begin{align*}
R_C\cdot \frac{\partial p_{i_1,C}(\vec{x})}{\partial x_{i_1,C}}
=&\frac{R_C\alpha_{C,i_1}\alpha_{C,i_2}x_{i_2,C}}{(\alpha_{C,i_1}x_{i_1,C}+\alpha_{C,i_2}x_{i_2,C})^2}\\
=&\frac{R_C\alpha_{C,i_1}\alpha_{C,i_2}\frac{
R_C\alpha_{C,i_1}\alpha_{C,i_2}\lambda_{i_1}
}{
(\alpha_{C,i_2}\lambda_{i_1}+\alpha_{C,i_1}\lambda_{i_2})^2}}{(\frac{
R_C\alpha_{C,i_1}\alpha_{C,i_2}
}{
\alpha_{C,i_2}\lambda_{i_1}+\alpha_{C,i_1}\lambda_{i_2}})^2}\\
=&\lambda_{i_1},
\end{align*}
and similarly, $R_C\cdot \frac{\partial p_{i_2,C}(\vec{x})}{\partial x_{i_2,C}}=\lambda_{i_2}$.

(b) If only one of $\lambda_{i_1}$ and $\lambda_{i_2}$ is zero, without loss of generality, assume $\lambda_{i_1}>0$ and $\lambda_{i_2}=0$. We have 
$$x_{i_1,C}=\hat{x}_{i_1,C}(\vec{\lambda})=
\frac{
R_C\alpha_{C,i_1}\alpha_{C,i_2}\lambda_{i_2}
}{
(\alpha_{C,i_2}\lambda_{i_1}+\alpha_{C,i_1}\lambda_{i_2})^2}=0$$ and 
$$x_{i_2,C}=\hat{x}_{i_2,C}(\vec{\lambda})=
\frac{
R_C\alpha_{C,i_1}\alpha_{C,i_2}\lambda_{i_1}
}{
(\alpha_{C,i_2}\lambda_{i_1}+\alpha_{C,i_1}\lambda_{i_2})^2}=\frac{
R_C\alpha_{C,i_1}
}{
\alpha_{C,i_2}\lambda_{i_1}}.$$ We can calculate that
\begin{align*}
R_C\cdot \frac{\partial p_{i_1,C}(\vec{x})}{\partial x_{i_1,C}}
=&R_C\cdot \frac{\alpha_{C,i_1}\alpha_{C,i_2}x_{i_2,C}}{(\alpha_{C,i_1}x_{i_1,C}+\alpha_{C,i_2}x_{i_2,C})^2}\\
=&\frac{R_C\alpha_{C,i_1}}{\alpha_{C,i_2}x_{i_2,C}}\\
= &\lambda_{i_1},
\end{align*}
and $R_C\cdot \frac{\partial p_{i_2,C}(\vec{x})}{\partial x_{i_2,C}}=0=\lambda_{i_2}$.

In summary, the condition of \Cref{lemma:bestresponse-lagrange} is satisfied for any contestant $i$. Therefore, we get that for each $i\in[n]$, $x_i$ is a best response to $\vec{x}_{-i}$, which implies that $\vec{x}$ is a contestant equilibrium. Finally, since $\vec{\lambda}$ and $\vec{x}$ satisfy the conditions in \Cref{lemma:lagrange-dual},  $\vec{\lambda}$ is an equilibrium multiplier vector.
\end{proof}

\subsection{Proof of \Cref{lemma:characterize-contestant-equilibrium-by-multiplier}}
\CEcharacterizeAllCEbyEMV*
\begin{proof} We first show the necessity. 

\subsubsection{Necessity:} Suppose $\vec{x}$ is a contestant equilibrium corresponding to $\vec{\lambda}$. For any contestant $i$, since $x_i$ is a feasible strategy, we have $\sum_{C\in\mathcal{A}(i,\mathcal{C})}x_{i,C}\leq T_i$. By \Cref{lemma:reconstruct-effort-from-EMV}, for any contest $C\in\mathcal{A}(i,\mathcal{C})$, we have $x_{i,C}\geq \hat{x}_{i,C}(\vec{\lambda})$. By the definition of $\mathcal{X}(\vec{\lambda})$, we get $\vec{x}\in\mathcal{X}(\vec{\lambda})$.

\subsubsection{Sufficiency:} Suppose $\vec{x}\in\mathcal{X}(\vec{\lambda})$ is a contestant equilibrium.  Similar to the proof of sufficiency of \Cref{lemma:dual-solution-condition}, we prove that $\vec{x}$ and $\vec{\lambda}$  satisfies the condition of \Cref{lemma:bestresponse-lagrange} for any contestant $i\in[n]$. 
For any contest $C\in\mathcal{A}$, suppose $S_C=\{i_1,i_2\}$. Since $\vec{\lambda}$ is valid, we consider the following two cases:

(a) If both $\lambda_{i_1}$ and $\lambda_{i_2}$ are positive, we have $x_{i_1,C}=\hat{x}_{i_1,C}(\vec{\lambda})$ and $x_{i_2,C}=\hat{x}_{i_2,C}(\vec{\lambda})$. We already know from the proof of sufficiency of \Cref{lemma:dual-solution-condition} that
$R_C\cdot \frac{\partial p_{i_1,C}(\vec{x})}{\partial x_{i_1,C}}=\lambda_{i_1}$ and $R_C\cdot \frac{\partial p_{i_2,C}(\vec{x})}{\partial x_{i_2,C}}=\lambda_{i_2}$.

(b) If only one of $\lambda_{i_1}$ and $\lambda_{i_2}$ is zero, without loss of generality, assume $\lambda_{i_1}>0$ and $\lambda_{i_2}=0$. We have $x_{i_1,C}=\hat{x}_{i_1,C}(\vec{\lambda})=
\frac{
R_C\alpha_{C,i_1}\alpha_{C,i_2}\lambda_{i_2}
}{
(\alpha_{C,i_2}\lambda_{i_1}+\alpha_{C,i_1}\lambda_{i_2})^2}=0$ and $x_{i_2,C}\geq\hat{x}_{i_2,C}(\vec{\lambda})=\frac{
R_C\alpha_{C,i_1}
}{
\alpha_{C,i_2}\lambda_{i_1}}$. We can calculate that
\begin{align*}
R_C\cdot \frac{\partial p_{i_1,C}(\vec{x})}{\partial x_{i_1,C}}
=&R_C\cdot \frac{\alpha_{C,i_1}\alpha_{C,i_2}x_{i_2,C}}{(\alpha_{C,i_1}x_{i_1,C}+\alpha_{C,i_2}x_{i_2,C})^2}\\
=&\frac{R_C\alpha_{C,i_1}}{\alpha_{C,i_2}x_{i_2,C}}\\
\leq &\lambda_{i_1},
\end{align*}
and $R_C\cdot \frac{\partial p_{i_2,C}(\vec{x})}{\partial x_{i_2,C}}=0=\lambda_{i_2}$.

In summary, the condition of \Cref{lemma:bestresponse-lagrange} is satisfied for any contestant $i\in[n]$. Therefore, we get that for each $i\in[n]$, $x_i$ is a best response to $\vec{x}_{-i}$, and $\vec{x}$ is a contestant equilibrium corresponding to $\vec{\lambda}$.
\end{proof}

\subsection{Proof of \Cref{thm:contestant-equilibrium-existence}}
\CEthmExistenceofEMV*

\begin{proof}
To prove that there exists an equilibrium multiplier vector $\vec{\lambda}$, we construct a continuous updating function of the vector $\vec{\lambda}\in\mathbb{R}_{\geq 0}^{n}$, and apply Brouwer's fixed point theorem to show the existence of a fixed point of this updating function. Then, we show that the fixed point satisfies the conditions in \Cref{lemma:dual-solution-condition} and is consequently an equilibrium multiplier vector.

Let $\overline{T}=2\max_{i\in[n]}T_i$.
We define an updating function $\phi:\mathbb{R}_{\geq 0}^{n}\to\mathbb{R}_{\geq 0}^{n}$, such that for any contestant $i\in[n]$,
\begin{align*}
    &\phi_i(\vec{\lambda})=\min\{\hat{\lambda_i}\geq 0:\\
&\sum_{C\in\mathcal{A}(i,\mathcal{C})}
\frac{R_C\alpha_{C,i}\alpha_{C,\op_{i,C}}\lambda_{\op_{i,C}}}
{(\alpha_{C,\op_{i,C}}\hat{\lambda_i}+\alpha_{C,i}\lambda_{\op_{i,C}})^2}
\leq T_i\\
&\land \forall C\in\mathcal{A}(i,\mathcal{C}),
\frac{\hat{\lambda_i}}{\alpha_{C,i}}+\frac{\lambda_{\op_{i,C}}}{\alpha_{C,\op_{i,C}}}
\geq\frac{R_C}{(\alpha_{C,i}+\alpha_{C,\op_{i,C}})\overline{T}}\}.
\end{align*}
Specially, we assume $\frac{R_C\alpha_{C,i}\alpha_{C,\op_{i,C}}\lambda_{\op_{i,C}}}
{(\alpha_{C,\op_{i,C}}\hat{\lambda_i}+\alpha_{C,i}\lambda_{\op_{i,C}})^2}=0$ when $\hat{\lambda_i}=\lambda_{\op_{i,C}}=0$ in the above definition.

Let $M=\max_{i\in[n]}\frac{1}{T_i}\sum_{C\in\mathcal{A}(i,\mathcal{C})}\frac{R_C\alpha_{C,i}}{\alpha_{C,\op_{i,C}}}$, consider the restriction of $\phi$ on the closed region $\Omega=[0,M]^n$. We can show that for any $\vec{\lambda}\in\Omega$, it holds that $\phi(\vec{\lambda})\in\Omega$.

For any contestant $i$, we prove that $\phi(\vec{\lambda})\leq M$. Since $\vec{\lambda}\in\Omega$, we have $\lambda_{\op_{i,C}}\leq M$ for any $C\in\mathcal{A}(i,\mathcal{C})$ and 
\begin{align*}
&\sum_{C\in\mathcal{A}(i,\mathcal{C})}
\frac{R_C\alpha_{C,i}\alpha_{C,\op_{i,C}}\lambda_{\op_{i,C}}}
{(\alpha_{C,\op_{i,C}}M+\alpha_{C,i}\lambda_{\op_{i,C}})^2}\\
\leq&\sum_{C\in\mathcal{A}(i,\mathcal{C})}
\frac{R_C\alpha_{C,i}\alpha_{C,\op_{i,C}}M}{(\alpha_{C,\op_{i,C}}M)^2}\\
=&\frac1{M}\sum_{C\in\mathcal{A}(i,\mathcal{C})}
\frac{R_C\alpha_{C,i}}{\alpha_{C,\op_{i,C}}}\\
\leq& T_i.
\end{align*}
Meanwhile, we also have 
\begin{align*}
    & \frac{M}{\alpha_{C,i}}+\frac{\lambda_{\op_{i,C}}}{\alpha_{C,\op_{i,C}}}\geq \frac{M}{\alpha_{C,i}}\geq \frac{1}{\alpha_{C,i}}\frac{1}{T_i}\frac{R_C\alpha_{C,i}}{\alpha_{C,\op_{i,C}}}\\
&= \frac{R_C}{T_i\alpha_{C,\op_{i,C}}}>\frac{R_C}{(\alpha_{C,i}+\alpha_{C,\op_{i,C}})\overline{T}},
\end{align*} 
for any contest $C\in\mathcal{A}(i,\mathcal{C})$. By the definition of $\phi$, we have that $\phi_i(\vec{\lambda})\leq M$. Therefore, $\phi$ maps $\Omega$ into $\Omega$.

Next we prove that $\phi$ is a continuous on $\Omega$, i.e., for any contestant $i$, we prove that $\phi_i$ is continuous on $\Omega$. Define 
\begin{align*}
& \underline{\lambda}_i(\vec{\lambda}_{-i})=\max\{0,\\
& \max_{C\in\mathcal{A}(i,\mathcal{C})}(\frac{R_C}{(\alpha_{C,i}+\alpha_{C,\op_{i,C}})\overline{T}}-\frac{\lambda_{\op_{i,C}}}{\alpha_{C,\op_{i,C}}}
)\alpha_{C,i}\}.
\end{align*}
Let $h_i(\hat{\lambda_i},\vec{\lambda})$ denote $\sum_{C\in\mathcal{A}(i,\mathcal{C})}
\frac{R_C\alpha_{C,i}\alpha_{C,\op_{i,C}}\lambda_{\op_{i,C}}}
{(\alpha_{C,\op_{i,C}}\hat{\lambda_i}+\alpha_{C,i}\lambda_{\op_{i,C}})^2}$.

Let $\Omega_1=\{\vec{\lambda}\in\Omega:h_i(\underline{\lambda}_i(\vec{\lambda}_{-i}),\vec{\lambda})\geq T_i\}$ and $\Omega_2=\{\vec{\lambda}\in\Omega:h_i(\underline{\lambda}_i(\vec{\lambda}_{-i}),\vec{\lambda})\leq T_i\}$. For any $\vec{\lambda}\in\Omega_1$, note that $\max_{C\in\mathcal{A}(i,\mathcal{C})}\lambda_{\op_{i,C}}>0$, otherwise $h_i(\underline{\lambda}_i(\vec{\lambda}_{-i}),\vec{\lambda})=0$. Therefore, $h_i(\hat{\lambda_i},\vec{\lambda})$ is strictly decreasing in $\hat{\lambda_i}$ for $\hat{\lambda_i}\in[\underline{\lambda}_i(\vec{\lambda}_{-i}),M]$. Recall that $h_i(M,\vec{\lambda})\leq T_i$, so there exists a unique $\tilde{\lambda_i}(\vec{\lambda})\in[\underline{\lambda}_i(\vec{\lambda}_{-i}),M]$ such that $h_i(\tilde{\lambda_i}(\vec{\lambda}),\vec{\lambda})=T_i$. By implicit function theorem, $\tilde{\lambda_i}(\vec{\lambda})$ is a continuous and differentiable function in $\vec{\lambda}$ for $\vec{\lambda}\in\Omega_1$.

Observe that for any $\vec{\lambda}\in\Omega_1\cap\Omega_2$, $\tilde{\lambda_i}(\vec{\lambda})=\underline{\lambda}_i(\vec{\lambda}_{-i})$. Also observe that $\Omega_1\cup\Omega_2=\Omega$, so $\phi_i(\vec{\lambda})$ can be written as
\begin{align*}
\phi_i(\vec{\lambda})=\begin{cases}
\tilde{\lambda_i}(\vec{\lambda}),&\text{if }\vec{\lambda}\in\Omega_1,\\
\underline{\lambda}_i(\vec{\lambda}_{-i}),&\text{if }\vec{\lambda}\in\Omega_2.\\
\end{cases}
\end{align*}
Therefore, $\phi_i(\vec{\lambda})$ is continuous on $\Omega$.

Since $\phi$ is a continuous mapping on $\Omega$ and $\Omega$ is convex and compact, by Brouwer's fixed point theorem, there exists a fixed point $\vec{\lambda}^*\in\Omega$ such that $\phi(\vec{\lambda}^*)=\vec{\lambda}^*$.

Now we prove that $\vec{\lambda}^*$ satisfies the conditions of \Cref{lemma:dual-solution-condition}. For Statement 1 of \Cref{lemma:dual-solution-condition}, since $\phi(\vec{\lambda}^*)=\vec{\lambda}^*$, we know that for any contestant $i$ and any contest $C\in\mathcal{A}(i,\mathcal{C})$, $\frac{\lambda_i^*}{\alpha_{C,i}}+\frac{\lambda_{\op_{i,C}}^*}{\alpha_{C,\op_{i,C}}}
\geq\frac{R_C}{(\alpha_{C,i}+\alpha_{C,\op_{i,C}})\overline{T}}$, which means that $\lambda_i^*+\lambda_{\op_{i,C}}^*>0$ and $\vec{\lambda}^*$ is valid.

For Statement 2 and Statement 3 of \Cref{lemma:dual-solution-condition}, since $\phi(\vec{\lambda}^*)=$ $\vec{\lambda}^*$, we know that for any contestant $i$, $\sum_{C\in\mathcal{A}(i,\mathcal{C})}
\frac{R_C\alpha_{C,i}\alpha_{C,\op_{i,C}}\lambda_{\op_{i,C}}^*}
{(\alpha_{C,\op_{i,C}}\lambda_i^*+\alpha_{C,i}\lambda_{\op_{i,C}}^*)^2}
\leq T_i$, i.e., $\hat{T}_i(\lambda^*)\leq T_i$. It remains to prove that $\hat{T}_i(\lambda^*)=T_i$ if $\lambda_i^*>0$.

For any contestant $i$ with $\lambda_i^*>0$, since $\lambda_i^*=\phi_i(\vec{\lambda}^*)$, we know that either $\sum_{C\in\mathcal{A}(i,\mathcal{C})}
\frac{R_C\alpha_{C,i}\alpha_{C,\op_{i,C}}\lambda_{\op_{i,C}}^*}
{(\alpha_{C,\op_{i,C}}\lambda_i^*+\alpha_{C,i}\lambda_{\op_{i,C}}^*)^2}=T_i$, or there exists a contest $C\in\mathcal{A}(i,\mathcal{C})$, such that $\frac{\lambda_i^*}{\alpha_{C,i}}+\frac{\lambda_{\op_{i,C}}^*}{\alpha_{C,\op_{i,C}}}
=\frac{R_C}{(\alpha_{C,i}+\alpha_{C,\op_{i,C}})\overline{T}}$.
Suppose for contradiction that there exists $C\in\mathcal{A}(i,\mathcal{C})$, such that $\frac{\lambda_i^*}{\alpha_{C,i}}+\frac{\lambda_{\op_{i,C}}^*}{\alpha_{C,\op_{i,C}}}
=\frac{R_C}{(\alpha_{C,i}+\alpha_{C,\op_{i,C}})\overline{T}}$.

If $\frac{\lambda_i^*}{\alpha_{C,i}}\geq \frac{R_C}{2(\alpha_{C,i}+\alpha_{C,\op_{i,C}})\overline{T}}$, we can calculate that 
\begin{align*}
&\hat{T}_{\op_{i,C}}(\vec{\lambda}^*)\\
\geq&\hat{x}_{\op_{i,C},C}(\vec{\lambda}^*)\\
=&\frac{R_C\alpha_{C,i}\alpha_{C,\op_{i,C}}\lambda_i^*}
{(\alpha_{C,\op_{i,C}}\lambda_i^*+\alpha_{C,i}\lambda_{\op_{i,C}}^*)^2}\\
=&\frac{R_C\lambda_i^*}
{\alpha_{C,i}\alpha_{C,\op_{i,C}}(\frac{R_C}{(\alpha_{C,i}+\alpha_{C,\op_{i,C}})\overline{T}})^2}\\
=&\frac{(\alpha_{C,i}+\alpha_{C,\op_{i,C}})^2\overline{T}^2\lambda_i^*}
{\alpha_{C,i}\alpha_{C,\op_{i,C}}R_C}\\
\geq&\frac{(\alpha_{C,i}+\alpha_{C,\op_{i,C}})^2\overline{T}^2R_C}
{\alpha_{C,\op_{i,C}}R_C\cdot 2(\alpha_{C,i}+\alpha_{C,\op_{i,C}})\overline{T}}\\
>&\frac{\overline{T}}{2}\geq T_{\op_{i,C}}.
\end{align*}
This contradicts with that $\hat{T}_{\op_{i,C}}(\vec{\lambda}^*)\leq T_{\op_{i,C}}$. 

If $\frac{\lambda_i^*}{\alpha_{C,i}}\leq \frac{R_C}{2(\alpha_{C,i}+\alpha_{C,\op_{i,C}})\overline{T}}$, on the opposite side, we have $\frac{\lambda_{\op_{i,C}}^*}{\alpha_{C,\op_{i,C}}}\geq \frac{R_C}{2(\alpha_{C,i}+\alpha_{C,\op_{i,C}})\overline{T}}$, and similarly we can obtain $\hat{T}_{i}(\vec{\lambda}^*)\geq\hat{x}_{i,C}(\vec{\lambda}^*)>T_i$, which is also a contradiction.

Since both cases lead to contradiction, we know that for any  $C\in\mathcal{A}(i,\mathcal{C})$, $\frac{\lambda_i^*}{\alpha_{C,i}}+\frac{\lambda_{\op_{i,C}}^*}{\alpha_{C,\op_{i,C}}}
>\frac{R_C}{(\alpha_{C,i}+\alpha_{C,\op_{i,C}})\overline{T}}$, and therefore $\hat{T}_i(\vec{\lambda}^*)=\sum_{C\in\mathcal{A}(i,\mathcal{C})}
\frac{R_C\alpha_{C,i}\alpha_{C,\op_{i,C}}\lambda_{\op_{i,C}}^*}
{(\alpha_{C,\op_{i,C}}\lambda_i^*+\alpha_{C,i}\lambda_{\op_{i,C}}^*)^2}=T_i$ must holds.

By \Cref{lemma:dual-solution-condition}, $\vec{\lambda}^*$ is an equilibrium multiplier vector. This completes the proof.
\end{proof}

\subsection{Proof of \Cref{lemma:monotone}}
\CEMonotoneProperty*

\begin{proof}
For any valid $\vec{\lambda}$ and any contest $C\in\mathcal{C}$, suppose $S_C=\{i_1,i_2\}$. With some calculations, we have:
\begin{align*}
\frac{\partial \hat{x}_{i_1,C}(\vec{\lambda})}{\partial \lambda_{i_1}}&=R_C\frac{-2\frac{\lambda_{i_2}}{\alpha_{C,i_2}}}{\alpha_{C,i_1}^2(\frac{\lambda_{i_1}}{\alpha_{C,i_1}}+\frac{\lambda_{i_2}}{\alpha_{C,i_2}})^3},\\
\frac{\partial \hat{x}_{i_1,C}(\vec{\lambda})}{\partial \lambda_{i_2}}&=R_C\frac{\frac{\lambda_{i_1}}{\alpha_{C,i_1}}-\frac{\lambda_{i_2}}{\alpha_{C,i_2}}}{\alpha_{C,i_1}\alpha_{C,i_2}(\frac{\lambda_{i_1}}{\alpha_{C,i_1}}+\frac{\lambda_{i_2}}{\alpha_{C,i_2}})^3}.
\end{align*}
Therefore, we can calculate the Jacobian matrix $\mathrm{J}(\vec{\lambda})=\frac{\partial \hat{T}(\vec{\lambda})}{\partial \vec{\lambda}}$, where $\mathrm{J}_{i,k}(\vec{\lambda})=\frac{\partial\hat{T}_i(\vec{\lambda})}{\partial\lambda_k}$. For any contestants $i_1,i_2$ such that $i_1\neq i_2$, we have
\begin{align*}
\mathrm{J}_{i_1,i_2}(\vec{\lambda})=&\sum_{C\in\mathcal{C}:S_C=\{i_1,i_2\}}\frac{\partial \hat{x}_{i_1,C}(\vec{\lambda})}{\partial \lambda_{i_2}}\\
=&\sum_{C\in\mathcal{C}:S_C=\{i_1,i_2\}}R_C\frac{\frac{\lambda_{i_1}}{\alpha_{C,i_1}}-\frac{\lambda_{i_2}}{\alpha_{C,i_2}}}{\alpha_{C,i_1}\alpha_{C,i_2}(\frac{\lambda_{i_1}}{\alpha_{C,i_1}}+\frac{\lambda_{i_2}}{\alpha_{C,i_2}})^3}.
\end{align*}
For any contestant $i$, we get
\begin{align*}
\mathrm{J}_{i,i}(\vec{\lambda})=&\sum_{C\in\mathcal{A}(i,\mathcal{C})}\frac{\partial \hat{x}_{i,C}(\vec{\lambda})}{\partial \lambda_{\op_{i,C}}}\\
=&\sum_{C\in\mathcal{A}(i,\mathcal{C})}R_C\frac{-2\frac{\lambda_{\op_{i,C}}}{\alpha_{C,\op_{i,C}}}}{\alpha_{C,i}^2(\frac{\lambda_{i}}{\alpha_{C,i}}+\frac{\lambda_{\op_{i,C}}}{\alpha_{C,\op_{i,C}}})^3}.
\end{align*}

Observe that for any valid $\vec{\lambda}$, $\mathrm{J}_{i,i}(\vec{\lambda})\leq 0$, where the strict inequality holds if there exists $C\in\mathcal{A}(i,\mathcal{C})$ such that $\lambda_{\op_{i,C}}>0$. Also observe that for any $i_1,i_2\in[n]$ such that $i_1\neq i_2$, it holds that $\mathrm{J}_{i_1,i_2}(\vec{\lambda})+\mathrm{J}_{i_2,i_1}(\vec{\lambda})=0$.

For any valid multiplier vectors $\vec{\lambda},\vec{\lambda}'$, let $v=\vec{\lambda}'-\vec{\lambda}$. Note that for any $t\in[0,1]$, $\vec{\lambda}+tv=(1-t)\vec{\lambda}+t\vec{\lambda}'$ is also a valid multiplier vector. Define $h(t)=\sum_{i=1}^nv_i\hat{T}_i(\vec{\lambda}+tv)$, for any $t\in[0,1]$, we have $h'(t)=\sum_{i=1}^nv_i\frac{\partial\hat{T}_i(\vec{\lambda}+tv)}{\partial t}=\sum_{i=1}^nv_i\sum_{k=1}^n\mathrm{J}_{i,k}(\vec{\lambda}+tv)v_k=\sum_{i=1}^n\mathrm{J}_{i,i}(\vec{\lambda}+tv)v_i^2\leq0$. Therefore, 
$\sum_{i=1}^n(\lambda_i'-\lambda_i)(\hat{T}_i(\vec{\lambda}')-\hat{T}_i(\vec{\lambda}))=h(1)-h(0)\leq 0$.

Moreover, when there exists a contestant $i$ such that $\lambda_i'\neq\lambda_i$, and $\max_{C\in\mathcal{A}(i,\mathcal{C})}\max\{\lambda_{\op_{i,C}},\lambda_{\op_{i,C}}'\}>0$, we know that $v_i^2>0$, and that for any $t\in(0,1)$, $\mathrm{J}_{i,i}(\vec{\lambda}+tv)<0$. It follows that $h'(t)\leq \mathrm{J}_{i,i}(\vec{\lambda}+tv)v_i^2<0$. Therefore, $\sum_{i=1}^n(\lambda_i'-\lambda_i)(\hat{T}_i(\vec{\lambda}')-\hat{T}_i(\vec{\lambda}))=h(1)-h(0)<0$.
\end{proof}

\subsection{Proof of \Cref{thm:unique-contestant-equilibrium-multiplier}}
\CEthmUniquenessofEMV*

\begin{proof}
We prove this theorem by contradiction. Suppose that there exists two distinct equilibrium multiplier vectors $\vec{\lambda}$ and $\vec{\lambda}'$.

For any contestant $i$, if $\lambda_i'>\lambda_i$, by \Cref{lemma:dual-solution-condition}, we have $\hat{T}_i(\vec{\lambda})\leq T_i$ and $\hat{T}_i(\vec{\lambda}')=T_i$. Therefore $(\lambda_i'-\lambda_i)(\hat{T}_i(\vec{\lambda}')-\hat{T}_i(\vec{\lambda}))\geq0$.
If $\lambda_i'<\lambda_i$, similarly we know $(\lambda_i'-\lambda_i)(\hat{T}_i(\vec{\lambda}')-\hat{T}_i(\vec{\lambda}))\geq0$. And if $\lambda_i'=\lambda_i$,  we obtain $(\lambda_i'-\lambda_i)(\hat{T}_i(\vec{\lambda}')-\hat{T}_i(\vec{\lambda}))=0$. Therefore, it always holds that $(\lambda_i'-\lambda_i)(\hat{T}_i(\vec{\lambda}')-\hat{T}_i(\vec{\lambda}))\geq0$, and  we get $\sum_{i\in[n]}(\lambda_i'-\lambda_i)(\hat{T}_i(\vec{\lambda}')-\hat{T}_i(\vec{\lambda}))\geq0$.

However, since $\lambda_i'\neq\lambda_i$, there exists a contestant $i$ such that $\lambda_i'\neq\lambda_i$.

If $\lambda_i=0$, take an arbitrary contest $C\in\mathcal{A}(i,\mathcal{C})$ and we have $\lambda_{\op_{i,C}}>0$,  since $\vec{\lambda}$ is valid.

If $\lambda_i>0$, we know $\hat{T}_i(\vec{\lambda}))=T_i>0$. Recall that $\hat{T}_i(\vec{\lambda}))=\sum_{C\in\mathcal{A}(i,\mathcal{C})}
\frac{R_C\alpha_{C,i}\alpha_{C,\op_{i,C}}\lambda_{\op_{i,C}}}
{(\alpha_{C,\op_{i,C}}\lambda_i+\alpha_{C,i}\lambda_{\op_{i,C}})^2}$. Therefore, there exists a contest $C\in\mathcal{A}(i,\mathcal{C})$ such that $\lambda_{\op_{i,C}}>0$.

In summary, we know that when $\lambda_i'\neq\lambda_i$, it holds that $\max_{C\in\mathcal{A}(i,\mathcal{C})}\lambda_{\op_{i,C}}>0$. By \Cref{lemma:monotone}, we get $\sum_{i\in[n]}(\lambda_i'-\lambda_i)(\hat{T}_i(\vec{\lambda}')-\hat{T}_i(\vec{\lambda}))<0$, which is a contradiction.
\end{proof}

\subsection{Proof of \Cref{thm:contestant-equilibrium-algorithm}}

We first show the following technical lemma to help us design the algorithm.
\begin{lemma}
\label{lemma:addtermcontestantequilibrium}
When $\mathcal{C}$ is given, for any $\vec{a}=(a_i)_{i\in[n]}\in\mathbb{R}_{\geq 0}^n$, there exists a unique multiplier vector $\vec{\lambda}\in\mathbb{R}_{\geq 0}^n$, satisfying that:
\begin{enumerate}
    \item $\vec{\lambda}$ is valid, and for any contestant $i\in[n]$ with $a_i>0$, $\lambda_i>0$;
    \item For any contestant $i$ with $\lambda_i>0$, $T_i=\hat{T}_i(\vec{\lambda})+\frac{a_i}{\lambda_i}$;
    \item For any contestant $i$ with $\lambda_i=0$, $T_i\geq\hat{T}_i(\vec{\lambda})$.
\end{enumerate}
\end{lemma}
\begin{proof}
We prove this lemma by reducing finding such a multiplier to finding an EMV under another instance. 

Since, for any contestant $i\in[n]$ such that $\mathcal{A}(i,\mathcal{C})=\emptyset$ and $a_i=0$, we can always take $\lambda_i=0$, we assume that it holds for any $i\in[n]$ that $\mathcal{A}(i,\mathcal{C})\neq \emptyset$ or $a_i>0$.

We construct a new instance of the game of contestants in PLCCG, in which there are $n'=2n$ contestants, where each contestant $i\in[n]$ in the original instance is copied, and corresponds to the contestants $2i-1$ and $2i$ in the new instance, whose total efforts are $T'_{2i-1}=T'_{2i}=T_i$. The set of contests $\mathcal{C}'$ is constructed as follows:
\begin{enumerate}
    \item For each original contest $C\in\mathcal{C}$, let $S_C=\{i_1,i_2\}$ and construct four contests $C'_{C,2i_1-1,2i_2-1},C'_{C,2i_1-1,2i_2},$ $C'_{C,2i_1,2i_2-1},C'_{C,2i_1,2i_2}$ in the new instance, such that for $k_1\in\{2i_1-1,2i_1\}$ and $k_2\in\{2i_2-2,2i_2\}$, $C'_{C,k_1,k_2}$ is defined with $S_{C'_{C,k_1,k_2}}=\{k_1,k_2\}$, $R_{C'_{C,k_1,k_2}}=\frac12 R_C$, and $\alpha_{C'_{C,k_1,k_2},k_1}=\alpha_{C,i_1}$, $\alpha_{C'_{C,k_1,k_2},k_2}=\alpha_{C,i_2}$. Let $\mathcal{C}^{\prime,copy}=\{C'_{C,2i_1-1,2i_2-1},C'_{C,2i_1-1,2i_2},$ $C'_{C,2i_1,2i_2-1},C'_{C,2i_1,2i_2}:C\in\mathcal{C},S_C=\{i_1,i_2\}\}$.
    \item For each original contestant $i\in[n]$ such that $a_i>0$, construct a contest $C'_{self,i}$ in the new instance, with $S_{C'_{self,i}}=\{2i-1,2i\}$, $R_{C'_{self,i}}=4a_i$, and $\alpha_{C'_{self,i},2i-1}=\alpha_{C'_{self,i},2i}=1$. Let $\mathcal{C}^{\prime,self}=\{C'_{self,i}:i\in[n]\land a_i>0\}$.
    \item The set of contests in the new instance is $\mathcal{C}'=\mathcal{C}^{\prime,copy}\cup\mathcal{C}^{\prime,self}$.
\end{enumerate}

By \Cref{thm:contestant-equilibrium-existence} and \Cref{thm:unique-contestant-equilibrium-multiplier}, we know that there is a unique EMV under the new intance, denoted by $\vec{\lambda}^*\in\mathbb{R}_{\geq0}^{2n}$.
For any contest $C\in\mathcal{C}'$ and any contestant $i\in S_C$, define $\hat{x}'_{i,C}(\vec{\lambda}^*)=\frac{
R_C\alpha_{C,i}\alpha_{C,\op_{i,C}}\lambda^*_{\op_{i,C}}
}{
(\alpha_{C,\op_{i,C}}\lambda^*_i+\alpha_{C,i}\lambda^*_{\op_{i,C}})^2
}$, the same as $\hat{x}_{i,C}$. We also define $\hat{T}'_i(\vec{\lambda}^*)=\sum_{i\in\mathcal{A}(i,\mathcal{C}')}\hat{x}'_{i,C}(\vec{\lambda}^*)$.
We know that $\vec{\lambda}^*$ satisfies the conditions in \Cref{lemma:dual-solution-condition} with $\hat{T}'$.

Now we show that, for any contestant $i$, it holds that $\lambda^*_{2i-1}=\lambda^*_{2i}$. Suppose for contradiction that there is some contestant $i$, that $\lambda^*_{2i-1}\neq\lambda^*_{2i}$. Without loss of generality we can assume $\lambda^*_{2i-1}<\lambda^*_{2i}$. 
We discuss in two cases:

(a) If $\lambda^*_{2i-1}>0$, by Statement 2 of \Cref{lemma:dual-solution-condition}, we have $\hat{T}'_{2i-1}(\vec{\lambda}^*)=T'_{2i-1}=T_i$ and $\hat{T}'_{2i}(\vec{\lambda}^*)=T'_{2i}=T_i$. By the construction of $\mathcal{C}'$, 
\begin{align*}
&\hat{T}'_{2i-1}(\vec{\lambda}^*)=\sum_{C\in\mathcal{A}(i,\mathcal{C})}\frac{ \frac12 R_C\alpha_{C,i}\alpha_{C,\op_{i,C}}\lambda^*_{2\op_{i,C}-1}
}{
(\alpha_{C,\op_{i,C}}\lambda^*_{2i-1}+\alpha_{C,i}\lambda^*_{2\op_{i,C}-1})^2}
+\\
&\sum_{C\in\mathcal{A}(i,\mathcal{C})}\frac{ \frac12 
R_C\alpha_{C,i}\alpha_{C,\op_{i,C}}\lambda^*_{2\op_{i,C}}
}{
(\alpha_{C,\op_{i,C}}\lambda^*_{2i-1}+\alpha_{C,i}\lambda^*_{2\op_{i,C}})^2}+\frac{4a_i\lambda^*_{2i}}{(\lambda^*_{2i-1}+\lambda^*_{2i})^2}
\end{align*}
and 
\begin{align*}
&\hat{T}'_{2i}(\vec{\lambda}^*)=\sum_{C\in\mathcal{A}(i,\mathcal{C})}\frac{ \frac12 
R_C\alpha_{C,i}\alpha_{C,\op_{i,C}}\lambda^*_{2\op_{i,C}-1}
}{
(\alpha_{C,\op_{i,C}}\lambda^*_{2i}+\alpha_{C,i}\lambda^*_{2\op_{i,C}-1})^2}
+\\
&\sum_{C\in\mathcal{A}(i,\mathcal{C})}\frac{ \frac12 
R_C\alpha_{C,i}\alpha_{C,\op_{i,C}}\lambda^*_{2\op_{i,C}}
}{
(\alpha_{C,\op_{i,C}}\lambda^*_{2i}+\alpha_{C,i}\lambda^*_{2\op_{i,C}})^2}+\frac{4a_i\lambda^*_{2i-1}}{(\lambda^*_{2i-1}+\lambda^*_{2i})^2}.
\end{align*}
Since $\lambda^*_{2i-1}<\lambda^*_{2i}$, we have $\hat{T}'_{2i-1}(\vec{\lambda}^*)>\hat{T}'_{2i}(\vec{\lambda}^*)$, which contradicts with that $\hat{T}'_{2i-1}(\vec{\lambda}^*)=\hat{T}'_{2i}(\vec{\lambda}^*)=T_i$.

(b) If $\lambda^*_{2i-1}=0$, similarly we have $\hat{T}'_{2i-1}(\vec{\lambda}^*)>\hat{T}'_{2i}(\vec{\lambda}^*)$. However, by \Cref{lemma:dual-solution-condition}, we have $\hat{T}'_{2i-1}(\vec{\lambda}^*)\leq T_i=\hat{T}'_{2i}(\vec{\lambda}^*)$, which is a contradiction.

Now we can construct $\vec{\lambda}\in\mathbb{R}_{\geq 0}^n$ such that for each contestant $i$, $\lambda_i=\lambda^*_{2i-1}=\lambda^*_{2i}$. We show that $\vec{\lambda}$ satisfies the requirements in this lemma. For any contestant $i$, we discuss in two cases:

(a) $a_i>0$. Since $\vec{\lambda}^*$ is valid under $\mathcal{C}'$, we know that $0<\sum_{k\in C'_{self,i}}\lambda^*_k=\lambda^*_{2i-1}+\lambda^*_{2i}$, so $\lambda_i>0$. We can calculate that $\hat{T}'_{2i-1}(\vec{\lambda}^*)=\sum_{C\in\mathcal{A}(i,\mathcal{C})}\frac{ \frac12 
R_C\alpha_{C,i}\alpha_{C,\op_{i,C}}\lambda_{\op_{i,C}}
}{
(\alpha_{C,\op_{i,C}}\lambda_{i}+\alpha_{C,i}\lambda^*_{\op_{i,C}})^2}+
\sum_{C\in\mathcal{A}(i,\mathcal{C})}\frac{ \frac12 
R_C\alpha_{C,i}\alpha_{C,\op_{i,C}}\lambda_{\op_{i,C}}
}{
(\alpha_{C,\op_{i,C}}\lambda_{i}+\alpha_{C,i}\lambda^*_{\op_{i,C}})^2}+\frac{4a_i\lambda_i}{\lambda_i^2}=\hat{T}_i(\vec{\lambda})+\frac{a_i}{\lambda_i}$. Since $\hat{T}'_{2i-1}(\vec{\lambda}^*)=T'_{2i-1}=T_i$, the requirements are satisfied.

(b) $a_i=0$. We can calculate that $\hat{T}'_{2i-1}(\vec{\lambda}^*)=\hat{T}_i(\vec{\lambda})$.  Since $\hat{T}'_{2i-1}(\vec{\lambda}^*)\leq T'_{2i-1}=T_i$, the requirements are satisfied.

In summary, $\vec{\lambda}$ satisfies the requirements in this lemma.

Finally we show such $\vec{\lambda}$ is unique. For any $\vec{\lambda}$ satisfying the requirements, it is not hard to see that if we construct $\vec{\lambda}^*$ such that for any contestant $i$, $\lambda^*_{2i-1}=\lambda^*_{2i}=\lambda_i$, it holds that $\vec{\lambda}^*$ is an EMV in the above new instance. Since the EMV is unique, we have that $\vec{\lambda}$ is unique.
\end{proof}

Now, we give our algorithm which can output an $\epsilon$-approximate contestant equilibrium.

\begin{algorithm}[htbp!]
\caption{Algorithm to Compute $\epsilon$-Contestant Equilibrium}\label{alg:Alg1}
\begin{algorithmic}[1]
\STATE \textbf{Input} $n, T_1,\cdots,T_n,\mathcal{C}$, precision $\epsilon$, step size $\gamma$;
\STATE $t\gets 0$;
\STATE $\epsilon'\gets \frac14\epsilon$;
\STATE $\lambda_i^{(1)}\gets 1$ for $i=1,\cdots, n$;
\STATE $a_i^{(1)}\gets \epsilon'^2T_i\min_{C\in\mathcal{A}(i,\mathcal{C})}\frac{R_C}{T_i+\frac{\alpha_{C,\op_{i,C}}}{\alpha_{C,i}}T_{\op_{i,C}}}$ for $i=1,\cdots, n$;
\REPEAT
    \STATE $t\gets t+1$;
    \STATE $Z_i^{(t)}\gets \hat{T}_i(\vec{\lambda}^{(t)})-T_i$ for $i=1,\cdots, n$;
    \STATE $\lambda_i^{(t+1)}\gets \lambda_i^{(t)}+a_i/\lambda_i^{(t)}+\gamma Z_i^{(t)}$ for $i=1,\cdots, n$;
\UNTIL $\max_{i\in[n]}|\frac{Z_i^{(t)}}{T_i}|\leq \epsilon'$
\STATE $x_{i,C}\gets \hat{x}_{i,C}((1+\epsilon')\vec{\lambda}^{(t)})$ for all $C\in\mathcal{C}$ and $i\in S_C$;
\STATE \textbf{Output} $\vec{x}$;

\end{algorithmic}
\end{algorithm}
\CEthmAlgorithm*
\begin{proof}
Without loss of generality, we can assume that for any contestant $i$, $\mathcal{A}(i,\mathcal{C})\neq\emptyset$.
We firstly define some notions. 
Let $\overline{T}=\max_{i\in[n]}T_i$ and $\epsilon'=\frac14\epsilon$.
For each $i\in[n]$, define $a_i=\epsilon'^2T_i\min_{C\in\mathcal{A}(i,\mathcal{C})}\frac{R_C}{T_i+\frac{\alpha_{C,\op_{i,C}}}{\alpha_{C,i}}T_{\op_{i,C}}}$.
Take $\gamma<\min_{i\in[n]}\frac{L_i}{T_i}$. 

The algorithm is presented in \Cref{alg:Alg1}.

Now we prove that \Cref{alg:Alg1} finds an $\epsilon$-approximate contestant equilibrium in polynomial time.

We prove the following claims:
\begin{description}
    \item [Claim 1] For any $\vec{\lambda}\in\mathbb{R}_{>0}^n$ such that $\hat{T_i}(\vec{\lambda})+\frac{a_i}{\lambda_i}\in[(1-2\epsilon')T_i,T_i]$ for all $i\in[n]$, construct $\vec{x}$ such that $x_{i,C}=\vec{x}_{i,C}(\vec{\lambda})$, then $x_{i,C}$ is an $\epsilon$-contestant equilibrium.
    \item [Claim 2] Define $L_i=\min\{\frac{a_i}{2T_i},1\}$, and $U_i=\max\{\frac{\sum_{\mathcal{A}(i,\mathcal{C})}\frac12R_C+2a_i}{T_i},1\}$, and define $\Omega^+=\{\vec{\lambda}\in\mathbb{R}_{>0}^n:\forall i\in[n],~L_i\leq \lambda_i\leq U_i\}$. In \Cref{alg:Alg1}, when $\gamma\leq \min\{\frac{L_i}{T_i},\frac{U_i}{\sum_{C\in\mathcal{A}(i,\mathcal{C})}\frac{R_CT_i}{a_i}+2T_i}\}$, if $\vec{\lambda}^{(t)}\in\Omega^+$, then $\vec{\lambda}^{(t+1)}\in\Omega^+$.
    \item [Claim 3] Define $L=\min_{i\in[n]} L_i$ and $U=\max_{i\in[n]} U_i$, and $\overalpha=\max_{C\in\mathcal{C}}\frac{\max_{i\in S_C}\alpha_{C,i}}{\min_{i\in S_C}\alpha_{C,i}}$. Define $\hat{T}^+_i(\vec{\lambda})=\hat{T}_i(\vec{\lambda})+\frac{a_i}{\lambda_i}$. For any $\vec{\lambda}\in\Omega^+$, for any contestants $i,k$ that $i\neq k$, $|\frac{\partial{\hat{T}^+_i(\vec{\lambda})}}{\partial\lambda_k}|\leq \sum_{C\in\mathcal{C}:S_C=\{i,k\}}\frac{R_C U\overalpha^3}{L^3}=:\rho_{i,k}$, and $\frac{\partial{\hat{T}^+_i(\vec{\lambda})}}{\partial\lambda_i}\leq -2\sum_{C\in\mathcal{A}(i,\mathcal{C})}\frac{R_CL}{\overalpha^3U^3}-\frac{a_i}{U^2}=:\mu_i$.
    \item [Claim 4] In \Cref{alg:Alg1}, when the step size $\gamma$ is small enough, if $\vec{\lambda}^{(t)}\in\Omega^+$, and $\max_{i\in[n]}|\frac{Z_i^{(t)}}{T_i}|> \eta$, then $\sum_{i\in[n]}(Z_i^{(t+1)})^2\leq\sum_{i\in[n]}(Z_i^{(t)})^2-\frac12M_4\eta^2M_5\gamma$, where $M_4,M_5$ are defined later.
\end{description}

To prove Claim 1, for each contestant $i\in[n]$, we discuss in two cases:

(a) $\frac{a_i}{\lambda_i}>\epsilon' T_i$. In this case, we have $\lambda_i<\frac{a_i}{\epsilon' T_i}=\epsilon'\min_{C\in\mathcal{A}(i,\mathcal{C})}\frac{R_C}{T_i+\frac{\alpha_{C,\op_{i,C}}}{\alpha_{C,i}}T_{\op_{i,C}}}$.
For any $C\in\mathcal{A}(i,\mathcal{C})$, observe that $\frac{\hat{x}_{i,C}(\vec{\lambda})}{\alpha_{C,\op_{i,C}}}+\frac{\hat{x}_{\op_{i,C},C}(\vec{\lambda})}{\alpha_{C,i}}=\frac{R_C}{\alpha_{C,i}\lambda_{\op_{i,C}}+\alpha_{C,\op_{i,C}}\lambda_{i}}$, and by assumption we have $\hat{x}_{i,C}(\vec{\lambda})\leq \hat{T}_i(\vec{\lambda})\leq T_i$, and $\hat{x}_{\op_{i,C},C}(\vec{\lambda})\leq \hat{T}_{\op_{i,C}}(\vec{\lambda})\leq T_{\op_{i,C}}$, so $\frac{R_C}{\alpha_{C,i}\lambda_{\op_{i,C}}+\alpha_{C,\op_{i,C}}\lambda_{i}}\leq \frac{T_i}{\alpha_{C,\op_{i,C}}}+\frac{T_{\op_{i,C}}}{\alpha_{C,i}}$, and $\alpha_{C,i}\lambda_{\op_{i,C}}+\alpha_{C,\op_{i,C}}\lambda_{i}\geq \frac{R_C}{\frac{T_i}{\alpha_{C,\op_{i,C}}}+\frac{T_{\op_{i,C}}}{\alpha_{C,i}}}$. Therefore, we get $\hat{p}_{i,C}(\vec{\lambda})=\frac{\alpha_{C,i}\lambda_{\op_{i,C}}}{\alpha_{C,i}\lambda_{\op_{i,C}}+\alpha_{C,\op_{i,C}}\lambda_{i}}
=1-\frac{\alpha_{C,\op_{i,C}}\lambda_{i}}{\alpha_{C,i}\lambda_{\op_{i,C}}+\alpha_{C,\op_{i,C}}\lambda_{i}}
>1-\frac{\alpha_{C,\op_{i,C}}\epsilon'\frac{R_C}{T_i+\frac{\alpha_{C,\op_{i,C}}}{\alpha_{C,i}}T_{\op_{i,C}}}}{\frac{R_C}{\frac{T_i}{\alpha_{C,\op_{i,C}}}+\frac{T_{\op_{i,C}}}{\alpha_{C,i}}}}=1-\epsilon'$. This means that 
$u_i^{contestant}(\vec{\mathcal{C}},\vec{x})\geq (1-\epsilon')\sum_{C\in\mathcal{A}(i,\mathcal{C})}R_C\geq (1-\epsilon')u_i^{contestant}(\vec{\mathcal{C}},(x_i',\vec{x}_{-i}))$ for any $x_i'$.

(b) $\frac{a_i}{\lambda_i}\leq \epsilon' T_i$. In this case, observe that, by \Cref{lemma:bestresponse-lagrange}, we have $u_i^{contestant}(\vec{\mathcal{C}},(x_i,\vec{x}_{-i}))$ is equal to the optimal value of the following optimization:
\begin{align*}
\max_{x'_{i,C}\geq0 \text{ for }C\in\mathcal{A}(i,\mathcal{C})} \quad & \sum_{C\in\mathcal{A}(i,\mathcal{C})}R_C\cdot p_{i,C}(x'_i,\vec{x}_{-i}),\label{eq:contestant-bestresponse-optimization-problem}\\
\mbox{s.t.} \quad & \sum_{C\in \mathcal{A}(i,\mathcal{C})}x'_{i,C}\leq \hat{T}_i(\vec{\lambda}).\nonumber
\end{align*}
In other words, $x_i$ is the best response of contestant $i$ if we replace her total effort by $\hat{T}_i(\vec{\lambda})$. Observe that each $p_{i,C}(x'_i,\vec{x}_{-i})$ is concave in $x'_i$, combining this with $\hat{T}_i(\vec{\lambda})= (1-2\epsilon')T_i-\frac{a_i}{\lambda_i}\geq (1-3\epsilon')T_i$, we have $u_i^{contestant}(\vec{\mathcal{C}},(x_i,\vec{x}_{-i}))\geq (1-3\epsilon')\max_{x_i':\sum_{C\in \mathcal{A}(i,\mathcal{C})}x'_{i,C}\leq T_i}u_i^{contestant}(\vec{\mathcal{C}},(x_i',\vec{x}_{-i}))$.

In summary, we obtain that for any contestant $i\in[n]$ and any feasible strategy $x_i'$, $u_i^{contestant}(\vec{\mathcal{C}},\vec{x})\geq (1-\epsilon')\sum_{C\in\mathcal{A}(i,\mathcal{C})}R_C\geq (1-3\epsilon')u_i^{contestant}(\vec{\mathcal{C}},(x_i',\vec{x}_{-i}))\geq (1-\epsilon)u_i^{contestant}(\vec{\mathcal{C}},(x_i',\vec{x}_{-i}))$. Therefore $\vec{x}$ is an $\epsilon$-approximate contestant equilibrium.

For Claim 2, suppose $\vec{\lambda}^{(t)}\in\Omega^+$. For any $i\in[n]$, we prove that $\lambda_i^{(t+1)}\in[L_i,U_i]$.

For the lower bound, we discuss in two cases:

(a) If $\lambda_i^{(t)}\leq 2L_i$, we have $Z_i^{(t)}=\hat{T}_i(\vec{\lambda}^{(t)})+\frac{a_i}{\lambda_i^{(t)}}-T_i\geq \frac{a_i}{\lambda_i^{(t)}}-T_i\geq \frac{a_i}{2L_i}-T_i\geq 0$. Therefore, it holds that  $\lambda_i^{(t+1)}=\lambda_i^{(t)}+\gamma Z_i^{(t)}\geq \lambda_i^{(t)}\geq L_i$.

(b) If $\lambda_i^{(t)}> 2L_i$, observing that $Z_i^{(t)}\geq -T_i$, we have $\lambda_i^{(t+1)}=\lambda_i^{(t)}+\gamma Z_i^{(t)}\geq \lambda_i^{(t)}-\gamma T_i\geq \lambda_i^{(t)}-L_i>L_i$.

For the upper bound, we also discuss in two cases:

If $\lambda_i^{(t)}\geq \frac12U_i$, then we have $Z_i^{(t)}=\hat{T}_i(\vec{\lambda}^{(t)})+\frac{a_i}{\lambda_i^{(t)}}-T_i\leq \sum_{C\in\mathcal{A}(i,\mathcal{C})}\frac{R_C}{4\lambda_i^{(t)}}+\frac{a_i}{\lambda_i^{(t)}}-T_i\leq \frac{2}{U_i}(\sum_{C\in\mathcal{A}(i,\mathcal{C})}\frac{R_C}{4}+a_i)-T_i\leq 0$, so $\lambda_i^{(t+1)}=\lambda_i^{(t)}+\gamma Z_i^{(t)}\leq \lambda_i^{(t)}\leq U_i$.

If $\lambda_i^{(t)}<\frac12U_i$, by assumption $\lambda_i^{(t)}\geq L_i$, so we have $Z_i^{(t)}\leq \sum_{C\in\mathcal{A}(i,\mathcal{C})}\frac{R_C}{4\lambda_i^{(t)}}+\frac{a_i}{\lambda_i^{(t)}}-T_i\leq \frac{1}{L_i}(\sum_{C\in\mathcal{A}(i,\mathcal{C})}\frac{R_C}{4}+a_i)-T_i=\sum_{C\in\mathcal{A}(i,\mathcal{C})}\frac{R_CT_i}{2a_i}+T_i$. 
We have that $\gamma Z_i^{(t)}\leq \frac{U_i}{2}$.
Therefore we have $\lambda_i^{(t+1)}=\lambda_i^{(t)}+\gamma Z_i^{(t)}\leq \lambda_i^{(t)}+\frac{U_i}{2}<U_i$.

In summary, when $\vec{\lambda}^{(t)}\in\Omega^+$, we have $\vec{\lambda}^{(t+1)}\in\Omega^+$.

To prove claim 3, we calculate for any contestants $i,k$ such that $i\neq k$ the partial derivatives:
\begin{align*}
\frac{\partial{\hat{T}^+_i(\vec{\lambda})}}{\partial\lambda_k}=&\sum_{C\in\mathcal{C}:S_C=\{i,k\}}\frac{\partial \hat{x}_{i,C}(\vec{\lambda})}{\partial \lambda_{k}}\\
=&\sum_{C\in\mathcal{C}:S_C=\{i,k\}}R_C\frac{\frac{\lambda_{i}}{\alpha_{C,i}}-\frac{\lambda_{k}}{\alpha_{C,k}}}{\alpha_{C,i}\alpha_{C,k}(\frac{\lambda_{i}}{\alpha_{C,i}}+\frac{\lambda_{k}}{\alpha_{C,k}})^3}.
\end{align*}
We have $|\frac{\partial{\hat{T}^+_i(\vec{\lambda})}}{\partial\lambda_k}|\leq \sum_{C\in\mathcal{C}:S_C=\{i,k\}}\frac{R_C U\overalpha^3}{L^3}$. 
Also we can see that $\frac{\partial{\hat{T}^+_k(\vec{\lambda})}}{\partial\lambda_i}=-\frac{\partial{\hat{T}^+_i(\vec{\lambda})}}{\partial\lambda_k}$.

For any contestant $i$, we have
\begin{align*}
\frac{\partial{\hat{T}^+_i(\vec{\lambda})}}{\partial\lambda_i}=&\sum_{C\in\mathcal{A}(i,\mathcal{C})}\frac{\partial \hat{x}_{i,C}(\vec{\lambda})}{\partial \lambda_{\op_{i,C}}}+\frac{\partial \frac{a_i}{\lambda_i}}{\partial \lambda_i}\\
=&\sum_{C\in\mathcal{A}(i,\mathcal{C})}R_C\frac{-2\frac{\lambda_{\op_{i,C}}}{\alpha_{C,\op_{i,C}}}}{\alpha_{C,i}^2(\frac{\lambda_{i}}{\alpha_{C,i}}+\frac{\lambda_{\op_{i,C}}}{\alpha_{C,\op_{i,C}}})^3}-\frac{a_i}{\lambda_i^2}.
\end{align*}

Furthermore, we derive $\frac{\partial{\hat{T}^+_i(\vec{\lambda})}}{\partial\lambda_i}\leq -2\sum_{C\in\mathcal{A}(i,\mathcal{C})}\frac{R_CL}{\overalpha^3U^3}-\frac{a_i}{U^2}\leq -2\sum_{C\in\mathcal{A}(i,\mathcal{C})}\frac{R_CL}{\overalpha^3U^3}$.

To prove claim 4, for any $i\in[n]$ we can calculate that
\begin{align*}
&Z_i^{(t+1)}-Z_i^{(t)}\\
=&\hat{T}^+_i(\vec{\lambda}^{(t+1)})-\hat{T}^+_i(\vec{\lambda}^{(t)})\\
=&\hat{T}^+_i(\vec{\lambda}^{(t)}+\gamma \vec{Z}^{(t)})-\hat{T}^+_i(\vec{\lambda}^{(t)})
\end{align*}
Denote the Jacobian matrix of $\hat{T}^+(\vec{\lambda})$ as $J(\vec{\lambda})$ where $J_{i,k}(\vec{\lambda})=\frac{\partial{\hat{T}^+_i(\vec{\lambda})}}{\partial\lambda_k}$.
By Lagrange's mean value theorem, there is some $\xi_i\in[0,1]$ such that $Z_i^{(t+1)}-Z_i^{(t)}=\hat{T}^+_i(\vec{\lambda}^{(t)}+\gamma \vec{Z}^{(t)})-\hat{T}^+_i(\vec{\lambda}^{(t)})=\sum_{k\in[n]}J_{i,k}(\vec{\lambda}+\xi_i\gamma\vec{Z}^{(t)})\gamma Z_k^{(t)}$. Let $\vec{\lambda}_{\xi_i}$ denote $\vec{\lambda}+\xi_i\gamma\vec{Z}^{(t)}$, then $\vec{\lambda}_{\xi_i}\in\Omega^+$. With some calculation we have $|J_{i,k}(\vec{\lambda}_{\xi_i})-J_{i,k}(\vec{\lambda})|\leq(\sum_{C\in\mathcal{C}}R_C\frac{12U\overalpha^4}{L^4}+\sum_{i\in[n]}\frac{2a_i}{L^3})\sum_{k\in[n]}\gamma |Z_k^{(t)}|$ for all $i\in[n],k\in[n]$.

Recall that for any $i\in[n]$, for any $\vec{\lambda}^{(t)}\in\Omega^+$, we have $Z_i^{(t)}\leq \sum_{C\in\mathcal{A}(i,\mathcal{C})}\frac{R_CT_i}{2a_i}+T_i$, and $Z_i^{(t)}\geq-T_i$.

Define $M_1=\max_{i,k}\rho_{i,k},M_2=\sum_{C\in\mathcal{C}}R_C\frac{12U\overalpha^4}{L^4}+\sum_{i\in[n]}\frac{2a_i}{L^3}$, $M_3=\max_{i\in[n]}\sum_{C\in\mathcal{A}(i,\mathcal{C})}\frac{R_CT_i}{2a_i}+T_i$, $M_4=\min_{i\in[n]}-\mu_i$, $M_5=\sum_{i\in[n]}T_i$.
Now we calculate
\begin{align*}
    &\sum_{i\in[n]}(Z_i^{(t+1)})^2-\sum_{i\in[n]}(Z_i^{(t)})^2\\
=&\sum_{i\in[n]}(2Z_i^{(t)}(Z_i^{(t+1)}-Z_i^{(t)})+(Z_i^{(t+1)}-Z_i^{(t)})^2)\\
=&\sum_{i\in[n]}(2Z_i^{(t)}(Z_i^{(t+1)}-Z_i^{(t)})+(Z_i^{(t+1)}-Z_i^{(t)})^2)\\
=&\sum_{i\in[n]}(2Z_i^{(t)}\sum_{k\in[n]}J_{i,k}(\vec{\lambda}_{\xi_i})\gamma Z_k^{(t)}+(\sum_{k\in[n]}J_{i,k}(\vec{\lambda}_{\xi_i})\gamma Z_k^{(t)})^2)\\
=& \sum_{i\in[n]}(2Z_i^{(t)}\sum_{k\in[n]}J_{i,k}(\vec{\lambda}_{\xi_i})\gamma Z_k^{(t)}+(\sum_{k\in[n]}J_{i,k}(\vec{\lambda}_{\xi_i})\gamma Z_k^{(t)})^2)\\
\leq& \gamma\vec{Z}^{(t)T}J(\vec{\lambda}) \vec{Z}^{(t)}+\sum_{i\in[n]}(2Z_i^{(t)}M_1(\sum_{k\in[n]}\gamma |Z_k^{(t)}|)^2\\
&+(\sum_{k\in[n]}M_2\gamma Z_k^{(t)})^2)\\
=& \gamma\sum_{i\in[n]}\mu_i(Z_i^{(t)})^2+
\sum_{i\in[n]}(2Z_i^{(t)}M_1(\sum_{k\in[n]}\gamma |Z_k^{(t)}|)^2\\
&+(\sum_{k\in[n]}M_2\gamma Z_k^{(t)})^2)\\
\leq& -\gamma M_4\sum_{i\in[n]}(Z_i^{(t)})^2+
\gamma^2(2M_1+M_2)(\sum_{k\in[n]}Z_k^{(t)})^2\\
\leq& -\gamma M_4\eta^2M_5+
\gamma^2(2M_1+M_2)M_3^2
\end{align*}

When we take $\gamma\leq \frac{\eta^2M_4M_5}{2(2M_1+M_2)n^2M_3^2}$, it holds that \begin{align*}
    \sum_{i\in[n]}(Z_i^{(t+1)})^2-\sum_{i\in[n]}(Z_i^{(t)})^2\leq -\frac12M_4\eta^2M_5\gamma.
    \end{align*}

By Claim 4, we know that when $\gamma$ is taken small enough but polynomial in $\epsilon$ and the input size, $\sum_{i\in[n]}(Z_i^{(t)})^2$ monotonely decreases by at least $-\frac12M_4\eta^2M_5\gamma$ after each iteration. Also, we have $\sum_{i\in[n]}(Z_i^{(1)})^2\leq n^2M_3^2$. Since all numbers are polynomial in $\epsilon$ and the input sizes, we know that the algorithm finds in polynomial time a $\vec{\lambda}$ such that $|\frac{\hat{T}_i(\vec{\lambda})-T_i}{T_i}|\leq \epsilon'$ holds for every $i\in[n]$.
It's not hard to see that for any $i\in[n]$, $\frac{\hat{T}_i((1+\epsilon)\vec{\lambda})}{T_i}\in[(1-\epsilon)/(1+\epsilon),1]$, and by claim 1 we know that $\vec{x}$ is an $\epsilon$-approximate contestant equilibrium.

\end{proof}

\section{Missing Proofs in \Cref{sec:indivisible}}

\subsection{Proof of \Cref{thm:ind-SPE-not-exist}}
\INDSPEnotexist*

\begin{proof}
Consider an instance in the indivisible prize model with $n=3$ contestants and $m=2$ designers, where $T_1=T_2=T_3=1$ and $B_1=B_2=1$. We prove that there is no designer equilibrium by contradiction. Suppose that $\vec{\mathcal{C}}$ is a designer equilibrium and $(\vec{\mathcal{C}},\vec{x})$ is an SPE. Let $\vec{\lambda}$ be the equilibrium multiplier vector of contestants under $\vec{\mathcal{C}}$.

We discuss in two cases:

(a) If $S_{C_1}=S_{C_2}$, i.e., the same participants, without loss of generality, assume $S_{C_1}=S_{C_2}=\{1,2\}$. We calculate the contestant equilibrium. 

Firstly, we show that it must hold that $\lambda_1>0$ and $\lambda_2>0$. Suppose for contradiction that $\lambda_1=0$. We know $\lambda_2>0$, which implies that $\hat{T}_2(\vec{\lambda})=T_i>0$ by \Cref{lemma:dual-solution-condition}. It contradicts with that $\hat{T}_2(\vec{\lambda})=\frac{R_{C_1}\alpha_{C_1,1}\alpha_{C_1,2}\lambda_1}{(\alpha_{C_1,1}\lambda_2+\alpha_{C_1,2}\lambda_1)^2}+\frac{R_{C_2}\alpha_{C_2,1}\alpha_{C_2,2}\lambda_1}{(\alpha_{C_2,1}\lambda_2+\alpha_{C_2,2}\lambda_1)^2}=0$. Similarly, we get $\lambda_2>0$.

By \Cref{lemma:dual-solution-condition}, we know that 
\begin{align*}
1 & =T_1=\hat{T}_1(\vec{\lambda})\\
&=\frac{R_{C_1}\alpha_{C_1,1}\alpha_{C_1,2}\lambda_2}{(\alpha_{C_1,1}\lambda_2+\alpha_{C_1,2}\lambda_1)^2}+\frac{R_{C_2}\alpha_{C_2,1}\alpha_{C_2,2}\lambda_2}{(\alpha_{C_2,1}\lambda_2+\alpha_{C_2,2}\lambda_1)^2}
\end{align*}
and 
\begin{align*}
    1&=T_2=\hat{T}_2(\vec{\lambda})\\
    &=\frac{R_{C_1}\alpha_{C_1,1}\alpha_{C_1,2}\lambda_1}{(\alpha_{C_1,1}\lambda_2+\alpha_{C_1,2}\lambda_1)^2}+\frac{R_{C_2}\alpha_{C_2,1}\alpha_{C_2,2}\lambda_1}{(\alpha_{C_2,1}\lambda_2+\alpha_{C_2,2}\lambda_1)^2}.
\end{align*} It follows that $\lambda_1=\lambda_2$. Substituting $\lambda_2$ by $\lambda_1$, we have 
$\lambda_1=\lambda_2=\frac{R_{C_1}\alpha_{C_1,1}\alpha_{C_1,2}}{(\alpha_{C_1,1}+\alpha_{C_1,2})^2}+\frac{R_{C_2}\alpha_{C_2,1}\alpha_{C_2,2}}{(\alpha_{C_2,1}+\alpha_{C_2,2})^2}$. 

For each designer $j=1,2$, let $h(C_j)$ denote $\frac{R_{C_j}\alpha_{C_j,1}\alpha_{C_j,2}}{(\alpha_{C_j,1}+\alpha_{C_j,2})^2}$, it holds that \begin{align*}
    x_{1,C_1}=x_{2,C_1}&=\frac{h(C_1)}{h(C_1)+h(C_2)},\\
    x_{1,C_2}=x_{2,C_2}&=\frac{h(C_2)}{h(C_1)+h(C_2)}.
\end{align*}
Note that $R_{C_j}\leq B_j=1$ and $\frac{\alpha_{C_2,1}\alpha_{C_2,2}}{(\alpha_{C_2,1}+\alpha_{C_2,2})^2}\leq \frac14$, which means that $h(C_j)\in(0,\frac14]$. Since $u^D_j(\vec{x})=\frac{2h(C_1)}{h(C_1)+h(C_2)}$, if there is a $h(C_j)<\frac14$ for some $j\in\{1,2\}$, designer $j$ can improve her utility by increasing $h(C_j)$, which will contradict with that $\vec{\mathcal{C}}$ is a designer equilibrium. Thus, we know $h(C_1)=h(C_2)=\frac14$, i.e., $R_{C_1}=R_{C_2}=1$, and $\alpha_{C_1,1}=\alpha_{C_1,2}$, $\alpha_{C_2,1}=\alpha_{C_2,2}$. Both designers get a utility of $1$.

Let designer $2$ deviate to $C_2'$, where $S_{C_2'}=\{1,3\}$, $\alpha_{C_2',1}=\alpha_{C_2',2}$, and $R_{C_2'}=1$. Let $\vec{\mathcal{C}}'$ denote the new strategy profile of designers. It is easy to verify that the unique contestant equilibrium under $\vec{\mathcal{C}}'$ is $x'_{1,C_1}=x'_{1,C_2'}=\frac12,x'_{2,C_1}=1,x'_{3,C_2'}=1$, with equilibrium multiplier vector $\lambda_1'=\frac49$ and $\lambda_2'=\lambda_3'=\frac29$. We can see that designer $2$'s utility increases to $\frac32>1$. Therefore $\vec{\mathcal{C}}$ cannot be a designer equilibrium if $S_{C_1}=S_{C_2}$.

(b) If $S_{C_1}\neq S_{C_2}$, without loss of generality, assume that $S_{C_1}=\{1,2\},S_{C_2}=\{1,3\}$. We calculate the contestant equilibrium. We know that $\lambda_1>0$ by similar argument to (a). Additionally, since $\hat{x}_{1,C_1}(\vec{\lambda})+\hat{x}_{1,C_2}(\vec{\lambda})=T_1>0$, we get that $\lambda_2+\lambda_3>0$. 

We claim that for any designer $j\in\{1,2\}$, her utility must be at least $\frac32$. Otherwise, suppose without loss of generality that designer $1$'s utility is less than $\frac32$. The, she can deviate to $C_1'$, where $R_{C_1'}=R_{C_2}$, $S_{C_1'}=S_{C_1}=\{1,2\}$, $\alpha_{C_1',1}=\alpha_{C_2,1}$, $\alpha_{C_1',2}=\alpha_{C_2,3}$. It is easy to verify that the contestant equilibrium is $x'_{1,C_1'}=x'_{1,C_2}=\frac12,x'_{2,C_1'}=1,x'_{3,C_2}=1$. This increases designer $1$'s utility to $\frac32$, contradicting with that $\vec{\mathcal{C}}$ is a designer equilibrium. The same holds for designer $2$ symmetrically. In other words we have $x_{1,C_1}+x_{2,C_1}=\frac32,x_{1,C_2}+x_{3,C_2}=\frac32$, which implies that $x_{1,C_1}=x_{1,C_2}=\frac12$, $x_{2,C_1}=x_{3,C_2}=1$. 

When $x_{1,C_1}>0$ and $x_{1,C_2}>0$, we know that 
\begin{align*}
\lambda_1&=\frac{R_{C_1}\alpha_{C_1,1}\alpha_{C_1,2}x_{2,C_1}}{(\alpha_{C_1,1}x_{1,C_1}+\alpha_{C_1,2}x_{2,C_1})^2}\\
&=\frac{R_{C_2}\alpha_{C_2,1}\alpha_{C_2,3}x_{3,C_2}}{(\alpha_{C_2,1}x_{1,C_2}+\alpha_{C_2,3}x_{3,C_2})^2}.
\end{align*} When $\alpha_{C_1}$ and $\alpha_{C_2}$ is slightly perturbed, it will still hold that in the contestant equilibrium $x_{2,C_1}=x_{3,C_2}=1$. Replacing $x_{2,C_1}$ and $x_{3,C_2}$ by $1$, we get $$\frac{R_{C_1}\alpha_{C_1,1}\alpha_{C_1,2}}{(\alpha_{C_1,1}x_{1,C_1}+\alpha_{C_1,2})^2}=\frac{R_{C_2}\alpha_{C_2,1}\alpha_{C_2,3}}{(\alpha_{C_2,1}x_{1,C_2}+\alpha_{C_2,3})^2}$$ Then, we know 
\begin{align*}
    &R_{C_1}(\frac{\alpha_{C_2,1}}{\alpha_{C_2,3}}x_{1,C_2}^2+2x_{1,C_2}+\frac{\alpha_{C_2,3}}{\alpha_{C_2,1}})\\
    =&R_{C_2}(\frac{\alpha_{C_1,1}}{\alpha_{C_1,2}}x_{1,C_1}^2+2x_{1,C_1}+\frac{\alpha_{C_1,2}}{\alpha_{C_1,1}}).
\end{align*}

Define $\rho_1=\frac{\alpha_{C_1,1}}{\alpha_{C_1,2}}$ and $\rho_2=\frac{\alpha_{C_2,1}}{\alpha_{C_2,2}}$, and replace $x_{1,C_2}$ by $1-x_{1,C_1}$. It holds that 
\begin{align*}
    &R_{C_1}(\rho_2(1-x_{1,C_1})^2+2(1-x_{1,C_1})
    +\frac1{\rho_2})\\
    -&R_{C_2}(\rho_1x_{1,C_1}^2+2x_{1,C_1}+\frac1{\rho_1})=0.
\end{align*}
Taking derivatives, we get 
\begin{align*}
&(-2R_{C_1}\rho_2(1-x_{1,C_1})-2R_{C_2}\rho_1x_{1,C_1}-2R_{C_1}-2R_{C_2})\mathrm{d}x_{1,C_1}\\
&+R_{C_1}((1-x_{1,C_1})^2-\rho_2^{-2})\mathrm{d}\rho_2-R_{C_2}(x_{1,C_1}^2-\rho_1^{-2})\mathrm{d}\rho_1=0.
\end{align*}
It always holds that $-2R_{C_1}\rho_2(1-x_{1,C_1})-2R_{C_2}\rho_1x_{1,C_1}$ $-2R_{C_1}-2R_{C_2}<0$, so
we have 
$$\frac{\partial x_{1,C_1}}{\partial \rho_1}=\frac{-R_{C_2}(x_{1,C_1}^2-\rho_1^{-2})}{R_{C_1}\rho_2(1-x_{1,C_1})+2R_{C_2}\rho_1x_{1,C_1}+2R_{C_1}+2R_{C_2}}$$ and $$\frac{\partial x_{1,C_1}}{\partial \rho_2}=\frac{R_{C_1}((1-x_{1,C_1})^2-\rho_2^{-2})}{R_{C_1}\rho_2(1-x_{1,C_1})+2R_{C_2}\rho_1x_{1,C_1}+2R_{C_1}+2R_{C_2}}.$$ Since $\vec{\mathcal{C}}$ is a designer equilibrium, it must satisfy that $\frac{\partial x_{1,C_1}}{\partial \rho_1}=0$ and $\frac{\partial x_{1,C_1}}{\partial \rho_2}=0$. It follows that $\rho_1={x_{1,C_1}}^{-1}=2$ and $\rho_2=(1-x_{1,C_1})^{-1}=2$.

Now we know that $\frac{\alpha_{C_1,1}}{\alpha_{C_1,2}}=2$.
Let designer $2$ deviate to $C_2'$, where $S_{C_2'}=\{2,3\}$, $\alpha_{C_2',2}=1,\alpha_{C_2',3}=7-2\sqrt{10}$, and $R_{C_2'}=R_{C_1}$. Under this new strategy profile of designers, one can verify with some calculation that the unique contestant equilibrium is $x_{2,C_1}'=2\sqrt{10}-6,x_{2,C_2'}'=7-2\sqrt{10},x_{1,C_1}'=x_{3,C_2'}'=1$. The designer $2$'s utility increases to $1+7-2\sqrt{10}\approx1.6754>1.5$ (In fact, this is designer $2$'s best response to $C_1$ if $R_{C_1}=1$.). This contradicts with that $\vec{\mathcal{C}}$ is a designer equilibrium.

Since both cases of $S_{C_1}=S_{C_2}$ and $S_{C_1}\neq S_{C_2}$ lead to a contradiction, we know that there is no designer equilibrium.



\end{proof}


\subsection{Proof of \Cref{lemma:winning-probability-from-lambda}}
\INDhatplambda*

\begin{proof}
For any contest $C$, suppose $S_C=\{i_1,i_2\}$. There are three possible cases:

(a) If $\lambda_{i_1}>0$ and $\lambda_{i_2}>0$, for any contestant equilibrium $\vec{x}$, we have $x_{i_1,C}=\hat{x}_{i_1,C}(\vec{\lambda})$ and $x_{i_2,C}=\hat{x}_{i_2,C}(\vec{\lambda})$. Therefore, it holds that  $\frac{x_{i_1,C}}{x_{i_2,C}}=\frac{\hat{x}_{i_1,C}(\vec{\lambda})}{\hat{x}_{i_2,C}(\vec{\lambda})}=\frac{\lambda_{i_2}}{\lambda_{i_1}}$, which implies that $p_{i_1,C}(\vec{x})=\frac{\alpha_{C,i_1}\lambda_{i_2}}{\alpha_{C,i_1}\lambda_{i_2}+\alpha_{C,i_2}\lambda_{i_1}}$, and $p_{i_2,C}(\vec{x})=\frac{\alpha_{C,i_2}\lambda_{i_1}}{\alpha_{C,i_1}\lambda_{i_2}+\alpha_{C,i_2}\lambda_{i_1}}$.

(b) If $\lambda_{i_1}>0$ and $\lambda_{i_2}=0$, for any contestant equilibrium $\vec{x}$, we have $x_{i_1,C}=\hat{x}_{i_1,C}(\vec{\lambda})=0$ and $x_{i_2,C}\geq\hat{x}_{i_2,C}(\vec{\lambda})>0$. which means that $p_{i_1,C}(\vec{x})=0=\frac{\alpha_{C,i_1}\lambda_{i_2}}{\alpha_{C,i_1}\lambda_{i_2}+\alpha_{C,i_2}\lambda_{i_1}}$, and $p_{i_2,C}(\vec{x})=1=\frac{\alpha_{C,i_2}\lambda_{i_1}}{\alpha_{C,i_1}\lambda_{i_2}+\alpha_{C,i_2}\lambda_{i_1}}$.

(c) If $\lambda_{i_1}=0$ and $\lambda_{i_2}>0$, similar with the argument in (b),  we have $p_{i_1,C}(\vec{x})=1=\frac{\alpha_{C,i_1}\lambda_{i_2}}{\alpha_{C,i_1}\lambda_{i_2}+\alpha_{C,i_2}\lambda_{i_1}}$ and $p_{i_2,C}(\vec{x})=0=\frac{\alpha_{C,i_2}\lambda_{i_1}}{\alpha_{C,i_1}\lambda_{i_2}+\alpha_{C,i_2}\lambda_{i_1}}$ for any contestant equilibrium $\vec{x}$.
\end{proof}

\subsection{Proof of \Cref{lemma:designer-control-probability-as-bias}}
\INDControlProbabilityAsBias*

\begin{proof}
For any contest $C\in\mathcal{C}^{var}$, define $\tilde{Q}_{C}=\prod_{i\in S_C}\tilde{p}_{i,C}$. For any contestant  $i$, let $a_i=\sum_{C\in\mathcal{A}(i,\mathcal{C}^{var})}R_C\tilde{Q}_{C}$ and $\hat{T}_i(\vec{\lambda})=\sum_{C\in\mathcal{C}^{fix}}\hat{x}_{i,C}(\vec{\lambda})$. Then, it is not hard to see that, the biases of all $C\in \mathcal{C}^{var}$ and the EMV $\vec{\lambda}$ under $\mathcal{C}$ can achieve the assigned winning probabilities if and only if $\vec{a}=(a_i)_{i\in[n]}$ and $\vec{\lambda}$ satisfies the conditions in \Cref{lemma:addtermcontestantequilibrium}. Thus, by \Cref{lemma:addtermcontestantequilibrium}, there exists a unique $\vec{\lambda}$ satisfying the requirement.

For any contest $C\in\mathcal{C}^{var}$, let $S_C=\{i_1,i_2\}$, and we have $\frac{\alpha_{C,i_1}\lambda_{i_2}}{\alpha_{C,i_2}\lambda_{i_1}}=\frac{\tilde{p}_{i_1,C}}{\tilde{p}_{i_2,C}}$. It follows that $\frac{\alpha_{C,i_1}}{\alpha_{C,i_2}}=\frac{\tilde{p}_{i_1,C}\lambda_{i_1}}{\tilde{p}_{i_2,C}\lambda_{i_2}}$. When normalized, the biases are uniquely determined as $\alpha^*_{C,i_1}=\frac{\tilde{p}_{i_1,C}\lambda_{i_1}}{\tilde{p}_{i_1,C}\lambda_{i_1}+\tilde{p}_{i_2,C}\lambda_{i_2}}$ and $\alpha^*_{C,i_2}=\frac{\tilde{p}_{i_2,C}\lambda_{i_2}}{\tilde{p}_{i_1,C}\lambda_{i_1}+\tilde{p}_{i_2,C}\lambda_{i_2}}$.
\end{proof}

\subsection{Proof of \Cref{thm:example-balancing-is-not-best}}
\INDthmExampleBalanceNotBest*
\begin{proof}
    
    Consider the following instance: there are four contestants and four designers, where the total efforts are $\vec{T}=(0.251, 251, 2, 0.002)$ and the prizes are $R_{C_1}=1,R_{C_2}=R_{C_3}=R_{C_4}=1.002001$.
    The strategy profile of participants selection is given by $S_{C_1}=\{1,2\},S_{C_2}=\{1,3\},S_{C_3}=\{2,4\},S_{C_4}=\{3,4\}$.

    If the current bias profile of designers is $\alpha_{C_1,1}=1000,\alpha_{C_1,2}=1$, and $\alpha_{C_j,i}=1$ for all $j\in\{2,3,4\},i\in S_{C_j}$. It is not hard to check that the contestants' equilibrium multiplier vector is $\vec{\lambda}=(1,0.001,0.001,1)$, and the unique contestant equilibrium $\vec{x}$ is given by $x_{1,C_1}=0.25,x_{2,C_1}=250,x_{1,C_2}=0.001,x_{3,C_2}=1,x_{2,C_3}=1,x_{4,C_3}=0.001,x_{3,C_4}=1,x_{4,C_4}=0.001$. Moreover, it holds that $\alpha_{C_1,1}x_{1,C_1}=\alpha_{C_1,2}x_{2,C_1}$, so $p_{1,C_1}(\vec{x})=p_{2,C_1}(\vec{x})=\frac12$. The utility of designer 1 is $x_{1,C_1}+x_{2,C_1}=250.25$.

    Now suppose that designer 1 changes the bias to $\alpha_{C_1,1}=990,\alpha_{C_1,2}=1$. Under the new strategy profile of designers, we can use \Cref{alg:Alg1} to compute the new contestant equilibrium. Up to small enough error, the new equilibrium multiplier vector is $\vec{\lambda}'=(0.9999754268144135, 0.0009999746482529694, \\ 0.0010001217961560266, 1.0000240866947854)$, and the contestant equilibrium $\vec{x}'$ is \\
$(x'_{1,C_1}= 0.2499998293421785,\\x'_{2,C_1}= 250.00002398734125,\\ x'_{1,C_2}= 0.0010001706578214883,\\x'_{3,C_2}= 1.0000242813288966,\\x'_{2,C_3}= 0.999976012658724,\\x'_{4,C_3}= 0.0009999265765935563,\\x'_{3,C_4}= 0.9999757186711037,\\x'_{4,C_4}= 0.001000073423406445)$. \par
One can find that $p_{1,C_1}(\vec{x}')=1-p_{2,C_1}(\vec{x}')=0.4974872425456253\neq\frac12$. However, the utility of designer 1 is $x'_{1,C_1}+x'_{2,C_1}=250.2500238166834$, which is larger than the utility of $250.25$ when she uses the balancing bias. Therefore using the balancing bias is not the best response of designer $1$ in the second substage of designers.
\end{proof}

\subsection{Proof of \Cref{thm:indiv-1/2probability-is-equilibrium}}
\INDthmEqualProbabilitySecondStageEquilibrium*

\begin{proof}
We prove this theorem by contradiction. Suppose that $\vec{\mathcal{C}}$ is not an equilibrium in the second substage of designers, that is, there is some designer who has incentive to deviate from $\vec{\mathcal{C}}$ by adjusting the bias of her contest. Without loss of generality, we assume that it is designer $1$ who deviates from $C_1$, where $S_{C_1}=\{1,2\}$, for convenience. Suppose that $\mathcal{C}_1'=\{C_1'\}$ is a beneficial deviation for designer $1$, satisfying that $R_{C_1'}=R_{C_1}$ and $S_{C_1'}=S_{C_1}=\{1,2\}$. Let $\vec{\mathcal{C}}'=(\mathcal{C}_1',\vec{\mathcal{C}}_{-1})$ denote the strategy profile after designer $1$ deviating. Let $\vec{\lambda}'$ denote the equilibrium multiplier vector under $\vec{\mathcal{C}}'$. There must exist a contestant equilibrium $\vec{x}'$ under $\vec{\mathcal{C}}'$ such that $x'_{1,C_1'}+x'_{2,C_1'}>\hat{x}_{1,C_1}(\vec{\lambda})+\hat{x}_{2,C_1}(\vec{\lambda})$.

Note that since $\hat{p}_{i,C}(\vec{\lambda})=\frac12$ for any $C\in\vec{\mathcal{C}}$ and $i\in S_C$, we have $\lambda_i>0$ and 
$\frac{\alpha_{C,i}}{\alpha_{C,\op_{i,C}}}=\frac{\lambda_i}{\lambda_{\op_{i,C}}}$ for any contest $i$. Because $\vec{\lambda}'\neq\vec{\lambda}$, by \Cref{lemma:designer-control-probability-as-bias} we know that $Q_{C_1'}(\vec{\lambda}')\neq Q_{C_1}(\vec{\lambda})$, i.e., $Q_{C_1'}(\vec{\lambda}')<\frac14$.

Firstly, we prove the following claims:
\begin{description}
\item[Claim 1] For any contestant $i\in[n]\setminus\{1,2\}$, it holds that $\lambda_i'\leq\lambda_i$.
\item[Claim 2] For any contestant $i\in\{1,2\}$, it holds that $\lambda_i'<\lambda_i$.
\item[Claim 3] For any contestant $i\in[n]$, it holds that $\lambda_i'>0$.
\end{description}

To prove Claim 1, recall that for any contestant $i$ and any contest $j$, $\hat{p}_{i,C_j}(\vec{\lambda})=\frac12$, and $Q_{i,C_j}(\vec{\lambda})=\frac14\geq Q_{i,C_j}(\vec{\lambda}')$.
Observe that for any $i>2$, $T_i=\sum_{C\in\mathcal{A}(i,\vec{\mathcal{C}}_{-1})}\hat{x}_{i,C}(\vec{\lambda})=\sum_{C\in\mathcal{A}(i,\vec{\mathcal{C}}_{-1})}\frac{R_CQ_C(\vec{\lambda})}{\lambda_i}=\sum_{C\in\mathcal{A}(i,\vec{\mathcal{C}}_{-1})}\frac{R_C}{4\lambda_i}$.
If $\lambda_i'>0$, we have $T_i=\sum_{C\in\mathcal{A}(i,\vec{\mathcal{C}}_{-1})}\hat{x}_{i,C}(\vec{\lambda}')=\sum_{C\in\mathcal{A}(i,\vec{\mathcal{C}}_{-1})}\frac{R_CQ_C(\vec{\lambda}')}{\lambda_i'}$, which means that  $\frac{\lambda_i'}{\lambda_i}=\frac{\sum_{C\in\mathcal{A}(i,\vec{\mathcal{C}}_{-1})}R_CQ_C(\vec{\lambda}')}{\sum_{C\in\mathcal{A}(i,\vec{\mathcal{C}}_{-1})}\frac14R_C}\leq 1$. In addition, if $\lambda_i'=0$, we have $\lambda_i>\lambda_i'$. In summary, it holds that $\lambda_i'\leq\lambda_i$.

For Claim 2, for each $i\in\{1,2\}$, it holds that $T_i=\sum_{C\in\mathcal{A}(i,\vec{\mathcal{C}}_{-1})}\frac{R_CQ_C(\vec{\lambda})}{\lambda_i}+\frac{R_{C_1}Q_{C_1}(\vec{\lambda})}{\lambda_i}$. If $\lambda_i'>0$,  we also have $T_i=\sum_{C\in\mathcal{A}(i,\vec{\mathcal{C}}_{-1})}\frac{R_CQ_C(\vec{\lambda}')}{\lambda_i'}+\frac{R_{C_1'}Q_{C_1'}(\vec{\lambda}')}{\lambda_i'}$. We know that $R_{C_1'}Q_{C_1'}(\vec{\lambda}')<\frac{R_{C_1}}{4}=R_{C_1}Q_{C_1}(\vec{\lambda})$, so $\frac{\lambda_i'}{\lambda_i}=\frac{\sum_{C\in\mathcal{A}(i,\vec{\mathcal{C}}_{-1})}R_CQ_C(\vec{\lambda}')+R_{C_1'}Q_{C_1'}(\vec{\lambda}')}{\frac14(\sum_{C\in\mathcal{A}(i,\vec{\mathcal{C}}_{-1})}R_C+R_{C_1})}<1$. And if $\lambda_i'=0$, it also holds that $\lambda_i'<\lambda_i$. In summary we get $\lambda_i'<\lambda_i$.

For Claim 3, we firstly prove that for each contestant $i\in\{1,2\}$, $\lambda_i'>0$. Suppose for contradiction that for some $i\in\{1,2\}$, $\lambda_i'=0$. Observe that 
\begin{align*}
T_i\geq& \sum_{C\in\mathcal{A}(1,\vec{\mathcal{C}}_{-1})}\hat{x}_{i,C}(\vec{\lambda}')+x'_{i,C_1'}\\
=&\sum_{C\in\mathcal{A}(i,\vec{\mathcal{C}}_{-1})}\frac{R_C\alpha_{C,\op_{i,C}}}{\alpha_{C,i}\lambda'_{\op_{i,C}}}+x'_{i,C_1'}\\
=&\sum_{C\in\mathcal{A}(i,\vec{\mathcal{C}}_{-1})}\frac{R_C\lambda_{\op_{i,C}}}{\lambda_{i}\lambda'_{\op_{i,C}}}+x'_{i,C_1'}\\
\geq&\sum_{C\in\mathcal{A}(i,\vec{\mathcal{C}}_{-1})}\frac{R_C}{\lambda_{i}}+x'_{i,C_1'}.
\end{align*}
Therefore, we get
\begin{align*}
x'_{i,C_1'} &\leq T_i-\sum_{C\in\mathcal{A}(i,\vec{\mathcal{C}}_{-1})}\frac{R_C}{\lambda_{i}}\\
& \leq T_i-\sum_{C\in\mathcal{A}(i,\vec{\mathcal{C}}_{-1})}\frac{R_C}{4\lambda_{i}}=\hat{x}_{i,C_1}(\vec{\lambda}).
\end{align*}
Meanwhile, we have $x'_{3-i,C_1'}=0<\hat{x}_{3-i,C_1}(\vec{\lambda})$, which means that $x'_{1,C_1'}+x'_{2,C_1'}<\hat{x}_{1,C_1}(\vec{\lambda})+\hat{x}_{2,C_1}(\vec{\lambda})$. It contradicts with the assumption that $x'_{1,C_1'}+x'_{2,C_1'}>\hat{x}_{1,C_1}(\vec{\lambda})+\hat{x}_{2,C_1}(\vec{\lambda})$.

Secondly, for any contestant $i>2$, if $\lambda_i'=0$,  we have 
\begin{align*}
T_i\geq&\sum_{C\in\mathcal{A}(1,\vec{\mathcal{C}}_{-1})}\hat{x}_{i,C}(\vec{\lambda}')\\
=&\sum_{C\in\mathcal{A}(i,\vec{\mathcal{C}}_{-1})}\frac{R_C\lambda_{\op_{i,C}}}{\lambda_{i}\lambda'_{\op_{i,C}}}\\
\geq&\sum_{C\in\mathcal{A}(i,\vec{\mathcal{C}}_{-1})}\frac{R_C}{\lambda_{i}}.
\end{align*}
This contradicts with that $\sum_{C\in\mathcal{A}(i,\vec{\mathcal{C}}_{-1})}\frac{R_C}{4\lambda_{i}}=T_i$ and $T_i>0$. This completes the proof of Claim 3.

By Claim 3, we have $x'_{i,C_1'}=\hat{x}_{i,C_1'}(\vec{\lambda}')=\frac{R_{C_1'}Q_{C_1'}(\vec{\lambda}')}{\lambda'_i}$ for both $i\in\{1,2\}$, so $x'_{i,C_1'}>\hat{x}_{i,C_1}(\vec{\lambda})$ if and only if $\frac{\lambda'_i}{\lambda_i}<\frac{Q_{C_1'}(\vec{\lambda}')}{Q_{C_1}(\vec{\lambda})}$. By the assumption that $x'_{1,C_1'}+x'_{2,C_1'}>\hat{x}_{1,C_1}(\vec{\lambda})+\hat{x}_{2,C_1}(\vec{\lambda})$, we know that $\min_{i\in\{1,2\}}\frac{\lambda'_i}{\lambda_i}<\frac{Q_{C_1'}(\vec{\lambda}')}{Q_{C_1}(\vec{\lambda})}<1$.

Next we prove the following three claims, which will lead to an impossible infinite descent.
\begin{description}
    \item[Claim 4] For any $y\in(0,1),z\in(0,1]$, if $\frac{yz}{(y+z)^2}\leq \frac{y}{4}$, it holds that $z<y$.
    \item[Claim 5] For any $i>2$, if $\frac{\lambda_i'}{\lambda_i}<1$, then there exists $k\in[n]$, such that $\frac{\lambda_k'}{\lambda_k}<\frac{\lambda_i'}{\lambda_i}$.
    \item[Claim 6] For any $i\in\{1,2\}$, if $\frac{\lambda_i'}{\lambda_i}<\frac{Q_{C_1'}(\vec{\lambda}')}{Q_{C_1}(\vec{\lambda})}$, there exists a $k\in[n]$, such that $\frac{\lambda_k'}{\lambda_k}<\frac{\lambda_i'}{\lambda_i}$.
\end{description}

We firstly prove Claim 4. If $\frac{yz}{(y+z)^2}\leq \frac{y}{4}$, we can calculate that $(z+y-2)^2\geq-4y+4=4(1-y)$. Since $z+y-2<0$, we have $z+y-2=-\sqrt{4(1-y)}$, and consequently $z=2-y-\sqrt{4(1-y)}=(1-\sqrt{1-y})^2=\frac{1-\sqrt{1-y}}{1+\sqrt{1-y}}y<y$.

Now we prove Claim 5. For any $i>2$, if $\frac{\lambda_i'}{\lambda_i}<1$, since we know from claim 3 that $\lambda_i'>0$, we have $\frac{\lambda_i'}{\lambda_i}=\frac{\sum_{C\in\mathcal{A}(i,\vec{\mathcal{C}}_{-1})}R_CQ_C(\vec{\lambda}')}{\sum_{C\in\mathcal{A}(i,\vec{\mathcal{C}}_{-1})}\frac14R_C}$.

Therefore, there exists $C\in\mathcal{A}(i,\vec{\mathcal{C}}_{-1})$ such that $Q_C(\vec{\lambda}')\leq \frac14\frac{\lambda_i'}{\lambda_i}$. Take $k=\op_{i,C}$. Since $Q_C(\vec{\lambda}')=\hat{p}_{i,C}(\vec{\lambda}')\hat{p}_{k,C}(\vec{\lambda}')=\frac{\lambda_i\lambda'_k\lambda_k\lambda'_i}{(\lambda_i\lambda'_k+\lambda_k\lambda'_i)^2}=\frac{\frac{\lambda'_k}{\lambda_k}\frac{\lambda'_i}{\lambda_i}}{(\frac{\lambda'_k}{\lambda_k}+\frac{\lambda'_i}{\lambda_i})^2}$, it holds that $\frac{\frac{\lambda'_k}{\lambda_k}\frac{\lambda'_i}{\lambda_i}}{(\frac{\lambda'_k}{\lambda_k}+\frac{\lambda'_i}{\lambda_i})^2}\leq\frac14\frac{\lambda_i'}{\lambda_i}$. We know that $\frac{\lambda'_k}{\lambda_k}\in[0,1]$ and $\frac{\lambda'_i}{\lambda_i}\in(0,1)$, so by claim 4, we obtain $\frac{\lambda'_k}{\lambda_k}<\frac{\lambda'_i}{\lambda_i}$.

We can prove Claim 6 similarly. For any contestant $i\in\{1,2\}$, we know that $\frac{\lambda_i'}{\lambda_i}=\frac{\sum_{C\in\mathcal{A}(i,\vec{\mathcal{C}}_{-1})}R_CQ_C(\vec{\lambda}')+R_{C_1}Q_{C_1'}(\vec{\lambda}')}{\sum_{C\in\mathcal{A}(i,\vec{\mathcal{C}}_{-1})}\frac14R_C+R_{C_1}Q_{C_1}(\vec{\lambda})}$. If $\frac{Q_{C_1'}(\vec{\lambda}')}{Q_{C_1}(\vec{\lambda})}>\frac{\lambda_i'}{\lambda_i}$, we have $\frac{\sum_{C\in\mathcal{A}(i,\vec{\mathcal{C}}_{-1})}R_CQ_C(\vec{\lambda}')}{\sum_{C\in\mathcal{A}(i,\vec{\mathcal{C}}_{-1})}\frac14R_C}<\frac{\lambda_i'}{\lambda_i}$, so there exists $C\in\mathcal{A}(i,\vec{\mathcal{C}}_{-1})$ such that $Q_C(\vec{\lambda}')\leq \frac14\frac{\lambda_i'}{\lambda_i}$. Take $k=\op_{i,C}$. It holds that $\frac{\frac{\lambda'_k}{\lambda_k}\frac{\lambda'_i}{\lambda_i}}{(\frac{\lambda'_k}{\lambda_k}+\frac{\lambda'_i}{\lambda_i})^2}\leq\frac14\frac{\lambda_i'}{\lambda_i}$, and by claim 4, we obtain $\frac{\lambda'_k}{\lambda_k}<\frac{\lambda'_i}{\lambda_i}$.

Finally, since $\min_{i\in\{1,2\}}\frac{\lambda'_i}{\lambda_i}<\frac{Q_{C_1'}(\vec{\lambda}')}{Q_{C_1}(\vec{\lambda})}$, by Claim 6 we know that $\min_{i\in[n]}\frac{\lambda'_i}{\lambda_i}<\min_{i\in\{1,2\}}\frac{\lambda'_i}{\lambda_i}$, which implies that  $\min_{i\in[n]\setminus\{1,2\}}\frac{\lambda'_i}{\lambda_i}<\min_{i\in\{1,2\}}\frac{\lambda'_i}{\lambda_i}<1$. However, by Claim 5, it follows that $\min_{i\in[n]}\frac{\lambda'_i}{\lambda_i}<\min_{i\in[n]\setminus\{1,2\}}\frac{\lambda'_i}{\lambda_i}$, which implies that $\min_{i\in\{1,2\}}\frac{\lambda'_i}{\lambda_i}<\min_{i\in[n]\setminus\{1,2\}}\frac{\lambda'_i}{\lambda_i}$. It is a contradiction. We complete the proof.
\end{proof}

\subsection{Proof of \Cref{thm:weak-SPE-existence}}
\INDweakSPEexists*

\begin{proof}
To prove the existence of weak designer equilibrium, firstly we show that the game in the first substage of designers is strategically equivalent to a variant of weighted congestion game, if we assume that all designers take the balancing bias in the second substage. Secondly we show that in this variant of weighted congestion game, the pure Nash equilibrium always exists, which implies the existence of weak designer equilibrium.

For any strategy profile of designers $\vec{\mathcal{C}}$ such that the second-substage strategies form the equilibrium given in \Cref{thm:indiv-1/2probability-is-equilibrium}, let $\vec{\lambda}$ and $\vec{x}$ denote the equilibrium multiplier vector and contestant equilibrium under $\vec{\mathcal{C}}$. We can observe that for any contestant $i$ and contest $C\in\mathcal{A}(i,\vec{\mathcal{C}})$, it holds that $x_{i,C}=\hat{x}_{i,C}(\vec{\lambda})=\frac{R_CQ_C(\vec{\lambda})}{\lambda_i}=\frac{R_C}{4\lambda_i}$.
Combining this with that $\sum_{C\in\mathcal{A}(i,\vec{\mathcal{C}})}\hat{x}_{i,C}(\vec{\lambda})=\hat{T}_i(\vec{\lambda})=T_i$, we get $\lambda_i=\sum_{C\in\mathcal{A}(i,\vec{\mathcal{C}})}\frac{R_C}{4T_i}$, and therefore we obtain $$x_{i,C}=\frac{R_C}{4\lambda_i}=T_i\frac{R_C}{\sum_{C'\in\mathcal{A}(i,\vec{\mathcal{C}})}R_{C'}}.$$

In other words, each contestant will distribute her total effort $T_i$ into the contests inviting her, such that the effort exerted into each contest is proportional to the prize amount of that contest.

In the first substage of designers, each designer $j$ decides prize $R_{C_j}$ and the two participants $S_{C_j}$, to maximize her utility under the second substage equilibrium and contestant equilibrium, which equals to $$\sum_{i\in S_{C_j}}T_i\frac{R_{C_j}}{\sum_{j'\in[m]:i\in S_{C_{j'}}}R_{C_{j'}}}.$$ It is easy to see that the setting $R_{C_j}=B_j$ dominantly maximizes her utility. We only need to consider the selection of $S_{C_j}$, which can be viewed as the following \textbf{variant of weighted congestion game}:

There are $m$ agents $a_1,\cdots,a_m$ representing $m$ designers and $n$ resources $e_1,\cdots,e_n$ representing $n$ contestants. Each agent $a_j$ has a weight $w_j=B_j$ and each resource $e_i$ has an amount of total reward $v_i=T_i$. Each agent selects a strategy $s_j$ from a common strategy space $\mathcal{S}=\{\{e_{i_1},e_{i_2}\}\subseteq \{e_1,\cdots,e_n\}:i_1\neq i_2\}$, where $s_j=\{e_{i_1},e_{i_2}\}$ represents $S_{C_j}=\{i_1,i_2\}$. Under a strategy profile $\vec{s}=(s_1,\cdots,s_m)$, each resource $e_i$'s load is defined as $c_i(\vec{s})=\sum_{j\in[m]:e_i\in{s_j}}w_j$ and its reward function is defined as $r_i(\vec{s})=\frac{v_i}{c_i(\vec{s})}$. Each agent $a_j$ tries to maximize her total reward $V_j(\vec{s})=\sum_{i\in s_j}r_i(\vec{s})$. It is easy to see that the game between designers in the first substage is strategically equivalent to this variant of weighted congestion game, since each designer $j$'s utility $\sum_{i\in S_{C_j}}T_i\frac{B_j}{\sum_{j'\in[m]:i\in S_{C_{j'}}}B_{j'}}$ is equal to $B_jV_j(\vec{s})$.

We adapt \cite{ACKERMANN20091552}'s proof of the existence of PNE in weighted matroid congestion games, to prove that the above variant of weighted congestion game always has a PNE, by constructing a lexicographical order.

Firstly, given any strategy profile $\vec{s}=(s_1,\cdots,s_m)$, if there is some agent $a_j$ who has the incentive to deviate from $s_j$, that is, there is some $s_j'\in\mathcal{S}$ such that $V_j(s_j',\vec{s}_{-j})>V_j(\vec{s})$, we prove that there exists $s_j''\in\mathcal{S}$ such that $|s_j''\cap s_j|=1$ and $V_j(s_j'',\vec{s}_{-j})>V_j(\vec{s})$. Let $c_i(\vec{s}_{-j})$ denote $\sum_{j'\in[m]\setminus\{j\}:e_i\in{s_{j'}}}w_{j'}$ and define $r_i^j(\vec{s}_{-j})=\frac{v_i}{w_j+c_i(\vec{s}_{-j})}$, which is the reward that the agent $a_j$ will get from $e_i$ if $i\in s_j$, when given the strategies of other agents $\vec{s}_{-j}$. Then, it holds that for any $s_j\in \mathcal{S}$, suppose $s_j=\{e_{i_1},e_{i_2}\}$, and $V_j(s_j,\vec{s}_{-j})=r_{i_1}^j(\vec{s}_{-j})+r_{i_2}^j(\vec{s}_{-j})$. Therefore, when $V_j(s_j',\vec{s}_{-j})>V_j(\vec{s})$, suppose $s_j=\{e_{i_1},e_{i_2}\}$ and $s_j'=\{e_{i_3},e_{i_4}\}$. Without loss of generality,  assume that $r_{i_1}^j(\vec{s}_{-j})\geq r_{i_2}^j(\vec{s}_{-j})$. There must exist $i'\in \{i_3,i_4\}$ such that $i'\notin \{i_1,i_2\}$ and $r_{i'}^j(\vec{s}_{-j})>r_{i_2}^j(\vec{s}_{-j})$. Taking $s_j''=\{e_{i'},e_{i_1}\}$ makes it satisfy that $|s_j''\cap s_j|=1$ and $V_j(s_j'',\vec{s}_{-j})>V_j(\vec{s})$.

For any strategy profile $\vec{s}$, define $\bar{r}(\vec{s})\in\mathbb{R}^n$ as the vector obtained by sorting $r_1(\vec{s}),\cdots,r_n(\vec{s})$\footnote{Specifically, define $r_i(\vec{s})=+\infty$ if $e_i\notin\cup_{j\in[m]}s_j$}
in ascending order, i.e., $\bar{r}_i(\vec{s})$ is the $i$-th smallest number among $r_1(\vec{s}),\cdots,r_n(\vec{s})$.

For any two strategy profiles $\vec{s}$ and $\vec{s}'$, we say $\bar{r}(\vec{s})$ is lexicographically greater than $\bar{r}(\vec{s}')$ if there exists $k\leq n$, such that $\bar{r}_i(\vec{s})=\bar{r}_i(\vec{s}')$ for all $i<k$ and $\bar{r}_k(\vec{s})>\bar{r}_k(\vec{s}')$, denoted by $\bar{r}(\vec{s})>_{lex}\bar{r}(\vec{s}')$. 

Now we show that, if there is some $j\in[m]$ and some $s_j'\in\mathcal{S}$ such that $|s_j'\cap s_j|=1$ and $V_j(s_j',\vec{s}_{-j})>V_j(\vec{s})$, we have $\bar{r}(s_j',\vec{s}_{-j})>_{lex}\bar{r}(\vec{s})$. Suppose $s_j=\{e_{i_1},e_{i_2}\}$ and $s_j'=\{e_{i_1},e_{i_3}\}$. We know that for any $i\in[n]\setminus \{i_2,i_3\}$, $r_i(s_j',\vec{s}_{-j})=r_i(\vec{s})$, so only $r_{i_2}$ and $r_{i_3}$ may change. We have $r_{i_2}(s_j',\vec{s}_{-j})>r_{i_2}(\vec{s})$ since $c_{i_2}(s_j',\vec{s}_{-j})=c_{i_2}(\vec{s})-w_j$, and we also know $r_{i_3}(s_j',\vec{s}_{-j})>r_{i_2}(\vec{s})$ by $r_{i_3}(s_j',\vec{s}_{-j})-r_{i_2}(\vec{s})=V_j(s_j',\vec{s}_{-j})-V_j(\vec{s})>0$. Therefore, it holds that $\min\{r_{i_2}(s_j',\vec{s}_{-j}),r_{i_3}(s_j',\vec{s}_{-j})\}>r_{i_2}(\vec{s})\geq \min\{r_{i_2}(\vec{s}),r_{i_3}(\vec{s})\}$, so $\bar{r}(s_j',\vec{s}_{-j})>_{lex}\bar{r}(\vec{s})$.

Since $\mathcal{S}$ is a finite set, there exists $\vec{s}^*\in\mathcal{S}$ such that $\bar{r}(\vec{s}^*)$ is maximal, i.e., there does not exist any $\vec{s}'\in\mathcal{S}$ such that $\bar{r}(\vec{s}')>_{lex}\bar{r}(\vec{s}^*)$. It follows that $\vec{s}^*$ is a PNE,  otherwise there is some $j\in[m]$ and $s_j'\in\mathcal{S}$ such that $|s_j'\cap s_j|=1$ and $V_j(s_j',\vec{s}_{-j})>V_j(\vec{s})$, implying that $\bar{r}(s_j',\vec{s}_{-j})>_{lex}\bar{r}(\vec{s})$, which is a  contradiction.

In coclusion, let the designers use the strategies corresponding to the first-stage equilibrium represented by $\vec{s}^*\in\mathcal{S}$, i.e., $S_{C_j}=\{i_1,i_2\}$ supposing $s_j^*=\{e_{i_1},e_{i_2}\}$ and $R_{C_j}=B_j$, and set balancing biases in the second-stage It is not hard to see that it is a weak designer equilibrium. 
\end{proof}
\section{Missing Proofs in \Cref{sec:divisible}}

\subsection{Proof of \Cref{thm:example-divisible-1/2notbest}}
\DIVExamplebanlanceNotBest*

\begin{proof}
    
    Consider the following instance:
    There are four contestants $\{1,2,3,4\}$, with total effort $\vec{T}=(1.001\times10^{-3}, 1.001\times10^{-3}, 1.001\times10^{6}, 1.001\times10^{6})$. There are two designers $\{1,2\}$ each holding two contests, with $\mathcal{C}_1=\{C_1,C_2\},\mathcal{C}_2=\{C_3,C_4\}$. The participants are $S_{C_1}=\{1,2\},S_{C_2}=\{3,4\},S_{C_3}=\{1,3\},S_{C_4}=\{2,4\}$. The rewards are $R_{C_1}=1,R_{C_2}=10^6,R_{C_3}=R_{C_4}=10^3$. When the biases are $\alpha_{C_1,1}=\alpha_{C_1,2}=1,\alpha_{C_2,3}=\alpha_{C_2,4}=1,\alpha_{C_3,1}=10^6,\alpha_{C_3,3}=1,\alpha_{C_4,2}=10^6,\alpha_{C_4,4}=1$, one can check that the EMV of contestants is $\vec{\lambda}=(2.5*10^5,2.5*10^5,0.25,0.25)$, and moreover, $\hat{p}_{i,C_k}(\vec{\lambda})=\frac12$ holds for all $k=1,2,3,4$ and $i\in S_{C_k}$. The unique contestant equilibrium $\vec{x}$ is given by
    $x_{1,C_1}=x_{2,C_1}=10^{-6},
x_{3,C_4}=x_{4,C_4}=10^6,
    x_{1,C_3}=10^{-3},x_{3,C_3}= 10^3,
    x_{2,C_4}=10^{-3},x_{4,C_4}= 10^3$.
    The utility of designer 1 is $x_{1,C_1}+x_{2,C_1}+x_{3,C_4}+x_{4,C_4}=2000000.000002$.
    
    However, suppose designer 1 changes the bias of $C_1$ to $\alpha_{C_1,1}=2,\alpha_{C_1,2}=1$, and we use \Cref{alg:Alg1} to compute the contestant equilibrium under the new strategy profile of designers. Up to small enough errors, the new EMV is $\vec{\lambda}'=[249972.24920282728, 249972.24920282728,\\ 0.24999999999923062, 0.24999999999923062]$, and the new contestant equilibrium $\vec{x}'$ is given by \\ $(x'_{1,C_1}= 8.889875693437923e-07,\\ x'_{2,C_1}= 8.889875693437923e-07,\\ x'_{3,C_2}= 1000000.0000030776, \\ x'_{4,C_2}= 1000000.0000030776, \\ x'_{1,C_3}= 0.001000111012430656, \\ x'_{3,C_3}= 999.9999969223088, \\ x'_{2,C_4}= 0.001000111012430656,\\ x'_{4,C_4}= 999.9999969223088)$. \par The utility of designer 1 becomes $2000000.00000793>2000000.000002$. Therefore designer 1 has the incentive to change the biases of her contests from the original strategy profile.
\end{proof}

\subsection{Proof of \Cref{thm:div-proportional-allocation-is-SPE}}
\DIVProportionalAllocationisSPE*

\begin{proof}
Note that for any contestant $i$, we have $$T_i=\sum_{C\in\mathcal{A}(i,\mathcal{C}_j)}\frac{R_CQ_C(\vec{\lambda})}{\lambda_i}=\frac1{\lambda_i}\sum_{j\in[m]}B_j\frac{T_i}{2\sum_{k\in[n]}T_{k}},$$ which means that  $\lambda_i=\frac{\sum_{j\in[m]}B_j}{2\sum_{k\in[n]}T_{k}}$, i.e., all contestants have an equal multipliers in $\vec{\lambda}$. Consequently, for any contest $C\in\vec{\mathcal{C}}$, suppose $S_C=\{i_1,i_2\}$. We have $\alpha_{C,i_1}=\alpha_{C,i_2}$ because $\frac{\alpha_{C,i_1}\lambda_{i_2}}{\alpha_{C,i_1}\lambda_{i_2}+\alpha_{C,i_2}\lambda_{i_1}}=\frac12$.

We prove this theorem by contradiction. Suppose that $\vec{\mathcal{C}}$ is not a designer equilibrium, i.e., there exists some contestant who has an incentive to deviate from $\vec{\mathcal{C}}$. Without loss of generality, we can assume that designer $1$ has a beneficial deviation strategy $\mathcal{C}_1'$. Let $\vec{\mathcal{C}}'=(\mathcal{C}_1',\vec{\mathcal{C}}_{-1})$ denote the new strategy profile of designers and  $\vec{\lambda}'$ denote the equilibrium multiplier vector under $\vec{\mathcal{C}}'$. There exists a contestant equilibrium  $\vec{x}'$ under $\vec{\mathcal{C}}'$ such that $\sum_{C\in\mathcal{C}_1'}\sum_{i\in S_C}x'_{i,C}>\sum_{C\in\mathcal{C}_1}\sum_{i\in S_C}\hat{x}_{i,C}(\vec{\lambda})$.

Firstly, we prove the following claims:
\begin{description}
    \item[Claim 1] For any contest $C\in\vec{\mathcal{C}}_{-1}$, it holds that $\sum_{i\in S_C}\hat{x}_{i,C}(\vec{\lambda}')=\frac{R_C}{\sum_{i\in S_C}\lambda_i'}$, and $\sum_{i\in S_C}\hat{x}_{i,C}(\vec{\lambda})=\frac{R_C}{\sum_{j\in[m]}B_j}\sum_{k\in[n]}T_{k}$.
    \item[Claim 2] $\sum_{C\in\vec{\mathcal{C}}_{-1}}\frac{R_C}{\sum_{i\in S_C}\lambda_i'}<\frac{\sum_{j\in[m]\setminus\{1\}}B_j}{\sum_{j\in[m]}B_j}\sum_{k\in[n]}T_{k}$.
    \item[Claim 3] For any contestant $i$, it holds that  $$\sum_{C\in\mathcal{A}(\mathcal{C}_1',i)}R_{C}Q_{C}(\vec{\lambda}')\geq \lambda_i'T_i-\frac12\sum_{j\in[m]\setminus\{1\}}B_j\frac{T_i}{\sum_{k\in[n]}T_k}.$$
\end{description}
To prove Claim 1, for any contest $C\in\vec{\mathcal{C}}_{-1}$, suppose $S_C=\{i_1,i_2\}$. Observe that $\hat{x}_{i_1,C}(\vec{\lambda}')+\hat{x}_{i_2,C}(\vec{\lambda}')=R_C\frac{\lambda_{i_2}'}{(\lambda_{i_1}'+\lambda_{i_2}')^2}+R_C\frac{\lambda_{i_1}'}{(\lambda_{i_1}'+\lambda_{i_2}')^2}=\frac{R_C}{\lambda_{i_1}'+\lambda_{i_2}'}$. Therefore, we have $\sum_{i\in S_C}\hat{x}_{i,C}(\vec{\lambda}')=\frac{R_C}{\sum_{i\in S_C}\lambda_i'}$. Similarly, it holds that 
\begin{align*}
&\sum_{i\in S_C}\hat{x}_{i,C}(\vec{\lambda})=\frac{R_C}{\sum_{i\in S_C}\lambda_i}\\
=&\frac{R_C}{\frac{\sum_{j\in[m]}B_j}{2\sum_{k\in[n]}T_{k}}+\frac{\sum_{j\in[m]}B_j}{2\sum_{k\in[n]}T_{k}}}=\frac{R_C}{\sum_{j\in[m]}B_j}\sum_{k\in[n]}T_{k}.
\end{align*}

To show Claim 2, by Claim 1 we can know that  $\sum_{C\in\vec{\mathcal{C}}_{-1}}\sum_{i\in S_C}\hat{x}_{i,C}(\vec{\lambda}')=\sum_{C\in\vec{\mathcal{C}}_{-1}}\frac{R_C}{\sum_{i\in S_C}\lambda_i'}$, and $\sum_{C\in\vec{\mathcal{C}}_{-1}}\sum_{i\in S_C}\hat{x}_{i,C}(\vec{\lambda})=\frac{\sum_{C\in\vec{\mathcal{C}}_{-1}}R_C}{\sum_{j\in[m]}B_j}\sum_{k\in[n]}T_{k}=\frac{\sum_{j\in[m]\setminus\{1\}}B_j}{\sum_{j\in[m]}B_j}\sum_{k\in[n]}T_{k}$.

From the assumption that $\sum_{C\in\mathcal{C}_1'}\sum_{i\in S_C}x'_{i,C}>\sum_{C\in\mathcal{C}_1}\sum_{i\in S_C}\hat{x}_{i,C}(\vec{\lambda})$, we have
\begin{align*}
&\sum_{C\in\vec{\mathcal{C}}_{-1}}\sum_{i\in S_C}\hat{x}_{i,C}(\vec{\lambda}')\\
\leq &\sum_{C\in\vec{\mathcal{C}}_{-1}}\sum_{i\in S_C}x'_{i,C}\\
\leq&\sum_{i\in[n]}T_i-\sum_{C\in\mathcal{C}_1'}\sum_{i\in S_C}x'_{i,C}\\
<&\sum_{i\in[n]}T_i-\sum_{C\in\mathcal{C}_1}\sum_{i\in S_C}\hat{x}_{i,C}(\vec{\lambda})\\
=&\sum_{C\in\vec{\mathcal{C}}_{-1}}\sum_{i\in S_C}\hat{x}_{i,C}(\vec{\lambda}).
\end{align*}

Combining these, we obtain $$\sum_{C\in\vec{\mathcal{C}}_{-1}}\frac{R_C}{\sum_{i\in S_C}\lambda_i'}<\frac{\sum_{j\in[m]\setminus\{1\}}B_j}{\sum_{j\in[m]}B_j}\sum_{k\in[n]}T_{k}.$$

For Claim 3, for any contestant $i$, we discuss in two cases:

(a) If $\lambda_i'>0$, it holds that $T_i=\sum_{C\in\mathcal{A}(\mathcal{C}_1',i)}\frac{R_CQ_C(\vec{\lambda}')}{\lambda_i'}+\sum_{C\in\mathcal{A}(\vec{\mathcal{C}}_{-1},i)}\frac{R_CQ_C(\vec{\lambda}')}{\lambda_i'}
$. Therefore, we have
\begin{align*}
\sum_{C\in\mathcal{A}(\mathcal{C}_1',i)}R_CQ_C(\vec{\lambda}')=&\lambda_i'T_i-\sum_{C\in\mathcal{A}(\vec{\mathcal{C}}_{-1},i)}R_CQ_C(\vec{\lambda}')\\
\geq& \lambda_i'T_i-\sum_{C\in\mathcal{A}(\vec{\mathcal{C}}_{-1},i)}\frac{R_C}{4}\\
=&\lambda_i'T_i-\frac12\sum_{j\in[m]\setminus\{1\}}B_j\frac{T_i}{\sum_{k\in[n]}T_k}.
\end{align*}

(b) If $\lambda_i'=0$, we can calculate $\sum_{C\in\mathcal{A}(\mathcal{C}_1',i)}R_CQ_C(\vec{\lambda}')\geq 0\geq\lambda_i'T_i-\frac12\sum_{j\in[m]\setminus\{1\}}B_j\frac{T_i}{\sum_{k\in[n]}T_k}$.

In summary, Claim 3 holds for any contestant $i\in[n]$.

For any $C\in\vec{\mathcal{C}}_{-1}$, define $\gamma_C=\sum_{i\in S_C}\lambda_i'$. By Claim 2, we have
\begin{align*}
\sum_{C\in\vec{\mathcal{C}}_{-1}}\frac{R_C}{\gamma_C}<\frac{\sum_{j\in[m]\setminus\{1\}}B_j}{\sum_{j\in[m]}B_j}\sum_{k\in[n]}T_{k}.
\end{align*}

Also, observe that
\begin{align*}
\sum_{C\in\vec{\mathcal{C}}_{-1}}R_C\gamma_C
=&\sum_{C\in\vec{\mathcal{C}}_{-1}}R_C\sum_{i\in S_C}\lambda_i'\nonumber\\
=&\sum_{i\in[n]}\lambda_i'\sum_{C\in\mathcal{A}(\vec{\mathcal{C}}_{-1},i)}R_C\nonumber\\
=&\sum_{i\in[n]}\lambda_i'\sum_{j\in[m]\setminus\{1\}}2B_j\frac{T_i}{\sum_{k\in[n]}T_k}\nonumber\\
=&\frac{2\sum_{j\in[m]\setminus\{1\}}B_j}{\sum_{k\in[n]}T_k}\sum_{i\in[n]}\lambda_i'T_i.
\end{align*}

By Cauchy's inequality, we can obtain
\begin{align*}
&(\sum_{C\in\vec{\mathcal{C}}_{-1}}\frac{R_C}{\gamma_C})(\sum_{C\in\vec{\mathcal{C}}_{-1}}R_C\gamma_C)\\
=&(\sum_{C\in\vec{\mathcal{C}}_{-1}}\sqrt{\frac{R_C}{\gamma_C}}^2)(\sum_{C\in\vec{\mathcal{C}}_{-1}}\sqrt{R_C\gamma_C}^2)\\
\geq&(\sum_{C\in\vec{\mathcal{C}}_{-1}}\sqrt{\frac{R_C}{\gamma_C}}\sqrt{R_C\gamma_C})^2\\
=&(\sum_{C\in\vec{\mathcal{C}}_{-1}}R_C)^2.
\end{align*}

Therefore, we have 
\begin{align*}
&\sum_{C\in\vec{\mathcal{C}}_{-1}}R_C\gamma_C\geq\frac{(\sum_{C\in\vec{\mathcal{C}}_{-1}}R_C)^2}{\sum_{C\in\vec{\mathcal{C}}_{-1}}\frac{R_C}{\gamma_C}}\\
>&\frac{(\sum_{j\in[m]\setminus\{1\}}B_j)^2}{\frac{\sum_{j\in[m]\setminus\{1\}}B_j}{\sum_{j\in[m]}B_j}\sum_{k\in[n]}T_{k}}=\frac{(\sum_{j\in[m]\setminus\{1\}}B_j)(\sum_{j\in[m]}B_j)}{\sum_{k\in[n]}T_{k}}.
\end{align*}

It follows that $$\sum_{i\in[n]}\lambda_i'T_i=\frac{\sum_{k\in[n]}T_k}{2\sum_{j\in[m]\setminus\{1\}}B_j}\sum_{C\in\vec{\mathcal{C}}_{-1}}R_C\gamma_C>\frac{\sum_{j\in[m]}B_j}{2}.$$

Finally, summing over the inequalities in Claim 3 for all $i\in[n]$, we obtain
\begin{align*}
    \sum_{i\in[n]}\sum_{C\in\mathcal{A}(\mathcal{C}_1',i)}R_CQ_C(\vec{\lambda}')\geq& \sum_{i\in[n]}\lambda_i'T_i-\frac12\sum_{j\in[m]\setminus\{1\}}B_j\\
    >&\frac12B_1.
\end{align*}
However, this contradicts with the fact that $$\sum_{i\in[n]}\sum_{C\in\mathcal{A}(\mathcal{C}_1',i)}R_CQ_C(\vec{\lambda}')\leq 2\sum_{C\in\mathcal{C}_1'}\frac{R_C}4\leq \frac{B_1}{2}.$$ Therefore, $\vec{\mathcal{C}}$ is a designer equilibrium.
\end{proof}

\subsection{Proof of \Cref{thm:divisible-SPE-existence}}
\DIVthmSPEExists*

\begin{proof}
We only need to construct a strategy profile of designers $\vec{\mathcal{C}}$ which satisfies the condition in \Cref{thm:indiv-1/2probability-is-equilibrium}. Actually we only need to construct a matrix $A=(A_{i,k})\in\mathbb{R}_{\geq 0}^{n\times n}$, such that it holds for any contestant $i$ that $\sum_{k\in[n]}(A_{i,k}+A_{k,i})=T_i$, and that for all $i\geq k$, $A_{i,k}=0$. Then, any designer $j$ can construct a contest $C$ for each pair of $i,k$ such that $S_C=\{i,k\}$, $R_C=B_j\frac{A_{i,k}}{\sum_{i'\in[n]}T_{i'}}$, and $\alpha_{C,i}=\alpha_{C,k}=1$. This will form a strategy profile $\vec{\mathcal{C}}$. We can easily verify that $\vec{\mathcal{C}}$ satisfies the condition in \Cref{thm:indiv-1/2probability-is-equilibrium} and is a designer equilibrium.

We only need to give a construction of such matrix $A$.
We can view each contestant $i\in[n]$ as an interval of length $T_i$ on the axis, and put the $n$ intervals together so that the $k$-th interval is placed at $I_k=[\sum_{i<k}T_i,\sum_{i<=k}T_i]$. Let $M=\frac{1}{2}\sum_{i=1}^nT_i$, which is the middle point of $[0,\sum_{i=1}^nT_i]$. For each point $x\in[0,M]$, we match the point $x$ with the point $x+M$. Finally, let $\mu(S)$ denote the measure (i.e., the length of an interval) of point set $S$, and we set $A_{i,k}=\mu(\{x\in[0,M]:x\in I_i\land x+M\in I_k\})$.
For all $i\in[n]$, we have $\sum_{k\in[n]}(A_{i,k}+A_{k,i})=\mu(\{x\in[0,M]:x\in I_i\})+\mu(\{x\in[0,M]:x+M\in I_i\})=\mu(I_i)=T_i$. For any $k<i$, we have $\{x\in[0,M]:x\in I_i\land x+M\in I_k\}=\emptyset$. Additionally, since $T_i\leq \frac{1}{2}\sum_{i=1}^nT_i M$, we have $\mu(\{x\in[0,M]:x\in I_i\land x+M\in I_i\})=0$. Therefore for any $i,k\in[n]$ that $i\geq k$, it holds that $A_{i,k}=0$. 
This completes the proof.


\end{proof}

\end{document}